\documentclass[12pt]{article}
\usepackage{caption}
\usepackage{a4wide}
\usepackage{amsmath,amsfonts,
amssymb, amsthm,dsfont,xcolor, MnSymbol, relsize} 
\usepackage{hyperref}
\usepackage{csquotes}
\usepackage{soul}
\usepackage{cancel}
\usepackage{tikz}
\usepackage{float}
\usetikzlibrary{arrows}
\usepackage{capt-of}

\usepackage[active]{srcltx}
\usepackage{dutchcal}

\usepackage{mathrsfs}
\usepackage{amscd,graphicx,mdframed}
\usepackage[shortlabels]{enumitem}
\usepackage[nottoc]{tocbibind}
\usepackage[toc,page]{appendix}

\newtheorem{theorem}{Theorem}[section]
\newtheorem{definition}[theorem]{Definition}
\newtheorem{proposition}[theorem]{Proposition}
\newtheorem{corollary}[theorem]{Corollary}
\newtheorem{lemma}[theorem]{Lemma}

\theoremstyle{definition}
\newtheorem{remark}[theorem]{Remark}

\newtheorem{notation}[theorem]{Notation}

\newtheorem{hypothesis}[theorem]{Hypothesis}

\def\x{\mathbf{x}}
\def\e{{\cal{e}}}

\def\R{\mathbb{R}}

\def\Co{{\mathbb C}} 
\def\H{{\cal H}} 
\def\B{\mathfrak{B}}

\def\cc{\mathcal{c}}
\def\ham{\mathfrak{H}}

\def\Op{\mathfrak{Op}} 

\def\X{\mathcal X}
\def\U{\mathcal{U}}

\def\E{\mathcal{E}}
\def\Ker{\mathbb{K}{\rm er}\,}
\def\Rge{\mathbb{R}{\rm ge}\,}
\def\s{\mathfrak{s}}
\def\p{\mathfrak{p}}

\def\F0{\mathlarger{\mathlarger{\mathbf{\Lambda}}}}
\def\Fb{\mathlarger{\mathlarger{\mathbf{\Lambda}}}_{\text{\tt bd}}}
\def\Fp{\mathlarger{\mathlarger{\mathbf{\Lambda}}}_{\text{\tt pol}}}

\def\z{\mathfrak{z}}

\def\Rd{\mathbb{R}^d}
\def\Zd{\mathbb{Z}^d}
\def\Sd{\mathbb{S}^d}
\def\Z{\mathbb{Z}}
\def\T{\mathbb{T}}
\def\bb1{{\rm{1}\hspace{-3pt}\mathbf{l}}}

\def\Id{\text{I}\hspace*{-2pt}\text{I{\sf d}}}
\def\dist{{\rm dist}}

\def\Int{\mathfrak{I}\mathit{nt}}
\def\supp{\mathop{\rm supp} \nolimits} % Support

\def\repi{\rightY\hspace*{-4pt}\rightarrow}

\def\BB{\mathbf{B}}

\def\UL1{\mathcal{U}\hspace*{-3pt}L^1}

\def\bz{\boldsymbol{\z}}

\def\lnu{\mathlarger{\nu}}

\def\MmN{\mathscr{M}_{n_\B}}

\def\zz{\mathcal{z}}

\def\tp{\tilde{\mathfrak{p}}}
\def\blPsi{\text{\large{$\boldsymbol{\Psi}$}}}
\def\ec{^{\epsilon,\cc}}

\def\tpsi{\widetilde{\psi}}

\def\mU{\mathfrak{U}}
\def\mW{\mathfrak{W}}
\def\mm{\mathfrak{m}}
\def\mM{\mathfrak{M}}
\def\mS{\mathfrak{S}}
\def\tLambda{\widetilde{\Lambda}}
\def\tOmega{\widetilde{\Omega}}
\def\tOp{\widetilde{\Op}}
\def\nB{n_{\B}}
\def\Nb{\mathbb{N}_\bullet}

\def\spr{{s^\prime}}

\def\Lsim{\text{\LARGE $\sim$}}

\def\beq{\begin{equation}}
	\def\eeq{\end{equation}}
    
	\def\Nb{\mathbb{N}_\bullet}

\numberwithin{equation}{section}

\definecolor{RawSienna}{cmyk}{0,0.72,1,0.45}

\newcommand{\clr}{\color{red}}
\newcommand{\clb} {\color{blue}}

\newcommand{\clbr} {\color{RawSienna}}

\title{A fresh look at the Peierls-Onsager substitution.}
\author{Horia D. Cornean\footnote{Department of Mathematical Sciences, Aalborg University, Thomas Manns Vej 23, 9220 Aalborg, Denmark; cornean@math.aau.dk}, Bernard Helffer\footnote{Laboratoire de Math{\'e}matiques Jean Leray,  Nantes Universit{\'e}  and CNRS, Nantes, France;
Bernard.Helffer@univ-nantes.fr}, Radu Purice\footnote{\enquote{Simion Stoilow} Institute of Mathematics of the Romanian Academy, P.O. Box 1-764, 014700 Bucharest, Romania; Radu.Purice@imar.ro}}

\begin{document}
	\maketitle
\thispagestyle{empty}

\setcounter{page}{1}

\begin{abstract}
We formulate a general version of the Peierls-Onsager substitution for a finite family of Bloch eigenvalues under a local spectral gap hypothesis, via strongly localized tight-frames and magnetic matrices. This extends the existing results to long-range magnetic fields  without any slow-variation hypothesis and without any triviality assumption for the associated Bloch sub-bundle. Moreover, our results cover a large class of periodic, elliptic pseudo-differential operators. We also prove the existence of an approximate time evolution for initial states supported inside the range of the isolated Bloch family, with a precise error control.
\end{abstract}

\tableofcontents

\section{Introduction}

The phenomenology of solids in the presence of external electro-magnetic fields is an extremely rich field for both theoretical and applied research. The basic model for these phenomena considers charged quantum particles moving in a potential that is periodic with respect to a regular lattice $\Gamma\cong\Zd$, describing the crystalline structure of the solid, under the action of some external electric and magnetic fields. An important feature implied by the periodicity of the  potential is \textit{the Bloch-Floquet decomposition} of the self-adjoint operator $H_0$ that describes a quantum particle moving in such a potential in a $d$-dimensional affine space (see \cite{CHP-5} and Appendix \ref{A-BF-Theory}). This means that $H_0$ may be decomposed as a smooth family of self-adjoint operators $\{\widehat{H}(\xi)\}_{\xi\in\T_*}$ (which we call fibre operators) acting on square integrable functions on the $d$-dimensional torus ($ L^2(\T)$ with $\T:=\Rd/\Zd$) and indexed by $\xi\in\T_*$, \textit{the dual torus} (see \eqref{D-torus} and the following text in Subsection  \ref{ss1.3}). Moreover, as described in \cite{CHP-5}, each fibre operator on $L^2(\T)$ has a compact resolvent and thus an unbounded discrete spectrum $\big\{\lambda_k(\xi)\,,\,k\in\mathbb{N}\setminus\{0\},\,\xi\in\T_*\big\}$. The functions $\lambda_k$, called the \textit{Bloch levels}, are continuous functions on the dual torus whose graphs might cross among themselves. 
We identify functions on a $d$-dimensional torus with periodic functions on $\Rd$.

\subsection{The Peierls-Onsager approximation procedure}
\label{SS-POapp}

In many situations there exists a group of Bloch levels which, although {possibly} crossing among themselves, stay well separated from the remaining infinite family of Bloch levels. Using the terminology introduced in \cite{FT}, we say that we have in this case an \textit{isolated Bloch family} and a \textit{local spectral gap condition}. We call it  a \textit{strictly isolated Bloch family} when {its range} is separated from the rest of the spectrum by some spectral gaps of the full operator $H_0$; in this case we say that the {\it global spectral gap condition} holds. In general, when such an isolated Bloch family $\B$ exists, one can define an infinite dimensional orthogonal projection $P_{\B}$, invariant under translations with vectors from $\Gamma$ and commuting with the Hamiltonian $H_0$ such that under certain topological conditions on the structure of $P_\B$, the dynamics of the states in $P_{\B}L^2(\Rd)$ is exactly given by the 'quantization' of a matrix-valued function $\widehat{\mm}_{\B}\big(\widehat{H}(\theta)\big)$ on the dual torus (see \eqref{F-PN-3}) considered as a periodic function on the dual space $\widehat{\mm}_\B(\xi)$. If the Bloch family is strictly isolated, then the projection $P_{\B}$ is a spectral projection of the operator $H_0$. 

Let us mention here that some interesting physical consequences may be deduced when the so-called Fermi level of the given material is situated in the domain of such an isolated Bloch family, but we shall not consider these aspects in the present paper.

If the external magnetic field  $B=dA$ is weak, one of the widely used approximations in the physics literature is the celebrated Peierls-Onsager substitution. It states that an isolated Bloch family will continue to manifest itself in some given \textit{spectral windows} of the operator $H^A$ (perturbation of $H_0$ by the magnetic field), through the existence of an 'approximatly' invariant subspace on which the dynamics may be approximated by "the minimal coupling procedure" applied to the matrix-valued function $\widehat{\mm}_{\B}(\xi)$, i.e. the quantization of the matrix-valued function $\widehat{\mm}_{\B}\big(\xi-A(x)\big)$. Let us notice that that when the isolated Bloch family has only one element $\lambda_{k_0}(\xi)$ and the topological condition is satisfied, then the matrix $\widehat{\mm_{\B}}$ has dimension 1 with the unique element $\lambda_{k_0}(\xi)$, so that the approximate dynamics is given by the quantization of the periodic function $\lambda_{k_0}\big(\xi-A(x)\big)$.

It turns out that  
transforming the above heuristic arguments into mathematical theorems is a rather challenging problem, the mathematics behind this approximating procedure being rather involved and thus a very large amount of literature has been devoted to it (let us cite just a few papers more closely related to our work \cite{Be1, BC, CIP, dNL,GMSj, FT, HS, IP, HS1, Ne-LMP,Ne-RMP, Pa, PST,Sj}).  We shall very briefly recall the main results of the existing literature in order to put into perspective our new results. 

\subsection{Very brief report on the existing results}

We emphasize from the begining that the quasi-totality of the existing references on the subject considers the case with 
\beq\label{F-f-1}
H_0=-(2m)^{-1}\Delta+\mu{W}(x)\,,
\eeq 
with $\Delta$ the Laplace operator on $\Rd$, $W$ some real potential function, $m>0$ the mass of the particles and $\mu\in\R$ some coupling constant, with the magnetic perturbation given by: 
\beq\label{F-f-2}
H^A=-(2m)^{-1}\Delta^A+\mu{W}(x)\,,
\eeq 
with \textit{the magnetic Laplacian}:
\beq\label{F-f-3}
\Delta^A\,:=\,\underset{1\leq j\leq d}{\sum}\big(\partial_j-iA_j(x)\big)^2.
\eeq

The main technical difficulty in dealing with magnetic perturbations comes when we consider magnetic fields that do not vanish fast enough towards infinity, having vector potentials whose derivatives might grow at infinity. In this case, the functions $\lambda\big(\xi-A(x)\big)$ might no longer be H\"ormander symbols. 

A second, more subtle technical difficulty has a topological nature and is related with the vector bundle structure of the Bloch-Floquet decomposition that we briefly review in Appendix \ref{A-BF-Theory}. More precisely, the Bloch-Floquet decomposition of the orthogonal projection associated with the isolated Bloch family (that we shall precisely define in Subsection \ref{SSS-iBb}) may allow for a smooth family of orthonormal basis, or not. In physical terms this smooth basis are known as localized composite Wannier functions (see also our Appendix \ref{A-BF-Theory}).

The first approaches to obtain a rigorous Peierls-Onsager approximation theory (see \cite{Be1,HS, HS1,Sj,GMSj}) have considered \textit{isolated Bloch families with only one Bloch level and constant magnetic fields} under the following two main assumptions:
\begin{hypothesis}\label{H-I}~
\begin{enumerate}
	\item \textit{the global gap condition};
	\item \textit{existence of a system of localized composite Wannier functions}.
\end{enumerate}
\end{hypothesis}
The constancy of the magnetic field has a very important technical consequence, namely the existence of \textit{a projective representation of the translations with elements from} $\Gamma$ that commutes with the perturbed operator. They are the Zak translations, that we recall in Paragraph \ref{SS-Z-trsl}. Extending these type of arguments to non-constant magnetic fields has unfortunately not been possible. The effort to relax the two assumptions in Hypothesis~\ref{H-I} and the constancy hypothesis on the magnetic field has generated a large number of improvements of these first results.

Starting with the paper \cite{PST} a new point of view has become dominant, namely to replace constancy of the magnetic field with \textit{a slow variation hypothesis}, i.e. to consider regular magnetic fields (i.e. with components of class $BC^\infty(\Rd)$) controlled by a small parameter $\epsilon\in[0,\epsilon_0]$ in the form: $B_\epsilon:=dA_\epsilon$ with $A_\epsilon(x):=A(\epsilon{x})$ for some $A$ with smooth components and to use a variant of adiabatic procedure called \textit{space-adiabatic perturbation theory} (SAPT) \cite{PST-0,T03}. This procedure allowed for the elimination of the global gap hypothesis for an isolated Bloch family with any finite number of Bloch levels. Nevertheless, let us clearly emphasize here that the slow variation  hypothesis is a  strong limitation, defining a rather small and very special subclass of long-range regular magnetic fields: namely those fields having weak components of some order $\epsilon>0$ with $k$-derivatives of order $\epsilon^k$ for any $k\in\mathbb{N}$; thus, on any finite region, for $\epsilon>0$ small, the magnetic field has a very simple structure.

The second assumption in Hypothesis \ref{H-I} remained necessary for practically all the following developments, except for the results in \cite{FT} where working with some non-trivial connections on the Bloch bundle allows the authors to treat some special situations when the localized composite Wannier functions do not exist. Nevertheless, one has to impose that the magnetic field is a slowly varying perturbation of a constant magnetic field, the non-constant part coming from a bounded magnetic potential. 

Finally, let us mention the paper \cite{CIP} where the authors propose the use of the magnetic pseudo-differential calculus (see \cite{MP-1},\cite{IMP-1} - \cite{IMP-3}) in order to obtain a natural and efficient framework to deal with the Peierls-Onsager approximation and under Hypothesis \ref{H-I} obtain an extension of the rigorous Peierls-Onsager approximation for strictly isolated Bloch families of a large class of magnetic pseudo-differential operators with elliptic symbols of strictly positive order and regular magnetic fields (i.e. with smooth components, bounded together with all their derivatives).

Let us also recall the approach using a kind of Grushin problem (\cite{GMSj, IP}) that allowed for important extensions concerning the location of the perturbed spectrum close to the isolated Bloch family, without any information related to the approximate dynamics.

Thus, to our knowledge, in the current state of the art, the Peierls-Onsager substitution without a global gap considiton has been rigorously defined and controlled only for slowly varying long range perturbing magnetic fields and, left aside some very special cases (treated in \cite{FT}), only in the presence of localized composite Wannier functions.

\subsection{Significance of our results}\label{SS-f-1}

In this paper we use an abstract general theorem concerning finite rank vector bundles over finite dimensional compact manifolds in order to construct a Parseval frame (see Appendix~\ref{A-frames}) associated with  any finite isolated Bloch family, and use this frame and the magnetic pseudo-differential calculus in order to define and control a Peierls-Onsager approximation for regular long-range perturbing magnetic fields, without the  global gap condition and without the hypothesis of existence of localized composite Wannier functions. Let us clearly state here what we mean by regular magnetic field;
\begin{definition}
We say that a magnetic field on a $d$-dimensional configuration space is \emph{regular}, when its components are smooth and bounded together with their derivatives of all orders.
\end{definition}
Our arguments apply to a large class of pseudo-differential operators with elliptic periodic symbols of strictly positive order (see Definition \ref{D-Hcirc}), the case in \eqref{F-f-1} being viewed as associated with  the elliptic symbol of second order: $(x,\xi)\mapsto (2m)^{-1}|\xi|^2+\mu{W}(x)$. In fact, as also stated in Definition \ref{D-Hcirc}, our quantum Hamiltonians may also incorporate a $\Gamma$-periodic regular magnetic field satisfying Hypothesis \ref{H-BGamma} and thus having a $\Gamma$-periodic vector potential.

Finally, let us mention that we can make a more refined analysis for the case when the magnetic field perturbation has a non-zero constant part, i.e. with the following structure:
$$ B^{\epsilon,\cc}(x)= \epsilon \, \big (B_0 +\cc B_1(x)\big ),$$ 
where $\cc\in[0,1)$, $B_0$ has constant components and $B_1(x)$ is a long-range (non vanishing at infinity), regular magnetic field, \textit{that is not supposed to be either periodic or slowly variable}. To this magnetic field one can always construct (see Appendix \ref{A-m-PsiDO}) a magnetic vector potential $A^{\epsilon,\cc}(x)$ which is an $1$-form whose coefficients are smooth and grow at most linearly at infinity. This linear growth at infinity and the lack of slow variation are the main reasons for which the spatial adiabatic methods do not work here. 

Moreover, an argument based on a Schur-Feshbach decomposition allows to compare the real dynamics and the approximate one with an error of order $\epsilon^2$ for the spectrum, and of order $\epsilon$ for the time evolution of states with energies in an associated spectral window.
		
Compared with the results of \cite{FT}, our approach allows non-constant magnetic field perturbations which do not come from bounded and slowly variable magnetic vector potentials and our effective operator is identified with a magnetic matrix instead of a pseudodifferential operator on a twisted two dimensional torus, and we may work in any dimension $d\geq 2$.

Regarding the construction of our Parseval frame, compared with \cite{CMM}, what we do here is valid in any dimension $d\geq 2$. Also, the way in which we construct the magnetic frame from the non-magnetic one is completely different and it may be of independent interest. Finally, let us emphasize the elaboration in Subsection  \ref{SS-psGS} of an abstract procedure replacing the Gram-Schmidt method and allowing to transform a frame into a Parseval frame.
	
\subsection{Brief overview of the paper}

The notations and conventions we use, as well as a number of basic facts from the literature are gathered in Subsection \ref{ss1.3}.

Let us emphasize from the beginning that the entire formulation of our arguments is done in the framework of the magnetic pseudo-differential calculus (see \cite{IMP-1, IMP-2, IMP-3, MPR1,CHP-3, CHP-4-1}), which to our knowledge is the most natural one allowing for a precise mathematical formulation of the Peierls-Onsager substitution, as quantization of the matrix-valued $\Gamma_*$-periodic function $\widehat{\mm_{\B}}\big(\xi-A(x)\big)$ as mentioned in Subsection \ref{SS-POapp}, for long-range regular magnetic fields.

Our approach to the Peierls-Onsager approximation is rather different from the one using the Space-Adiabatic Perturbation Theory, being a quite natural extension of the first approaches in \cite{Ne-LMP, Ne-RMP, HS}, based on three important technical new aspects:
\begin{itemize}
	\item the Parseval frames that one can build in the spaces of sections of finite rank vector bundles over finite dimensional compact manifolds, allowing to replace the lack of localized smooth Wannier functions,
	\item the magnetic pseudo-differential calculus, allowing to deal with non-constant long-range regular magnetic fields,
	\item a procedure based on the Schur complement used to obtain good spectral results for operators ``reduced" to some ``quasi-invariant" subspaces.
\end{itemize}

In Section \ref{S-mres} we present the main framework of our work, giving a precise formulation of the problem we consider and the main results we obtain.

Section \ref{S-Pfr-ibBf} introduces the main technical construction needed for our strategy starting from a general abstract result in the theory of smooth Euclidean vector bundles and defining a Parseval frame $\blPsi_{\B}$ \eqref{DF-freePfr} for the closed subspace defined by the isolated Bloch family $\B$ in Hypothesis \ref{H-isBf}.

Section \ref{S-3} starts from some ideas in \cite{Z,Lu} of combining $\Gamma$-translations with a local change of gauge, generalizing a procedure used in \cite{Ne-RMP, HS} based on the Zak translations in constant magnetic field (see Definition \ref{D-Ztr-cB}) in order to build a ``perturbed Parseval frame" $\blPsi^\epsilon_{\B}$ in \eqref{DF-Pfr-B-eps}. The closed subspace generated by this frame is the image of the perturbed orthogonal projection $P^\epsilon_{\B}$ announced in point \ref{point1} of  Theorem \ref{T-I}. The different properties of $P^\epsilon_{\B}$ are proved using Remark \ref{R-III-1}, the results in Subsection \ref{SS-B-P-fr},  Proposition~\ref{R-p-symb},  Remark~\ref{R-p-eps-B} and the results of Paragraph ~\ref{SS-comm-H-eps}.

Section \ref{S-4} uses the results in \cite{AMP,CP-1} concerning the continuous variation of the spectrum of the Hamiltonian with the intensity of weak regular magnetic fields and a variant of the Feshbach-Schur argument elaborated by the authors in \cite{CHP-2, CHP-4} in order to obtain a relation between the resolvents of the 
complete and the reduced perturbed Hamiltonians for a spectral parameter belonging to a certain ``window". In this way we prove Theorems~\ref{T-I-1} and \ref{C-T-I}, which leads to the results stated  in point \ref{point2} of Theorem \ref{T-I}.

Concerning the results stated as point \ref{point3} of  Theorem \ref{T-I}, they follow from the conclusions and arguments in Subsections \ref{SS-m-P-frame} - \ref{SS-PN-V-2} using the perturbed Parseval frame $\blPsi^\epsilon_{\B}$ in \eqref{DF-Pfr-B-eps}. The proclaimed $C^*$-algebra injective morphism $\mathfrak{M}^\epsilon_{\B}$ is just the composition of the $C^*$-algebra morphism $\mW^\epsilon_{\B}$ defined in Remark \ref{R-PN-V-2} composed to the right with the natural isomorphism $\mathbb{B}\big(\ell^2(\Gamma)\otimes\Co^{\nB}\big)\overset{\sim}{\longrightarrow}\mathscr{M}^\circ_\Gamma[\mathscr{M}(\nB)]$ induced by the canonical basis $\big\{{{\cal{e}}_\alpha\otimes{\cal{e}}_p}\big\}_{(\alpha,p)\in\Gamma\times\{1,\dots,\nB\}}$ of $\ell^2(\Gamma)\otimes\Co^{\nB}$. {Conclusion}  3.\ref{point3a} in Theorem \ref{T-I} is just the abstract property of injective $C^*$-algebra morphisms, while {Conclusion} 3.\ref{point3b} follows from the results in Subsections \ref{SS-PN-V-2} and \ref{SS-B-P-fr}. 

Section \ref{S-5} is dedicated to the proof of Theorem \ref{T-II} concerning the special case of a regular perturbing magnetic field having a non-zero constant part which allows for a more refined description of the associated effective Hamiltonian.

We have attached three appendices recalling a few basic facts about the Bloch-Floquet theory, the magnetic quantization and the theory of frames in Hilbert spaces. Let us notice that in the second appendix dedicated to the magnetic quantization we have also included some computational details of some arguments from Section \ref{S-3}.

\section{Precise formulation of the problem and the main result}\label{S-mres}
\subsection{Notations}
We denote by $\mathbb{N}_\bullet:=\mathbb{N}\setminus\{0\}$ and for $n\in\Nb$ let  $\underline{n}:=\{1,\ldots,n\}$. We use the convention $\R_+:=[0,+\infty)$. {Given a normed linear space $\mathscr{V}$ with the norm $|\cdot|$ we use the notation $<v>:=\sqrt{1+|v|^2}$ for $v\in\mathscr{V}$ and respectively $v(\epsilon)=\mathscr{O}(\epsilon)$ for a map $[0,\epsilon_0]\ni\epsilon\mapsto\,v(\epsilon)\in\mathscr{V}$, for some $\epsilon_0>0$, such that $v(0)=0$ and $\epsilon^{-1}|v(\epsilon)|\leq C<\infty$ for any $\epsilon\in(0,\epsilon_0]$.}
We denote by $\bb1$ the identity operator on any vector space, that we may sometimes indicate by an index; $\Id$ will denote the identity map on any set, that we may also sometimes indicate by an index. For a subset $M$ of a topological space we denote by $\mathring{M}$ its interior subset (i.e. the largest open subset contained in $M$). Given two topological vector spaces $\mathscr{V}_1$ and $\mathscr{V}_2$ we denote by $\mathcal{L}(\mathscr{V}_1;\mathscr{V}_2)$ the linear space of linear continuous maps $\mathscr{V}_1\rightarrow\mathscr{V}_2$ with the topology of uniform convergence on bounded sets. Given any complex Hilbert space $\big(\mathcal{H},(\cdot,\cdot)_\mathcal{H}\big)$, the scalar product will be considered anti-linear in the first factor; we shall use the notations $\mathbb{B}(\mathcal{H})$ for the bounded linear operators on $\mathcal{H}$, respectively $\mathbb{F}(\mathcal{H})$ for the finite-rank operators, $\mathbb{U}(\mathcal{H})$ for the unitary operators, $\mathbb{L}(\mathcal{H})$ for its orthogonal projections and $\mathbb{P}(\mathcal{H})$ for the family of 1-dimensional orthogonal projections. 

Given a densely defined closable operator $T$ acting in $\mathcal{H}$ we denote by $\overline{T}$ its closure. Given a self-adjoint operator $T$ acting in a Hilbert space $\H$ we denote by $E_J(T)$ the spectral projection of $T$  associated with  the Borel subset $J\subset\R$.

We recall that a Fr\'{e}chet space is a linear space endowed with a topology defined by a countable family of semi-norms $\{\lnu_k\}_{k\in\Nb}$, for which it is complete. A bounded subset of a Fr\'{e}chet space $\mathscr{V}$ is a subset $M\subset\mathscr{V}$ such that for any $k\in\Nb$ there exists $C_k(M)>0$ for which:
$	
\lnu_k(v)\leq\,C_k(M),\ \forall v\in M.
$

We shall freely use the notations introduced in Appendix \ref{A-m-PsiDO} in connection with the magnetic pseudo-differential calculus with various magnetic fileds.

\subsection{The framework}\label{ss1.3}

We shall work in dimension $d\geq2$ and denote by $\X\cong\Rd$ the \textit{'configuration space'}. In order to easily distinguish it from the \textit{'momentum variables'} living in the dual of the configuration space, we shall use the notation $\X\cong\R^d$ for the first one and $\X^*\cong\Rd$ for the second one. 
Denoting by $\big\{\e_1,\ldots,\e_d\big\}$ the canonical orthonormal basis of $\X\cong\Rd$ we are given a discrete, regular lattice $\Gamma:=\underset{1\leq j\leq d}{\bigoplus}\Z\e_j\subset\X$.

Our operators are defined as \textit{magnetic quantizations} using the extension of the Weyl quantization procedure introduced in \cite{MP-1, IMP-1} and very briefly recalled in Appendix \ref{A-m-PsiDO}.

\paragraph{The phase space.}
We will denote by $\Xi:=\X\times\X^*$ the phase space and by $<\cdot,\cdot>:\X^*\times\X\rightarrow\R$ the canonical bilinear duality. We shall use systematically notations of the form $X:=(x,\xi)\in\Xi$, $Y:=(y,\eta)\in\Xi$, and so on. 
We shall also consider the unique real scalar product $\X\times\X\ni(x,y)\mapsto x\cdot y\in\R$ making $\{\e_1,\ldots,\e_d\}$ an orthonormal basis. On the dual $\X^*$ we shall introduce the \textit{dual basis} $\big\{\e^*_j,\,j\in\underline{d}\big\}$ defined by the relations $<\e^*_j\,,\,\e_k>=2\pi\delta_{jk}$ for any pair $(j,k)\in\underline{d}\times\underline{d}$. We denote by $dx$, resp. $d\xi$ the usual Lebesque measures on $\X$, resp. $\X^*$, associated with the above defined orthogonal basis. We recall the natural action by translations of $\R^d$ on functions or distributions defined on $\X$, denoting by $\tau_x$ the translation with $x\in\X$, resp. by translations on functions or distributions defined on $\X^*$, denoting by $\tau_{\xi}$ for translations with $\xi\in\X^*$. 

We have chosen to simplify some formulas by inserting the factor $2\pi$ in the duality map $<\cdot,\cdot>:\X^*\times\X\rightarrow\Co$ and thus for a decomposition $\xi=\underset{1\leq j \leq d}{\sum}\xi_j\e^*_j$ we have the equality $<\xi,x>=2\pi\underset{1\leq j \leq d}{\sum}\xi_j\,x_j$. This leads to some modifications in the formula of the Fourier unitary operators and the fact that the elementary cells in $\Gamma$ and $\Gamma_*$ have both volume 1. 

Using the decomposition $\X\ni{x}\mapsto\iota(x)+\hat{x}\in\Gamma\times[-1/2,1/2)^d$ defined by $\iota(x)_j:=\lfloor x_j\rfloor :=\max\{k\in\mathbb{Z},\ k\leq x_j\}\in\Z$ and the similar one $\X^*\ni\xi\mapsto\iota^*(\xi)+\hat{\xi}\in\Gamma_*\times[-1/2,1/2)^d$ we shall consider \textit{the unit cells:}
\beq\begin{split}
	\E:=\{\hat{x}=\underset{1\leq j\leq d}{\sum}\,\hat{x}_j\,\e_j,\ -1/2\leq\hat{x}_j<1/2\}\subset\X,\\
	\mathcal{B}:=\{\hat{\xi}=\underset{1\leq j\leq d}{\sum}\,\hat{\xi}_j\,\e^*_j,\ -1/2\leq\hat{\xi}_j<1/2\}\subset\X^*.
\end{split}\eeq

We  use the H\"{o}rmander multi-index notation $\partial^\alpha_x:=\partial_{x_1}^{\alpha_1}\cdot\ldots\cdot\partial_{x_d}^{\alpha_d}$ and $|\alpha|:=\alpha_1+\ldots+\alpha_d$ for any $\alpha\in\mathbb{N}^d$. For any $j\in\underline{d}$ we denote by $\varepsilon_j\in\mathbb{N}^d$ the multi-index with components $(\varepsilon_j)_k=\delta_{jk}$.

\paragraph{Some classes of functions.}
Given a finite dimensional real Euclidean space $\mathcal{V}$ with Euclidean norm denoted by $|\cdot|$ let us consider the function spaces $BC^\infty(\mathcal{V})$ of smooth complex functions on $\mathcal{V}$ bounded together with all their derivatives, $C^\infty_c(\mathcal{V})$ of smooth and compactly supported complex functions on $\mathcal{V}$ and $C^\infty_{\text{\tt pol}}(\mathcal{V})$ of smooth complex functions on $\mathcal{V}$, with polynomial growth at infinity together with all their derivatives.
We define the weight functions (having the properties of a norm but allowed to take also the value $+\infty$) on $C^\infty(\mathcal{V})$:
\[
\forall(p,n)\in\R_+\times\mathbb{N},\quad\lnu_{p,n}(\phi):=\underset{y\in\mathcal{V}}{\sup}\,<y>^p\,\underset{|a|\leq n}{\max}\,\big|\big(\partial^a\phi\big)(y)\big|,\ \forall\phi\in C^\infty(\mathcal{V}).
\]
Then $BC^\infty(\mathcal{V})$ is a Fr\'{e}chet space for the familly of norms $\big\{\lnu_{0,n}\big\}_{n\in\mathbb{N}}$ and we also define the space of Schwartz test functions:
\[
\mathscr{S}(\mathcal{V})\,:=\,\big\{\phi\in BC^\infty(\mathcal{V})\,,\,\lnu_{p,n}(\phi)<\infty,\,\forall(p,n)\in\R_+\times\mathbb{N}\big\}.
\]
We denote by $\mathscr{S}^\prime(\mathcal{V})$ the space of tempered distributions, defined as the topological dual of $\mathscr{S}(\mathcal{V})$ on which we can consider either the weak topology with respect to $\mathscr{S}(\mathcal{V})$ or the strong dual topology defined as the topology of uniform convergence on bounded sets of $\mathscr{S}(\mathcal{V})$. We  denote by $\langle\cdot,\cdot\rangle_{\mathcal{V}}:\mathscr{S}^\prime(\mathcal{V})\times\mathscr{S}(\mathcal{V})\rightarrow\mathbb{C}$ the bilinear canonical duality map.
\begin{notation}\label{N-per-distr}
	We shall denote by:
	\begin{itemize} 
		\item $\mathscr{S}^\prime(\Xi)_{\Gamma}$ the tempered distributions on $\Xi$ that are $\Gamma$-periodic with respect to the variable $x\in\X$. Similarly $C^\infty_{\Gamma}(\X)$ is the space of smooth, $\Gamma$-periodic functions on $\X\cong\Rd$.
		\item $\mathring{\mathscr{S}}_{\Delta}(\X\times\X)\subset{C}^\infty(\X\times\X)$ with off-diagonal rapid decay together with all their derivatives. It has a natural Fr\'{e}chet topology.
	\end{itemize}
\end{notation}
For a distribution $\mathfrak{K}\in\mathscr{S}^\prime(\X\times\X)$ we  denote by $\Int\,\mathfrak{K}\in\mathcal{L}\big(\mathscr{S}(\X);\mathscr{S}^\prime(\X)\big)$ the linear operator defined as: 
\beq\label{DF-Int}
\langle\big(\Int\,\mathfrak{K}\big)\phi,\psi\rangle_{\X}:=\langle\mathfrak{K},\phi\otimes\psi\rangle_{\X\times\X} \mbox{ for any }(\phi,\psi)\in\mathscr{S}(\X)\times\mathscr{S}(\X)\,.
\eeq 
We  use the notations $\Int\,\mathfrak{K}^\dagger:=\big(\Int\,\mathfrak{K}\big)^*$, so that $\langle\mathfrak{K}^\dagger,\phi\otimes\psi\rangle_{\X\times\X}=\overline{\langle\mathfrak{K},\psi\otimes\phi\rangle}_{\X\times\X}$ and $\Int\big(\mathfrak{K}\diamond\mathfrak{K}^\prime\big):=\big(\Int\,\mathfrak{K}\big)\cdot\big(\Int\,\mathfrak{K}^\prime\big)$.

For the $d$-dimensional Fourier transform we use the definitions:
\begin{equation}\begin{split}\label{DF-FourierTrsf}
		&\big(\mathcal{F}_{\X}f\big)(\xi):=\int_{\X}\hspace*{-5pt}dx\,e^{-i<\xi,x>}\,f(x)=\int_{\X}\hspace*{-5pt}dx\,\exp\Big(-2\pi i\underset{j\in\underline{d}}{\sum}\xi_j\,x_j\Big)\,f(x),\ \forall f\in L^1(\X)\\
		&\big(\mathcal{F}_{\X^*}g\big)(x):=\int_{\X^*}\hspace*{-7pt}d\xi\,e^{i<\xi,x>}\,g(\xi)=\int_{\X^*}\hspace*{-5pt}d\xi\,\exp\Big(2\pi i\underset{j\in\underline{d}}{\sum}\xi_j\,x_j\Big)\,g(\xi),\ \forall g\in L^1(\X^*).
\end{split}\end{equation}

\paragraph{The H\"{o}rmander classes of symbols.} Given $p\in\mathbb{R}$ and $\rho\in\{0,1\}$ let us define:
\[ \begin{split}
	\qquad \qquad S^p_\rho(\Xi)\,:&=\,\big\{F\in C^\infty_{\text{\tt pol}}(\Xi)\,,\,\nu^{p,\rho}_{n,m}(F)<\infty,\,\forall(n,m)\in\mathbb{N}\times\mathbb{N}\big\},\\
	\text{where:}\qquad\qquad \qquad&\\& \nu^{p,\rho}_{n,m}(F):=\underset{|\alpha|\leq n}{\max}\,\underset{|\beta|\leq m}{\max}\,\underset{(x,\xi)\in\Xi}{\sup}<\xi>^{-p+\rho m}\big|\big(\partial_x^\alpha\partial_\xi^\beta F\big)(x,\xi)\big|\,.\qquad \qquad\qquad
\end{split}\]
We also define:
\[
S^\infty_\rho(\Xi)\,:=\,\underset{p\in\R}{\bigcup}S^p_\rho(\Xi),\quad{S}^-_\rho(\Xi)\,:=\,\underset{p<0}{\bigcup}S^p_\rho(\Xi),\quad{S}^{-\infty}(\Xi)\,:=\,\underset{p\in\R}{\bigcap}S^p_\rho(\Xi).
\]

Given $\Gamma\subset\X$ as above, we shall denote by $S^p_\rho(\Xi)_\Gamma$ the symbols of class $S^p_\rho(\Xi)$ that are periodic with respect to $\Gamma$ in the $x$-variable.
We shall use the following family of $\X$ independent symbols $\mathfrak{m}_s(x,\xi):=<\xi>^s$, with $s\in\R$. 

	\begin{definition}
	For $p>0$, a symbol $F\in S^p_\rho(\Xi)$ is called elliptic, if there exists positive constants $c$ and $R$ such that, for $(x,\xi) \in \X$ s.t.  $|\xi| \geq R$, 
	\[\nonumber 
	|F(x,\xi)\big|\,\geq\,c<\xi>^p.
	\]
\end{definition}

For any $s\geq0$ we consider \textit{the Sobolev space of order} $s$ defined by:
\[
\mathscr{H}^s(\X)\,:=\,\big\{f\in{L}^2(\X)\big|\ \mathcal{F}_{\X}\mm_sf\in{L}^2(\X^*)\big\}.
\]

\paragraph{The magnetic field.}
Let us denote by $\F0^p(\X)$ the real space of smooth $p$-forms on $\X$, by $\Fb^p(\X)$ those having components of class $BC^\infty(\X)$ and by $\Fp^p(\X)$ those with polynomial growth together with all their derivatives; let $d:\F0^p(\X)\rightarrow\F0^{p+1}(\X)$ be the exterior derivation.

A smooth magnetic field is an element $B\in\F0^2(\X)$ that is closed, i.e. satisfies $dB=0$. Due to the contractibility of our affine space $\X$ any magnetic field $B\in\Fp^2(\X)$ allows for a vector potential $A\in\Fp^1(\X)$ satisfying the equality $B=dA$ and we shall always use such a choice.
We shall work with fields $B\in\Fb^2(\X)$ verifying the closure condition: $dB=0$, that we call \textit{regular magnetic fields}. Given a magnetic field $B=dA$ we use the following notations for the imaginary exponentials of the invariant integrals of $p$-forms on $p$-simplices:
\[
\Lambda^A(x,y):=\exp\Big(-i\int_{[x,y]}A\Big),\qquad\Omega^B(x,y,z):=\exp\Big(-i\int_{<x,y,z>}B\Big).
\]

Stokes' Theorem clearly implies the formula:
\[
\Omega^B(x,y,z)\,=\,\Lambda^A(x,y)\,\Lambda^A(y,z)\,\Lambda^A(z,x).
\]

\begin{definition}
	The magnetic Sobolev space of order $ s\geq0$ and vector potential $A\in\Fp^1(\X)$ is   the weighted Hilbert space:
	\[
	\mathscr{H}^s_A(\X)\,:=\,\big\{f\in{L}^2(\X)\big|\  \Op^A(\mm_s)f\in{L}^2(\X) \big\}
	\]
	with $\Op^A(\mm_s)$ considered as an operator in $\mathcal{L}\big(\mathscr{S}^\prime(\X)\big)$.
\end{definition}

\paragraph{The $d$-dimensional torus.}
An essential role in dealing with periodic problems is played by the $d$-dimensional torus 
\beq \label{D-torus}
\R^d/\Z^d\cong\mathbb{S}^d:=\big\{\z\in\mathbb{C},\ |\z|=1\big\}^d\,.  \eeq
When dealing with the identification $\X\cong\Rd$ and $\Gamma\cong\Zd$ we  denote the quotient torus $\X/\Gamma$ by $\T\cong\mathbb{S}^d$, while for the dual space $\X^*$ with its dual lattice $\Gamma_*$ we use the notation $\X^*/\Gamma_*=:\T_*\cong\mathbb{S}^d$ and denote by $\theta\equiv(\z_1,\ldots,\z_d)$ its points (with $\z_j\in\mathbb{S}$ for $1\leq j\leq d$). We mainly work with $\T_*$ and consider on it the measure $d\theta$ that is equal to the restriction of the Lebesgue measure on $\R^{2d}$ with a normalization giving the total measure $1$ to the manifold. With our definitions we get an  explicit form for the canonical quotient projection:
\beq\label{DF-p}
\p_*:\X^*\repi\T,\quad\p_*\big(\underset{1\leq j\leq d}{\sum}\xi_j\e^*_j\big)\,=\,\big(e^{2\pi i\xi_1},\ldots,e^{2\pi i\xi_d}\big),\,\quad\forall(\xi_1,\ldots,\xi_d)\in\Rd\,,
\eeq
and we choose the following discontinuous injective section for it:
\[\nonumber 
\s_*:\mathbb{T}_*\rightarrow{\cal{B}}\subset\X^*,\quad\s_*(\z_1,\ldots\z_d):=(2\pi )^{-1}(\arg\,\z_1,\ldots,\arg\,\z_d),\quad\arg\z\in[-\pi,\pi).
\] 

We recall that the irreducible unitary representations of the group $\Gamma\cong\Zd$, that one calls its \textit{characters}, form an abelian topological group that is isomorphic to $\T_*\cong\Sd$; given any $\theta=(\z_1,\ldots,\z_d)\in\T_*$ its associated character is:
\beq
\chi_\theta:\Gamma\rightarrow\mathbb{S},\quad\chi_\theta(\gamma):=e^{-i<\s_*(\theta),\gamma>}.
\eeq

\subsection{The unperturbed Hamiltonian.}
\begin{hypothesis}\label{H-h1}
$h\in{S}^p_1(\Xi)_\Gamma$ is a lower semi-bounded elliptic and $\Gamma$-periodic symbol for some $p>0$. 
\end{hypothesis}
\begin{hypothesis}\label{H-BGamma} $B^\circ\in\Fb^2(\X)$ is a closed,  $\Gamma$-periodic 2-form verifying the following {\it zero-flux property}:
	\beq\label{H-Bper}
\forall(j,k)\in\underline{d}\times\underline{d}:\quad\int\limits_{\text{\rm R}_{jk}}B^\circ=0,\quad\text{where}\quad\text{\rm R}_{jk}:=\big\{s\mathfrak{e}_j+t\mathfrak{e}_k,(s,t)\in[0,1]^2\big\}.
	\eeq
\end{hypothesis}
\begin{remark}\label{R-HH}
In \cite{HH} the authors present an elegant cohomological argument proving that the zero-flux property \eqref{H-Bper} is a necessary and sufficient condition for the existence of a $\Gamma$-periodic vector potential $A^\circ\in\Fb^1(\X)$ such that $B^\circ=dA^\circ$. We shall always use such a choice of vector potential for $B^\circ$ satisfying the zero-flux property.
\end{remark}

\begin{definition} \label{D-Hcirc}
For  $h\in{S}^p_1(\Xi)_{\Gamma}$ satisfying Hypothesis \ref{H-h1}, $B^\circ\in\Fb^2(\X)$ satisfying Hypothesis \ref{H-BGamma} and $A^\circ$ chosen as in Remark \ref{R-HH}, 
we call  "unperturbed Hamiltonian" the operator   $$ H^\circ\,:=\,\overline{\Op^{A^\circ}(h)}\,.$$ 
\end{definition}

We shall work under the following hypothesis, that in fact brings no real loss of generality, as one can easily notice.

In fact, taking into account Proposition 2.2 in \cite{CHP-5} one may notice that the inclusion of the regular magnetic field with periodic vector potential in Definition \ref{D-Hcirc} does not necessitate to go outside the usual periodic Weyl calculus on $\Rd$ and there exists some symbol $h^\circ\in{S}^p_1(\Xi)_\Gamma$ such that:
\beq
\Op^{A^\circ}(h)\,=\,\Op(h^\circ).
\eeq 

\begin{remark}
In \cite{CHP-5} we also prove that $H^\circ$ from Definition \ref{D-Hcirc} is a lower semi-bounded, self-adjoint operator in $L^2(\X)$, with domain $\mathscr{H}^p(\X)$ and commuting with all the translations $\tau_\gamma$ with $\gamma\in\Gamma$ that clearly leave invariant its domain.
\end{remark}
\begin{hypothesis}\label{H-E0}
  $H^\circ\geq{E}_0\bb1$ for some given $E_0>0$.
\end{hypothesis}

The Bloch-Floquet theory, that we recall very briefly in Appendix \ref{A-BF-Theory}, provides a unitary operator $\U_{BF}:L^2(\X)\overset{\sim}{\rightarrow}\mathscr{F}$, with $\mathscr{F}$ the direct integral Hilbert space defined in \eqref{DF-UBF} and \eqref{DF-dir-prod}, such that the decomposition: 
\beq
\U_{BF}\,H^\circ\,\U_{BF}^{-1}=\int^\oplus_{{\T_*}}d\theta\Big(\underset{k\in\B}{\sum}\lambda_k(\theta)\,\hat{\pi}_k(\theta)\Big)
\eeq
has the properties described in Theorem \ref{T-FBdec} in the Appendix \ref{A-BF-Theory} (see details in \cite{CHP-5}).

\subsection{The isolated Bloch family.}\label{SSS-iBb} 

Making use of the results in \cite{CHP-5} and briefly recalled in Appendix \ref{A-BF-Theory}, for our periodic pseudo-differential operator $H^\circ$ we shall suppose that:
\begin{hypothesis}\label{H-isBf}
There exist $k_0\in\mathbb{N}_\bullet$ and $N\in\mathbb{N}$ such that: (see figure \ref{picture1})
\begin{enumerate}
	\item $\lambda_{k_0-1}(\theta)<\lambda_{k_0}(\theta),\quad\lambda_{k_0+N}(\theta)<\lambda_{k_0+N+1}(\theta),\quad\forall\theta\in\mathbb{T}_*\ $, (by convention $\lambda_0:=-\infty$).
	\item $E_-:=\underset{\theta\in\mathbb{T}_*}{\sup}\lambda_{k_0-1}(\theta)\, <\, E_+:=\underset{\theta\in\mathbb{T}_*}{\inf}\lambda_{k_0+N+1}(\theta),\quad\text{\tt d}_0:=E_+-E_->0\,.$ 
\end{enumerate}
\end{hypothesis}

\begin{figure}[H]
	\centering
	\begin{tikzpicture}
	\draw [->] (-4, 0) -- (4, 0) node[right] {Brillouin zone};
	\draw [->] (0, -0.1) -- (0, 5) node[right] {Energy};
	%\draw [->, brown] (-4, 2.7) -- (6, 2.7) node[right] {Fermi level};
	\draw[->, domain= 0: 0.5, black] plot(\x, 3) node[right]{$E_+$};
	\draw[->, domain= 0: 0.5, red] plot(\x, 3.7) node[right]{$E_+'$};
	\draw[->, domain= 0: 1.5, red] plot(\x, 1) node[right]{$E_-'$};
	\draw[->, domain= 0: 0.7, blue] plot(\x, 1.5) node[right]{$E_-$};
	\draw[->, domain= 0: 1, blue] plot(\x, 0.5) node[right]{$E_0$};
	\draw[red, thick] (-4,2) cos (-3, 3) sin (-2.1,3.7) cos (-1,2.3) sin (0, 1)  cos (1,2.3) sin (2.1,3.7) cos (3, 3) sin (4, 2) node[right] {$\lambda_{k_0}$};
	\draw[green, thick] (-4,2.4) cos (-3, 2.8) sin (-2.1,3.2) cos (-1,2.3) sin (0, 1.9)  cos (1,2.3) sin (2.1,3.2) cos (3, 2.8) sin (4, 2.4) node[right] {$\lambda_{k_0+N}$};
	\draw[blue, thick] (-4,0.5) cos (-3, 1) sin (-2,1.5 )  cos (-1, 1) sin (0, 0.5) cos (1, 1) sin (2,1.5) cos (3,1) sin (4, 0.5) node[right] {$\lambda_{k_0-1}$};
	\draw[black, thick] (-4,3) cos (-2.3, 4) sin (-1.2,4.5) cos (-0.5,4.2)  sin (0, 4) cos (0.5, 4.2) sin (1.2 ,4.5) cos (2.3,4) sin (4,3) node[right] {$\lambda_{k_0+N+1}$};
\end{tikzpicture}
\caption{\label{picture1} Here $k_0=2$ and $N=1$.}
\end{figure}
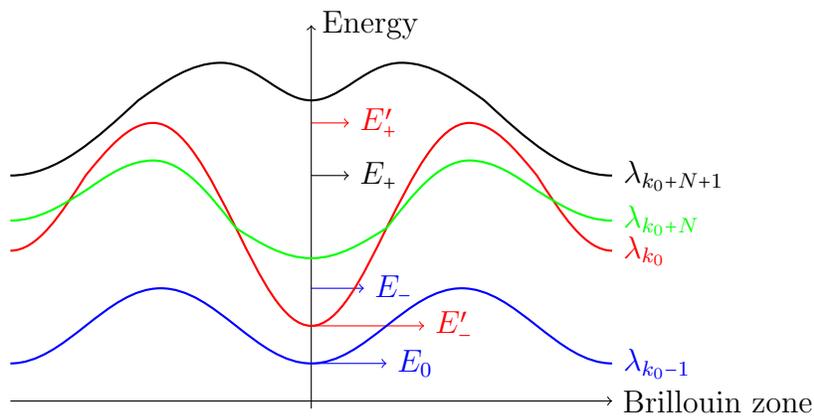

\begin{definition} The set $\B:=\big\{\lambda_k:\T_*\rightarrow\R\,,\,k_0\leq k\leq k_0+N\big\}$ satisfying Hypothesis \ref{H-isBf} is called an isolated Bloch family of the Hamiltonian $H^\circ$. We 
	call $\B$-operators the elements in $\mathbb{B}\big(P_{\B}L^2(\X)\big)$.
\end{definition}
Note that under the above Hypothesis \ref{H-isBf},  the interval $
J_\B:=(E_-,E_+)$ is not empty.
\begin{definition} \label{D-B-subsp}
	Under  Hypotheses \ref{H-h1}, \ref{H-BGamma} and \ref{H-isBf}, using  the results and notations in Appendix \ref{A-BF-Theory} we define:
\begin{itemize}
	\item \textbf{the $\B$-projection}: $P_{\B}:=\mU_{BF}^{-1}\Big(\int^\oplus_{{\T_*}}d\theta\Big(\underset{k\in\B}{\sum}\hat{\pi}_k(\theta)\Big)\mU_{BF}$,
	\item \label{dabsub} \textbf{the $\B$-subspace}: $\mathscr{L}_{\B}:=P_{\B}L^2(\X)$,
	\item \textbf{The $\B$-Hamiltonian}: $H_\B\,:=\,\mU_{BF}^{-1}\Big(\int^\oplus_{{\T_*}}d\theta\Big(\underset{k\in\B}{\sum}\lambda_k(\theta)\,\hat{\pi}_k(\theta)\Big)\mU_{BF}\,\in\,\mathbb{B}\big(\mathscr{L}_{\B}\big)$\,.
\end{itemize}
\end{definition}

An important technical fact is Proposition \ref{R-p-symb} below in which we prove that there exists $p_{\B}\in{S}^{-\infty}(\Xi)_{\Gamma}$ and $h_{\B}\in{S}^{-\infty}(\Xi)_{\Gamma}$ such that $P_{\B}=\Op^\circ(p_{\B})$ and $H_{\B}=\Op^\circ(h_{\B})$.

\subsection{The perturbation.} 

Let us consider the self-adjoint operator $H^\circ$ as defined in Definition \ref{D-Hcirc} with the isolated Bloch family $\B$ as in Hypothesis \ref{H-isBf} and let us add a perturbing magnetic field:
\beq\label{HF-Beps}
B(x)\,:=\,\epsilon{B}^{\epsilon}(x)\,=\,\epsilon\,{dA}^\epsilon(x)\,\equiv\,dA(x)
\eeq
where  $\epsilon\in[0,\epsilon_0]$ for some $\epsilon_0>0$ and the family $\BB:=\{B^\epsilon\}_{\epsilon\in[0,\epsilon_0]}$ is a bounded subset in $\Fb^2(\X)$. Thus we have a total magnetic field $$B^\circ(x)+B(x)=B^\circ(x)+\epsilon\,B^\epsilon(x)$$ with a vector potential $A^\circ(x)+A(x)=A^\circ(x)+\epsilon\,A^\epsilon(x)$ in $\Fp^1(\X)$.
\begin{notation}\label{N-B}
In the sequel we shall frequenly use the notation $C(\BB)$ for a strictly positive constant that depends on the bounded subset $\BB$ as defined after formula \eqref{HF-Beps}, more precisely on the supremum on the bounded subset $\BB$ of a finite set of continuous semi-norms on $\Fb^2(\X)$.
\end{notation}

\begin{definition}\label{D-Heps}
The perturbed Hamiltonian is:
$
H^\epsilon\,:=\,\overline{\Op^{A^\circ+A}(h)}$.
\end{definition}

The following proposition, proven in \cite{CHP-5},  gives  the technical ingredient to deal with this 2-step quantization obtained by superposing the perturbing magnetic field $B$ in \eqref{HF-Beps} on the periodic magnetic field $B^\circ $ from Hypothesis \ref{H-BGamma}.

\begin{proposition}
	Suppose given a symbol $F\in{S}^p_1(\Xi)$ for some $p\in\R$ and two magnetic fields $B=dA$ and $B^\prime=dA^\prime$ in $\Fb^2(\X)$ with vector potentials $A\in\Fp^1(\X)$ and $A^\prime\in\Fb^1(\X)$. Then:
	\[
	\Op^{A+A^\prime}(F)\,=\,\Op^A\big(\mS^{A^\prime}[F]\big)\,,
	\]
	with: $$\mS^{A^\prime}[F]\,=\,
	(\bb1\otimes\mathcal{F}_{\X})\big(\mathfrak{K}^{A^\prime}[F]\circ\Upsilon^{-1}\big)\,=\,
	(\bb1\otimes\mathcal{F}_{\X})\big((\Lambda^{A^\prime}\mathfrak{K}[F])\circ\Upsilon^{-1}\big)\in{S}^p_1(\Xi)\,.$$
\end{proposition}
We shall use the following short notations:
\begin{equation}\begin{split}\label{N-A-eps}
\Lambda^\epsilon\,\equiv\,\Lambda^{A^\circ+\epsilon{A^\epsilon}},\ \tLambda^\epsilon\,\equiv\,\Lambda^{\epsilon{A^\epsilon}},\ \Op^\epsilon\,\equiv\,\Op^{\epsilon{A^\epsilon}},\ \Omega^\epsilon\,\equiv\,\Omega^{B^\circ+\epsilon{B^\epsilon}},\ \tOmega^\epsilon\,\equiv\,\Omega^{\epsilon{B^\epsilon}},\ \sharp^\epsilon\,\equiv\,\sharp^{\epsilon{B^\epsilon}}.
\end{split}\end{equation}

\begin{remark}\label{R-main}
The above arguments imply that given any symbol $F\in{S}^m_1(\Xi)_\Gamma$, for any $m\in\R$, there exists a unique symbol $F^\circ\in{S}^m_1(\Xi)_\Gamma$ such that $\Op(F^\circ)=\Op^{A_\circ}(F)$ and we also have the equalities:
\beq
\Op^{A^\circ+\epsilon{A^\epsilon}}(h)=\Op^\epsilon(h^\circ),\qquad\,H^\epsilon\,:=\,\overline{\Op^\epsilon(h^\circ)}.
\eeq
\end{remark}

Using the notations and results recalled in Appendix \ref{A-m-PsiDO} we may prove the following statement.
\begin{proposition}
The operator $H^\epsilon$ in Definition \ref{D-Heps} has domain $\mathscr{H}^p_A(\X)$ and is a lower bounded self-adjoint operator in $L^2(\X)$.
\end{proposition}
    
\subsection{The main results.} 

\paragraph{Some spaces of matrices.} Given $M\in\mathbb{N}_\bullet$, we denote by $\mathscr{M}_M$ the $C^*$-algebra of $M\times M$ complex matrices. We shall use the notation $\mathscr{M}_\Gamma[\mathfrak{A}]$ for the complex linear space of infinite matrices indexed by $\Gamma\times\Gamma$, having entries in a $C^*$-algebra $\mathfrak{A}$. Our main interest will be in the complex subspace $\mathscr{M}^\circ_\Gamma[\mathfrak{A}]$ of matrices having rapid decay outside the diagonal; given any faithful representation $\rho:\mathfrak{A}\rightarrow\mathbb{B}(\mathcal{H})$ we may view $\mathscr{M}^\circ_\Gamma[\mathfrak{A}]$ as a sub-algebra of $\mathbb{B}\big(\ell^2(\Gamma;\mathcal{H})\big)$ endowed with the operator norm that we denote by $\|\cdot\|_{\mathbb{B}(\ell^2(\Gamma;\mathcal{H}))}$. We shall also deal with families of Toeplitz infinite matrices that we shall view as infinite sequences and use the notation:
$$
{\cal{s}}\big(\Gamma;\mathfrak{A}\big):=\big\{\mathring{V}:\Gamma\rightarrow\,\mathfrak{A},\ \forall n\in\mathbb{N},\ \underset{\gamma\in\Gamma}{\sup}<\gamma>^n\big\|\mathring{V}(\gamma)\big\|_{\mathfrak{A}}\hspace*{-0,2cm}<\infty\ \big\}.
$$

\begin{theorem}\label{T-I}
	Consider the self-adjoint operator $H^\circ$ from Definition \ref{D-Hcirc} with the isolated Bloch family $\B$ as in Hypothesis \ref{H-isBf} and its perturbed version $H^\epsilon:\mathscr{H}^p_A(\X)\rightarrow{L}^2(\X)$, as in Definition \ref{D-Heps} with a perturbing magnetic field $B=dA$ as in \eqref{HF-Beps}. Then, there exists some $\epsilon_0>0$ such that the following statements are true.
\begin{enumerate}
	\item \label{point1} There exist an orthogonal projection $P^\epsilon_{\B}=\Op^\epsilon(p^\epsilon_{\B})\in\mathbb{B}\big(L^2(\X)\big)$, a constant $C>0$ and for any continuous semi-norm $\lnu$ on ${S}^{-\infty}(\Xi)$ some constant $C_{\nu}>0$ such that  for any $\epsilon\in[0,\epsilon_0]$: 
    $$
    p^\epsilon_{\B}\in{S}^{-\infty}(\Xi),\ \Vert [\,H^\epsilon\,,\,P^\epsilon_{\B}\,]\Vert \,\leq\,C{\epsilon},\ \text{and }\lnu(p^\epsilon_{\B}-p_{\B})\leq{C}_{\nu}\epsilon\,,
   $$
   with $\Op(p_{\B})=P_{\B}$  given in Definition \ref{D-B-subsp}.
	\item \label{point2} If we define the \emph{reduced perturbed Hamiltonian} associated with $\B$ by $$\ham^\epsilon_{\B}:=P^\epsilon_{\B}H^\epsilon{P}^\epsilon_{\B}\,,$$ then $\ham^0_{\B}=H_{\B}$ given in Definition \ref{D-B-subsp} and
	there exists $\delta_0>0$ such that for any $\delta\in[0,\delta_0]$ the interval $J^\delta_\B:=\big(E_-+{2}\delta\,,\,E_+-{2}\delta\big)$ is non empty and there exists $\epsilon_\delta\in(0,\epsilon_0]$ and $C>0$ such that for any $\epsilon\in[0,\epsilon_\delta]$ we have that:
	\begin{enumerate}[i.]
		\item  $\max\Big\{\underset{\lambda\in\sigma(\ham^\epsilon_{\B})\cap{J^\delta_\B}}{\sup}\dist\Big(\lambda,\sigma(H^{\epsilon})\Big )\,,\,\underset{\lambda\in\sigma(H^{\epsilon})\cap\,J^\delta_\B}{\sup}\dist\Big(\lambda,\sigma(\ham^\epsilon_{\B})\Big ) \Big\}\,\leq\,C\epsilon^2\,.$
		\item  For any $v\in E_J( H^{\epsilon})\,L^2(\X)$:
		\[\nonumber
		\big\|\exp\big(-itH^{\epsilon}\big)v -\exp\big(-it\ham^\epsilon_{\B}\big)v\big\|_{L^2(\X)}\leq\,C\,\big[\epsilon\,+\,(1+|t|) ^3\, \epsilon^2\big]\|v\|_{L^2(\X)}.
		\]
	\end{enumerate}
	\item \label{point3} There exists  $\nB\in\{N+1,\dots,N+1+\lfloor d/2\rfloor\}$ \footnotemark\label{fnm:1}\footnotetext{we recall that $\lfloor d/2\rfloor :=\max\{k\in\mathbb{N},\ k\leq (d/2)\}$\label{fnt:1}} and an injective $C^*$-algebra morphism $\mM^\epsilon_{\B}:\mathbb{B}\big(P^\epsilon_{\B}L^2(\X)\big)\rightarrow\mathscr{M}^\circ_\Gamma[\mathscr{M}(n_{\B})]$ such that:
	\begin{enumerate}[\rm a.]
		\item \label{point3a} $\sigma(\ham^\epsilon_{\B})\,=\,\sigma(\mM^\epsilon_{\B}[\ham^\epsilon_{\B}]) \setminus\{0\}$ and $\mM^\epsilon_{\B}\big[\exp\big(-it\ham^\epsilon_{\B}\big)\big]=\exp\big(-it\mM^\epsilon_{\B}[\ham^\epsilon_{\B}]\big)$.
		\item \label{point3b} There exists $\mathring{\mm}_{\B}[\widehat{H}_{\B}]\in\mathcal{s}\big(\Gamma;\mathscr{M}(\nB)\big)$ independent of $\epsilon\in[0,\epsilon_0]$ such that:
		 \beq\label{point3c} \big[\mM^\epsilon_{\B}[\ham^\epsilon_{\B}]\big]_{\alpha,\beta}=\tLambda^\epsilon(\alpha,\beta)[\mathring{\mm}_{\B}[\widehat{H}_{\B}]]_{\alpha-\beta}+\mathcal{O}(\epsilon),\quad {\forall\epsilon\in[0,\epsilon_0]}.
		\eeq 
	\end{enumerate}
\end{enumerate}
\end{theorem}
\begin{remark}
	A more explicit form of the remainder in   Formula \eqref{point3c}  is obtained in Subsection  \ref{SS-PN-V-2} and Appendix \ref{SSS-PN-VI-2}  (see formula \eqref{F-PN-VI-1}).
\end{remark}

Let us also formulate a version of Theorem \ref{T-I} that gives a more detailed description of the \textit{effective Hamiltonian} $\mM^\epsilon_{\B}[\ham^\epsilon_{\B}]\in\mathscr{M}^\circ_\Gamma[\mathscr{M}(n_{\B})]$ when the perturbing magnetic field in \eqref{HF-Beps} is a weak fluctuation around a non-zero constant field.  

\begin{hypothesis}\label{H-Bepsc}
	Given a constant magnetic field $B^\bullet=dA^\bullet$, the perturbing magnetic field is  of the form:
	\beq \label{HF-Bepsc}
	B(x)\,:=\,\epsilon\big(B^\bullet\,+\,\cc\,B^{\epsilon,\cc}(x)\big)
	\eeq
	with $\big\{B^{\epsilon,\cc},\ (\epsilon,\cc)\in[0,\epsilon_0]\times[0,1]\big\}$ a bounded set in $\Fb^2(\X)$, $\epsilon\in[0,\epsilon_0]$ as above and $\cc\in[0,1]$ a new 'small parameter' controlling the weak fluctuations $\cc\epsilon\,B^\epsilon\in\Fb^2(\X)$.
\end{hypothesis}
\begin{remark}
As any magnetic field satisfying Hypothesis \ref{H-Bepsc} is of the form \eqref{HF-Beps} we deduce that Theorem \ref{T-I} is true for the perturbed Hamiltonian defined with such a magnetic field, that we shall denote by $H^{\epsilon,\cc}$. Suppose that it verifies Hypothesis \ref{H-isBf} having the isolated Bloch family $\B$. In order to keep in mind the special structure of the perturbing magnetic field we shall use the modified notations: $P^{\epsilon,\cc}_{\B}$ and $\ham^{\epsilon,\cc}_{\B}$ instead of $P^\epsilon_{\B}$ and $\ham^\epsilon_{\B}$ for the operators appearing in Theorem \ref{T-I} with $H^\epsilon$ replaced by $H^{\epsilon,\cc}$.
\end{remark}

Under Hypothesis \ref{H-Bepsc}, the third statement in Theorem \ref{T-I} may be replaced by the following more detailed statement:

\begin{theorem}\label{T-II}
Under the same hypothesis as in Theorem \ref{T-I}, let us suppose that the perturbing magnetic field in \eqref{HF-Beps} verifies Hypothesis \ref{H-Bepsc}. Let us denote by $H\ec$ the perturbed Hamiltonian defined as in Definition \ref{D-Heps} with the magnetic field as above. Then for any $(\epsilon,\cc)\in[0,\epsilon_0]\times[0,1]$ there exists an injective $C^*$-algebra morphism $\mM\ec_{\B}:\mathbb{B}\big(P\ec_{\B}L^2(\X)\big)\rightarrow\mathscr{M}^\circ_\Gamma[\mathscr{M}(n_{\B})]$ and a sequence $\mathring{\mm}^{\epsilon}_{\B}\in\mathcal{s}\big(\Gamma;\mathscr{M}(\nB)\big)$ such that points 3.\ref{point3a} and 3.\ref{point3b} in Theorem~\ref{T-I} remain true for $\ham^\epsilon_{\B}$ replaced by $\ham\ec_{\B}$ and  \eqref{point3c}   is replaced by the following equality:
\beq
\begin{aligned}\label{F-PN-V-11-c-T}
	&\big[\mathfrak{M}\ec_{\B}[\ham\ec_{\B}]_{\alpha,\beta}\big]_{p,q}\hspace*{-6pt}=[\tLambda\ec\tLambda^{\epsilon,0}](\alpha,\beta)\Big[\big[\mathring{\mathfrak{m}}^\epsilon_{\B}[\widehat{H}_{\B}]_{\alpha-\beta}\big ]_{\clb p,q}\hspace*{-6pt}+\cc\epsilon\big[\widetilde{\mathfrak{M}}\ec_{\B}[\ham\ec_{\B}]_{\alpha,\beta}\big]_{p,q}\Big]
\end{aligned}
\eeq
where: $$\mathring{\mathfrak{m}}^\epsilon_{\B}[\widehat{H}_{\B}]_{\gamma}=\mathring{\mathfrak{m}}_{\B}[\widehat{H}_{\B}]_{\gamma}+\epsilon\,\widetilde{\mathfrak{M}}^{\epsilon}_{\B}[\mathfrak{K}(\ham^\epsilon_\B)]_{\beta+\gamma,\beta}$$
 and $$\epsilon\widetilde{\mathfrak{M}}^{\epsilon}_{\B}[\mathfrak{K}(\ham^\epsilon_\B)]_{\alpha,\beta}$$ is  the remainder in  Formula \eqref{point3c}  for the case $B=\epsilon{B}^\bullet$.
\end{theorem}

\section{The Parseval frame of the isolated Bloch family}\label{S-Pfr-ibBf}

In order to study the properties of the isolated Bloch family projector $P_{\B}$ and define its ``magnetic" version, we shall make use of Parseval frames (the reader may find some basic facts about frames in Hilbert spaces in Appendix \ref{A-frames} and more details in \cite{Chris} or references therein). Starting with this section we shall constantly use the notations and definitions in the Appendices. 

\subsection{The decomposition of the Hamiltonian induced by the isolated Bloch family}\label{SS-free-dyn}

Let us begin by emphasizing some consequences of Hypothesis \ref{H-isBf} on the structure of the unperturbed operator $H^\circ$ in Definition \ref{D-Hcirc}, that play a crucial role in our arguments. 
Theorem \ref{T-FBdec} in Appendix \ref{A-BF-Theory} and  Hypothesis \ref{H-isBf} allow us to define the following orthogonal decomposition: 
\begin{equation}\label{dec-id}
	\begin{aligned}
		&\bb1_{\mathcal{H}}=P_0\,\oplus\,P_\B\,\oplus\,P_\infty,\quad\text{given by:}\\
		&P_0:=\U_{BF}^{-1}\Big[\int_{\T_*}^{\oplus}d \theta\,\widehat{P_0}(\theta)\Big] \U_{BF},\quad\widehat{P_0}(\theta):=\underset{1\leq k\leq k_0-1}{\sum}\hat{\pi}_k( \theta)\\
		&P_\B:=\U_{BF}^{-1}\Big[\int_{\T_*}^{\oplus}d \theta\,\widehat{P_\B}(\theta)\Big] \U_{BF},\quad\widehat{P_\B}(\theta):=\underset{k_0\leq k\leq k_0+N}{\sum}\hat{\pi}_k( \theta)\\
		&P_\infty:=\U_{BF}^{-1}\Big[\int_{\T_*}^{\oplus}d \theta\,\widehat{P_\infty}(\theta)\Big]\U_{BF},\quad\widehat{P_\infty}(\theta):=\underset{k_0+N+1\leq k}{\sum}\hat{\pi}_k(\theta)
	\end{aligned}
\end{equation}
and an associated decomposition of the Hamiltonian: $ H^\circ\,=\,H_0\,\oplus\,H_\B\,\oplus\,H_\infty$ given explicitely by:
\begin{align}\label{eq:3.1}
	&H_0:=\U_{BF}^{-1}\Big[\int_{\T_*}^{\oplus}d \theta\,\widehat{H_0}(\theta)\Big] \U_{BF},\quad\widehat{H_0}(\theta):=\underset{1\leq k\leq k_0-1}{\sum}\lambda_k(\theta)\,\hat{\pi}_k( \theta),\\
	\label{eq:3.2}
	&H_\B:=\U_{BF}^{-1}\Big[\int_{\T_*}^{\oplus}d \theta\,\widehat{H_\B}(\theta)\Big] \U_{BF},\quad\widehat{H_\B}(\theta):=\underset{k_0\leq k\leq k_0+N}{\sum}\lambda_k(\theta)\,\hat{\pi}_k( \theta),\\
	&H_\infty:=\U_{BF}^{-1}\Big[\int_{\T_*}^{\oplus}d \theta\,\widehat{H_\infty}(\theta)\Big]\U_{BF}\quad\widehat{H_\infty}(\theta):=\overline{\underset{k_0+N+1\leq k}{\sum}\lambda_k(\theta)\,\hat{\pi}_k( \theta)}.\label{H-infinit} 
\end{align}

Using Hypotheses \ref{H-h1}, \ref{H-E0} and \ref{H-isBf} we notice that (for $E_0\leq E_-<E_+$ as defined in this last Hypothesis):
\[
\sigma(H_0)\subset\{0\}\,\bigsqcup\, [E_0,E_-],\quad\sigma(H_\B)\subset\{0\}\bigsqcup\, [E_-^\prime,E_+^\prime],\quad\sigma(H_\infty)\subset\{0\}\bigsqcup\, [E_+,\infty),
\]
for some $E'_-$ and $E'_+$ such that $0<E_0\leq E'_-<E'_+<\infty$. If we have the strict inequalities $E_-<E_-^\prime$ and $E_+^\prime<E_+$ we are in a situation as in Figure \ref{picture2} where the isolated band consists of just one eigenvalue and it forms a spectral island for $H^\circ$, separated from the rest of the spectrum by two gaps. Then the three orthogonal projections in the above decomposition are in fact spectral projections of $H^\circ$ associated with disjoint components of the spectrum $\sigma(H^\circ)$ so that the ``{global} gap condition" is fulfilled and we are in the situation studied in \cite{CIP}. Thus, we shall be mainly interested in the case:
\beq\label{F-sp-points}
0<E_0\leq E_-^\prime\leq E_-<E_+\leq E_+^\prime<\infty,
\eeq
illustrated in Figure \ref{picture1}, where we have an isolated band with two crossing Bloch levels in red and green. (Formally, the green colour should always be on top of the red colour because $\lambda_{k_0}\leq \lambda_{k_0+N}$, but we will never treat them individually, only as a well-defined isolated family.) The energy interval we are interested in is $(E_-,E_+)$ where $E_-$ is the maximum of the blue eigenvalue $\lambda_{k_0-1}$ and $E_+$ is the minimum of the the black one $\lambda_{k_0+N+1}$ and the Hamiltonian $H^\circ$ does not have a spectral gap and the three orthogonal projections in \eqref{dec-id}  are not necessarily spectral projections of $H^\circ$.
\begin{center}
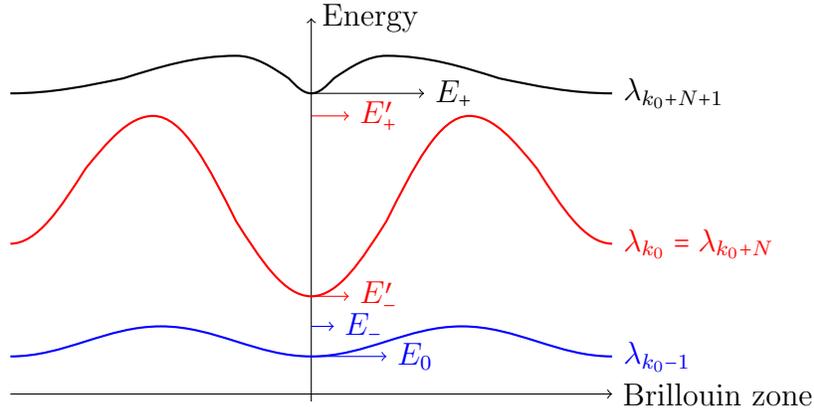

	\begin{tikzpicture}
		\draw [->] (-4, 0) -- (4, 0) node[right] {Brillouin zone};
		\draw [->] (0, -0.1) -- (0, 5) node[right] {Energy};
	%	\draw [->, brown] (-4, 2.5) -- (7, 2.5) node[right] {Fermi level};
		\draw[->, domain= 0: 1.5, black] plot(\x, 4) node[right]{$E_+$};
		\draw[->, domain= 0: 0.5, red] plot(\x, 3.7) node[right]{$E_+'$};
		\draw[->, domain= 0: 0.5, red] plot(\x, 1.3) node[right]{$E_-'$};
		\draw[->, domain= 0: 0.3, blue] plot(\x, 0.9) node[right]{$E_-$};
		\draw [->, domain= 0: 1, blue] plot(\x, 0.5) node[right]{$E_0$};
		\draw[red, thick] (-4,2) cos (-3, 3) sin (-2.1,3.7) cos (-1,2.3) sin (0, 1.3)  cos (1,2.3) sin (2.1,3.7) cos (3, 3) sin (4, 2) node[right] {$\lambda_{k_0}=\lambda_{k_0+N}$};
		\draw[blue, thick] (-4,0.5) cos (-3, 0.7) sin (-2,0.9) cos (-1,0.7) sin (0, 0.5) cos (1, 0.7) sin (2,0.9) cos (3,0.7) sin (4, 0.5) node[right] {$\lambda_{k_0-1}$};
		\draw[black, thick] (-4,4) cos (-2.5, 4.2) sin (-1,4.5) cos (-0.3,4.2)  sin (0, 4) cos (0.3, 4.2) sin (1,4.5) cos (2.5,4.2) sin (4,4) node[right] {$\lambda_{k_0+N+1}$};
	\end{tikzpicture}
	\captionof{figure}{\label{picture2} Here $k_0=2$ and $N=0$.}
\end{center}
 We emphasize that the non-zero spectrum of the fibre of $H_\B$ is always well isolated from the other Bloch bands which build up $H_0$ and $H_\infty$.
 
 \subsection{The vector bundle associated with  the Bloch-Floquet-Zak representation}\label{SS-BFZ-pbd}
 
 We shall recall here the construction elaborated in \cite{P-07} in order to have the necessary background to formulate our main technical ideas and go further with our analysis. The point is to build a vector bundle whose sections will be in bijective corespondence with the functions in the BFZ representation space $\mathscr{G}$ in Appendix \ref{A-BF-Theory}. In fact we shall need to use this construction for several families of functions with specific regularity properties. More precisely we shall need to work with all the spaces in \eqref{F-GF}. Let us denote generically by $\mathscr{K}$ any of the spaces $C^\infty(\T)\subset\mathscr{H}^p(\T)\subset{L}^2(\T)$ and recalling the definition of $\mathfrak{f}$ before \eqref{F-GF}, let us consider the unitary representations:
 \beq\label{F-a}
 \mathring{U}_*:\Gamma_*\rightarrow\mathbb{U}(\mathscr{K}),\qquad[\mathring{U}_*(\gamma^*)\mathring{f}](\omega):=\mathfrak{f}_{-\gamma^*}\big(\s(\omega)\big)\mathring{f}(\omega),\quad\forall\omega\in\T.
 \eeq
 
 We define $\mathfrak{X}[\mathscr{K}]:=\X^*\times{L}^2(\T)$, considering it as infinite dimensional smooth real manifold and consider the following equivalence relation on $\mathfrak{X}[\mathscr{K}]$ induced by the unitary representation \eqref{F-a}:
 \beq
 (\xi,\mathring{f})\,\Lsim\,(\zeta,\mathring{g})\quad\Leftrightarrow			\quad\exists\gamma^*\in\Gamma^*,\ \zeta=\xi+\gamma^*,\ \mathring{g}=	\mathring{U}_*(\gamma^*)\mathring{f}.			
 \eeq
 Let us denote the family of equivalence classes by $\widehat{\mathfrak{X}}[\mathscr{K}]:=\mathfrak{X}[\mathscr{K}]/\Lsim$ and consider the surjection $\mathfrak{t}:\widehat{\mathfrak{X}}[\mathscr{K}]\ni[\xi,\mathring{f}]_\sim\mapsto\p_*(\xi)\in\T_*$, that is well defined due to the fact that $(\xi,\mathring{f})\,\Lsim\,(\zeta,\mathring{g})$ implies that $\p_*(\xi)=\p_*(\zeta)$. We have the homeomorphisms: 
 \beq
 \mathfrak{t}^{-1}(\theta)=\big\{[\s_*(\theta),\mathring{f}],\ \mathring{f}\in\mathscr{K}\big\}\simeq\mathscr{K}
 \eeq 
 and $\mathfrak{t}^{-1}\mathcal{O}\simeq\mathcal{O}\times\mathscr{K}$ for any small enough neighborhood $\mathcal{O}$ of any $\theta\in\T_*$.
 
 \begin{notation}
 Let $\mathfrak{S}\big(\mathcal{O},\widehat{\mathfrak{X}}[\mathscr{K}]\big)$ be the complex space of smooth sections over the open set $\mathcal{O}\subset\T_*$ and $\mathfrak{S}^2\big(\mathcal{O},\widehat{\mathfrak{X}}[\mathscr{K}]\big)$ the complex space of classes of square integrable sections over the open set $\mathcal{O}\subset\T_*$.
 \end{notation}
 
 \begin{proposition}
 We have an unitary map 
 \beq
 \mathfrak{S}^2\big(\T_*,\widehat{\mathfrak{X}}[\mathscr{K}]\big)\overset{\sim}{\longrightarrow}\big\{F\in{L}^2_{\text{\tt loc}}\big(\X^*;\mathscr{K}\big),\ \tau_{\gamma^*}F=\mathfrak{f}_{-\gamma^*}F\big\}\,=:\,L^2_{l,\mathfrak{f}}[\mathscr{K}]\,,
 \eeq 
 that restricts to a bijection 
 \beq
 \mathfrak{S}\big(\T_*,\widehat{\mathfrak{X}}[\mathscr{K}]\big)\overset{\sim}{\longrightarrow}\big\{F\in{C}^\infty\big(\X^*;\mathscr{K}\big),\ \tau_{\gamma^*}F=\mathfrak{f}_{-\gamma^*}F\big\}\,=:\,C^\infty_{\mathfrak{f}}[\mathscr{K}]\,.
 \eeq
 \end{proposition}
 \begin{proof}
 Suppose given $F\in C^\infty_{\mathfrak{f}}[\mathscr{K}]$ and for $\theta\in\T_*$ let us define 
 \beq
 \Phi_F(\theta):=\big[\s_*(\theta),(F\circ\s_*)(\theta)\big]_\sim\in\widehat{\mathfrak{X}}[\mathscr{K}]_{\theta}\subset\widehat{\mathfrak{X}}[\mathscr{K}]\,.
 \eeq
 One easily verifies the smoothness properties using the smoothness of $F$ and concludes  that $\Phi_F\in\mathfrak{S}(\T_*,\widehat{\mathfrak{X}}[\mathscr{K}])$. 
 
 Conversely, let $\Phi\in\mathfrak{S}(\T_*,\widehat{\mathfrak{X}}[\mathscr{K}])$ and given $\xi\in\X^*$ let us notice that there exists a unique $\widetilde{\Phi}(\xi)\in\mathscr{K}$ such that:
 \beq
 (\Phi\circ\p_*)(\xi)\,=\,\big[\xi,\widetilde{\Phi}(\xi)\big]_\sim.
 \eeq
 In fact, using the decomposition $\X^*=\Gamma_*\times{\cal{B}}$ with the notations discussed in Appendix~\ref{ss1.3}, there exists a unique $\varphi\in\mathscr{K}$ such that we can write:
 \beq
 (\Phi\circ\p_*)(\xi)=\big\{\big(\hat{\xi}+\gamma^*\,,\,\mathfrak{f}_{-\gamma^*}\varphi\big)\big\}_{\gamma^*\in\Gamma_*}
 \eeq
 and we obtain that
 \beq
 \widetilde{\Phi}(\xi):=\mathfrak{f}_{-\iota^*(\xi)}\varphi\in\mathscr{K}.
 \eeq
 \end{proof}
 
 We shall use the notations:
 \beq
 \mathfrak{F}:=\widehat{\mathfrak{X}}[L^2(\T)],\quad\mathfrak{F}^p:=\widehat{\mathfrak{X}}[\mathscr{H}^p(\T)],\quad\mathfrak{F}^\infty:=\widehat{\mathfrak{X}}[C^\infty(\T)].
 \eeq
 
 \subsection{The vector sub-bundle associated with  the isolated Bloch family $\mathbf{\B}$.}
 
 \begin{proposition}\label{P-p-symb} There exists a vector sub-bundle $\mathfrak{F}_\B\subset\mathfrak{F}^\infty\repi\T_*$ with an unitary map $\mathfrak{S}^2(\T_*,\mathfrak{F}_\B)\overset{\sim}{\longrightarrow}\U_{BFZ}\big[P_\B{L}^2(\X)\big]$ putting into bijective correspondence the smooth sections in $\mathfrak{S}(\T_*,\mathfrak{F}_\B)$ with smooth functions in $\U_{BFZ}\big[P_\B{L}^2(\X)\big]\subset\mathscr{G}^\infty$.
 \end{proposition}
 \begin{proof}
We begin by noticing that the formulas \eqref{dec-id} imply that any vector in the range of $\hat{P}_\B(\theta)$ is a finite linear combination of elements in the ranges of the projections $\hat{\pi}_k(\theta)$ with $k_0\leq k\leq k_0+N$ and thus, Theorem \ref{T-FBdec-proj} implies that they belong to $\mathscr{F}_{\theta}^\infty$. It follows that 
$\U_{BFZ}\big[P_\B{L}^2(\X)\big]\subset\mathscr{G}^\infty$.

Due to Hypothesis \ref{H-I}, for any $\theta\in\T_*$ we can find a cut-off function $\chi_{\theta}\in\,C^\infty_0(\R)$ with support in $\widetilde{I}_{\theta}:=(\lambda_{k_0-1}(\theta),\lambda_{k_0+N+1}(\theta)\ne\emptyset$ and equal to 1 on $J_{\theta}:=[\lambda_{k_0}(\theta),\lambda_{k_0+N}(\theta)]$ so that:
 	\[
 	\widehat{P}_{\B}(\theta)\,=\,\chi_{\theta}\big(\widehat{H}^\circ(\theta)\big).
 	\]
 Let us notice further that for any $\theta\in\T^d_*$, we can find an open neighbourhood $\mathscr{O}\subset\T^d_*$ and a local smooth section $\s_{\mathscr{O}}:\mathscr{O}\rightarrow\X^*$ such that we may choose the same $\chi_{\theta}\equiv\chi_{\mathscr{O}}$ for all $\theta\in\mathscr{O}$. Using the second statement of Theorem \ref{T-FBdec} we deduce the smoothness of the application:
 	\begin{equation}
 		\begin{aligned}\label{F-201}
 			&\s_{\mathscr{O}}\mathscr{O}\ni\xi\mapsto\mathfrak{f}_{\xi}^{-1}\chi_{\mathscr{O}}\big(\widehat{H}^\circ\big(\p_*(\xi)\big)\mathfrak{f}_{\xi}\in\mathbb{B}\big(L^2(\T)\big).
 		\end{aligned}
 	\end{equation}
 	Due to the compactness of $\T_*$ it follows that we can find an open cover $\big\{\mathscr{O}_m,\ m\in\underline{N}\big\}$ for some $N\in\Nb$ and an associated finite partition of unity $\big\{\rho_m\in\,C^\infty_0(\T_*),\ m\in\underline{N}\big\}$, as well as a family of functions $\chi_m\in\,C^\infty_0(\R),\ m\in\underline{N}\big\}$ such that $\chi_m=1$ on $J_{\theta}$ and $\supp\chi_m\subset\widetilde{I}_{\theta}$ for any $\theta\in\mathscr{O}_m$. Finally, we notice that:
 	\beq
 	\widehat{P}_{\B}\,=\,\underset{1\leq m\leq N}{\sum}\rho_m\,\chi_m\big(\widehat{H}^\circ\big)
 	\eeq
 	and the following map is smooth:
 	\beq
 	\X^*\ni\xi\mapsto\widetilde{P}_\B(\xi):=\mathfrak{f}_{\xi}^{-1}\widehat{P}_{\B}\big(\p_*(\xi)\big)\mathfrak{f}_{\xi}\in\mathbb{B}\big(L^2(\T)\big).
 	\eeq
 	
 Finally, our previous arguments show that $\widetilde{P}_{\B}\mathscr{G}\subset\mathscr{G}^\infty$ and for any $\tilde{f}\in\mathscr{G}$:
 \beq\begin{split}
 \big[\tau_{\gamma^*}\widetilde{P}_{\B}\tilde{f}\big](\xi)=\mathfrak{f}_{\xi+\gamma*}^{-1}\widehat{P}_{\B}\big(\p_*(\xi+\gamma^*)\big)\mathfrak{f}_{\xi+\gamma^*}\,\tilde{f}(\xi+\gamma^*)=\mathfrak{f}_{\gamma*}^{-1}\widetilde{P}_{\B}(\xi)\,\tilde{f}(\xi)=\big[\mathfrak{f}_{\gamma^*}^{-1}\widetilde{P}_{\B}\tilde{f}\big](\xi).
 \end{split}\eeq
 We conclude that we may take $\mathscr{K}:=\widetilde{P}_{\B}(1)\big[C^\infty(\T)\big]$ in the abstract construction done in the previous section and obtain the vector bundle $\mathfrak{F}_{\B}$ of rank $N+1$ as sub-bundle of $\mathfrak{F}^\infty$.
 
 \end{proof}

\subsection{Construction of the ``unperturbed" Parseval frame}\label{SS-B-P-fr}
The main property of fibre bundles that we shall need in our construction is the following Embedding Theorem (see  Theorem 7.2 in Chapter 8 of \cite{Hu},   we use the notation introduced in footnote \ref{fnt:1}):
\begin{theorem}
Given a smooth Euclidean vector bundle $\mathfrak{B}\repi M$ of rank $n$ over a smooth real manifold $M$ of dimension $d$, there exists $m\in\mathbb{N}$ with 
	$0\leq m\leq\lfloor d/2\rfloor$ 
	and a smooth isometric bundle homomorphism $\mathfrak{I}:\mathfrak{B}\rightarrow M\times\mathbb{C}^{n+m}$.
	\end{theorem}
	
Applying the Embedding Theorem to our case where $n=N+1$ and $M=\T_*$ we have:
\begin{corollary}\label{C-bd-triv}
	There exist some $n_\B\in \Nb$ with $N+1\leq n_\B\leq\,{N+1+\lfloor d/2\rfloor}$ and a smooth bundle homomorphism  $\mathfrak{I}_\B$ from $\mathfrak{F}_\B$ into $\T_*\times\mathbb{C}^{n_\B}$ that is isometric on each fibre.
\end{corollary} 
\begin{corollary}\label{P-bd-triv}
	The application $\sigma\mapsto\mathfrak{I}_\B\circ\sigma$ defines an isometric map from $\mathfrak{S}^2(\T_*,\mathfrak{F}_\B)$ into $ L^2(\T_*;\mathbb{C}^{n_\B})$ and a continuous linear embedding $\mathfrak{S}(\T_*,\mathfrak{F}_\B)\hookrightarrow\big[C^\infty(\T_*)\big]^{n_\B}$.
\end{corollary}

\begin{proposition}\label{P-Pfr-Tstar-sect}
	There exists a family of $n_\B$ smooth global sections $\{\hat{\psi}_{p}\}_{1\leq p\leq n_\B}$ in the Bloch bundle $\mathfrak{F}_\B\repi\T_*$ such that the closed complex linear space they generate is equal to $\mathfrak{S}(\T_*,\mathfrak{F}_\B)$ and for any section $\sigma\in\mathfrak{S}^2(\T_*,\mathfrak{F}_\B)$ the following equalities hold:
	$$
	\sigma(\theta)\,=\,\underset{1\leq p\leq n_\B}{\sum}\,\big(\hat{\psi}_{p}(\theta)\,,\,\sigma(\theta)\big)_{\tilde{\p}_*^{-1}(\theta)}\hat{\psi}_{p}(\theta),\quad\forall\theta\in\T_*\,,
	$$
	with the identity
	$$ 
	\big\|\sigma(\theta)\big\|_{\tilde{\p}_*^{-1}(\theta)}^2\,=\,\underset{1\leq p\leq n_\B}{\sum}\,\big|\big(\hat{\psi}_{p}(\theta)\,,\,\sigma(\theta)\big)_{\tilde{\p}_*^{-1}(\theta)}\big|^2,\quad\forall\theta\in\T_*\,.
	$$
\end{proposition}
\begin{proof}
	We notice that the bundle homomorphism $\mathfrak{I}_\B:\mathfrak{F}_\B\rightarrow\T_*\times\mathbb{C}^{n_\B}$ is in fact a smooth family of linear isometries $\T_*\ni\theta\mapsto\mathfrak{I}_{\B,\theta}\in\mathcal{L}\big(\tilde{\p}_*^{-1}(\theta),\mathbb{C}^{n_\B}\big)$.
	For any ${\theta}\in\T_*$ let: 
	\beq\label{DF-LB-theta}
	\mathfrak{L}_\B({\theta}):=\big\{\mathfrak{I}_{\B,\theta}\mathfrak{v}\,,\, \mathfrak{v}\in\tilde{\p}_*^{-1}(\theta)\big\}\,\subset\,\mathbb{C}^{n_\B},
	\eeq
	so that $\mathfrak{I}_{\B,\theta}:\tilde{\p}_*^{-1}(\theta)\rightarrow\mathfrak{L}_\B({\theta})$ is invertible and let $Q_\B({\theta}):\mathbb{C}^{n_\B}\rightarrow\mathfrak{L}_\B({\theta})$ be the canonical orthogonal projection defined by this subspace.
	If $\big\{\mathcal{e}_1,\ldots,\mathcal{e}_{n_\B}\big\}$ is the canonical orthonormal basis of $\mathbb{C}^{n_\B}$, we can define the following global sections:
	\beq\label{DF-triv-sect}
	\hat{\psi}_{p}({\theta}):=\mathfrak{I}_{\B,\theta}^{-1}\big(Q_\B(\theta)\mathcal{e}_p\big),\quad\forall (p,{\theta})\in\underline{n_\B}\times\T_*.
	\eeq
	Given any $\sigma\in\mathfrak{S}(\mathfrak{F}_\B,\T_*)$ we can write that:
	\begin{align*}
		\sigma(\theta)&=\mathfrak{I}_{\B,\theta}^{-1}\big(Q_\B(\theta)\mathfrak{I}_{\B,\theta}\sigma(\theta)\big)=\mathfrak{I}_{\B,\theta}^{-1}\Big(Q_\B(\theta)\Big(\underset{1\leq p\leq n_\B}{\sum}\big(\mathfrak{I}_{\B,\theta}\sigma(\theta)\,,\,\mathcal{e}_p\big)_{\mathbb{C}^{n_\B}}\,\mathcal{e}_p\Big)\Big)\\
		&=\underset{1\leq p\leq n_\B}{\sum}\big(\mathfrak{I}_{\B,\theta}\sigma(\theta)\,,\,\mathcal{e}_p\big)_{\mathbb{C}^{n_\B}}\,\hat{\psi}_p(\theta)=\underset{1\leq p\leq n_\B}{\sum}\big(Q_\B(\theta)\mathfrak{I}_{\B,\theta}\sigma(\theta)\,,\,Q_\B(\theta)\mathcal{e}_p\big)_{\mathbb{C}^{n_\B}}\,\hat{\psi}_p(\theta)\\
		&=\underset{1\leq p\leq n_\B}{\sum}\big(\sigma(\theta)\,,\,\hat{\psi}_p(\theta)\big)_{\tilde{\p}_*^{-1}(\theta)}\,\hat{\psi}_p(\theta).
	\end{align*}
	In a similar way we notice that:
	\begin{align*}
		\big\|\sigma(\theta)\big\|_{\tilde{\p}_*^{-1}(\theta)}^2&=\big\|\mathfrak{I}_{\B,\theta}\sigma(\theta)\big\|_{\mathbb{C}^{n_\B}}^2=\underset{1\leq p\leq n_\B}{\sum}\big|\big(Q_\B(\theta)\mathfrak{I}_{\B,\theta}\sigma(\theta),\mathcal{e}_p\big)_{\mathbb{C}^{n_\B}}\big|^2\\
		&=\underset{1\leq p\leq n_\B}{\sum}\big|\big(\sigma(\theta),\hat{\psi}_p\big)_{\tilde{\p}_*^{-1}(\theta)}\big|^2.
	\end{align*}
\end{proof}

Let us {map} our $n_\B$ global smooth sections obtained in Proposition \ref{P-Pfr-Tstar-sect} {into $n_\B$} rapidly decaying {smooth} functions in $L^2(\X)$ by the inverse Bloch-Floquet transform:
\beq\label{F-PN-1}
\psi_{p}(\hat{x}+\gamma)\,:=\,\U_{BF}^{-1}\hat{\psi}_p\,=\,\int_{\T_*}d{\theta}\,e^{i<{\theta},\gamma>}\,\hat{\psi}_{p}(\hat{x};{\theta}).
\eeq
They define a well-localized finite-dimensional complex subspace of the isolated Bloch family subspace $P_\B\,L^2(\X)$. Moreover, for any $\alpha\in\Gamma$, the translated functions $\tau_\alpha\psi_{p}$ verify the equalities:
$$ 
\big(\U_{BF}\tau_\alpha\psi_{p}\big)(\hat{x};{\theta})=\underset{\gamma\in\Gamma}{\sum}e^{-i<{\theta},\gamma>}\,\psi_{p}(\hat{x}+\gamma+\alpha)=e^{i<{\theta},\alpha>}\hat{\psi}_{p}(\hat{x};{\theta})\in\big[\widehat{P}_\B({\theta})\,\mathscr{F}_{\theta}\big]
$$
and we conclude that:
\beq\label{DF-freePfr}
\blPsi_{\B}:=\Big\{\tau_\alpha\,\psi_{p},\ \alpha\in\Gamma,\ 1\leq p\leq n_\B\Big\}\,\subset\,P_\B\,L^2(\X).
\eeq

\begin{proposition}\label{P-band-frame}
	The family $\blPsi_{\B}$ defines a Parseval frame for the subspace $P_\B\,L^2(\X)$ associated with  the isolated Bloch family $\B$.
\end{proposition}
\begin{proof}
	We prove first that $P_\B\,L^2(\X)$ is the linear span of the family $\blPsi_{\B}$. Thus let us choose any $f\in P_\B\,L^2(\X)$ and consider its Bloch-Floquet transform:
	\[  
	\hat{f}_{\theta}(\hat{x}):=\big(\U_{BF}f\big)(\hat{x},{\theta})=\big(\U_{BF}P_\B f\big)(\hat{x},{\theta})=\hat{P}_\B({\theta})\big[\hat{f}_{\theta}(\hat{x})\big]
	\]
	that defines a bounded measurable global section $\T_*\ni\theta\mapsto\hat{f}_{\theta}\in L^2(\mathcal{E})\simeq\tilde{\p}_*^{-1}(\theta)$ and thus an element in $\mathfrak{S}^2(\mathfrak{F}_\B;\T_*)$.
	Using Proposition \ref{P-Pfr-Tstar-sect} we can write that:
	$$
	\hat{f}_{\theta}=\underset{1\leq p\leq n_\B}{\sum}\,\big(\hat{\psi}_{p}({\theta})\,,\,\hat{f}_{\theta}\big)_{\mathscr{F}_{\theta}}\hat{\psi}_{p}({\theta})
	$$
	and consequently:
	\begin{align*}
		f(\hat{x}+\gamma)\,&=\,\big(\U_{BF}^{-1}\hat{f}\big)(\hat{x}+\gamma)=\int_{\T_*}d{\theta}\,e^{i<{\theta},\gamma>}\,\hat{f}_{\theta}(\hat{x})\\
		&=\underset{1\leq p\leq n_\B}{\sum}\,\int_{\T_*}d{\theta}\,e^{i<{\theta},\gamma>}\,\big(\hat{\psi}_{p}({\theta})\,,\,\hat{f}_{\theta}\big)_{\tp_*^{-1}(\theta)}\hat{\psi}_{p}(\hat{x},{\theta})\\
		&=\,\underset{1\leq p\leq n_\B}{\sum}\,\underset{\alpha\in\Gamma}{\sum}\,\big[\mathcal{F}_{\T_*}\big(\hat{\psi}_{p}(\cdot)\,,\,\hat{f}\big)_{\tp_*^{-1}(\cdot)}\big](\alpha)\,\psi_{p}(\hat{x}+\gamma-\alpha).
	\end{align*}
	For $(\alpha,p)\in\Gamma\times\underline{n_\B}$ we introduce the notation:
	$$ 
\mathfrak{C}_\B(f)_{\alpha,p}:=\big[\mathcal{F}_{\T_*}\big(\hat{\psi}_{p}(\cdot)\,,\,\hat{f}\big)_{\tp_*^{-1}(\cdot)}\big](\alpha)\,,
	$$
	and conclude that:
	\begin{equation}\label{pars1}
	f=\underset{1\leq p\leq n_\B}{\sum}\,\underset{\alpha\in\Gamma}{\sum}\,\mathfrak{C}_\B(f)_{\alpha,p}\,[\tau_{-\alpha}\psi_p].
	\end{equation}
	
	In order to prove that it is a Parseval frame let us compute the scalar products $\big(\tau_{-\alpha}\psi_{p}\,,\,f\big)_{L^2(\X)}$ for some $f\in L^2(\X)$ and some indices $(\alpha,p)\in\Gamma\times\underline{n_\B}$:
	\begin{align*}
		\big(\tau_{-\alpha}\psi_{p},f\big)_{L^2(\X)}&=\underset{\gamma\in\Gamma}{\sum}\int_{\mathcal{E}}d\hat{x}\,f(\hat{x}+\gamma)\,\overline{\psi_{p}(\hat{x}+\gamma-\alpha)}\\& =\int_{\T_*}d{\theta}\,e^{i<{\theta},\alpha>}\int_{\mathcal{E}}d\hat{x}\,\hat{f}_{\theta}(\hat{x})\,\overline{\hat{\psi}_{p}(\hat{x},{\theta})}\\
		&=\int_{\T_*}d{\theta}\,e^{i<{\theta},\alpha>}\,\big(\hat{\psi}_{p}({\theta})\,,\,\hat{f}_{\theta}\big)_{\tp_*^{-1}(\theta)}=\mathfrak{C}_\B(f)_{\alpha,p}\,.
	\end{align*}
    Using this in \eqref{pars1} we finally obtain
\begin{equation*}
	f=\underset{1\leq p\leq n_\B}{\sum}\,\underset{\alpha\in\Gamma}{\sum}\,\big(\tau_{-\alpha}\psi_{p},f\big)_{L^2(\X)}\,[\tau_{-\alpha}\psi_p]\,.
	\end{equation*}
\end{proof}

Using the rapid decay of the functions in the Parseval frame $\blPsi_{\B}$, the inclusion $\mathfrak{F}_{\B}\subset\mathfrak{F}^\infty$ and the Cotlar-Stein procedure one concludes from the Proposition \ref{P-band-frame} that:
\begin{corollary}\label{C-PN-2} $P_{\B}$ has an integral kernel $\mathfrak{K}(P_{\B})\in\mathring{\mathscr{S}}(\X\times\X)$ (see Notation \ref{N-per-distr}) given by the series:
\[
\mathfrak{K}(P_{\B})(x,y):=\underset{\gamma\in\Gamma}{\sum}\,\underset{p\in\underline{\nB}}{\sum}(\tau_{-\gamma}\psi_p)(x)\,\overline{(\tau_{-\gamma}\psi_p)(y)}\,,
\]
with uniform convergence on compact sets in $\X\times\X$. Moreover, the limit in Proposition~\ref{C-wlim-id} exists in the strong operator topology on $L^2(\X)$.
\end{corollary}

We notice that the map $\mathfrak{C}_\B$ appearing in the above proof is in fact the coordinate map associated with the Parseval frame $\blPsi_{\B}$, as in Remark \ref{R-A-01}:
\[\nonumber 
\mathfrak{C}_\B:P_\B\,L^2(\X)\ni\,f\,\mapsto\,\big((\tau_{\alpha}\psi_{p}\,,\,f)_{L^2(\X)}\big)_{(\alpha,p)\in\Gamma\times\underline{n_\B}}\,\in\,[\ell^2(\Gamma)]^{n_\B}\cong\ell^2(\Gamma)\otimes\Co^{n_\B}.
\]

\begin{corollary}
The map $\mathfrak{W}_{\B}:\mathbb{B}\big(P_{\B}L^2(\X)\big)\ni\,T\mapsto\mathfrak{C}_{\B}\,T\,\mathfrak{C}_{\B}^*\in\mathbb{B}\big(\ell^2(\Gamma)\otimes\Co^{\nB}\big)$ is a $*$-homomorphism of $C^*$-algebras (in general not surjective).
\end{corollary}
\begin{notation}
Given any $T\in\mathbb{B}\big(P_{\B}L^2(\X)\big)$ and the canonical basis $\big\{{{\cal{e}}_\alpha\otimes{\cal{e}}_p}\big\}_{(\alpha,p)\in\Gamma\times\underline{n_\B}}$ of $\ell^2(\Gamma)\otimes\Co^{\nB}$ we shall use the notation $\mathfrak{M}_{\B}[T]\in\mathscr{M}^\circ_\Gamma[\MmN]$ for the infinite matrix:
\beq\label{C-PN-VI-3}
\big[\mathfrak{M}_{\B}[T]_{\alpha,\beta}\big]_{p,q}\,:=\,\big({\cal{e}}_\alpha\otimes{\cal{e}}_p\,,\,\mathfrak{W}_{\B}[T]\,{\cal{e}}_\beta\otimes{\cal{e}}_q\big)_{\ell^2(\Gamma)\otimes\Co^{n_\B}}\,=\,\big((\tau_{-\alpha}\psi_p)\,,\,T\,(\tau_{-\beta}\psi_q)\big)_{L^2(\X)}.
\eeq
\end{notation}

When the operator $T\in\mathbb{B}\big(P_{\B}L^2(\X)\big)$ has a distribution kernel $\mathfrak{K}(T)\in\mathscr{S}^\prime(\X\times\X)$, we can write:
\[\nonumber
\big[\mathfrak{M}_{\B}[T]_{\alpha,\beta}\big]_{p,q}=\big((\tau_{-\alpha}\psi_p)\,,\,(\Int\,\mathfrak{K}(T))\,(\tau_{-\beta}\psi_q)\big)_{L^2(\X)}=\big\langle\,\mathfrak{K}(T)\,,\,\overline{(\tau_{-\alpha}\psi_p)}\otimes(\tau_{-\beta}\psi_q)\,\big\rangle_{\X\times\X}.
\]
When $T$ commutes with the translations with vectors from $\Gamma$, its distribution kernel satisfies the identities: $$(\tau_{\gamma}\otimes\tau_{\gamma})\mathfrak{K}(T)=\mathfrak{K}(T) \mbox{ for any } \gamma\in\Gamma\,,$$ and thus the following equalities prove that there exists $\mathring{\mathfrak{m}}_{\B}[T]\in{\cal{s}}(\Gamma;\MmN)$ such that:

\begin{align*}
	\big[\mathfrak{M}_{\B}[T]_{\alpha,\beta}\big]_{p,q}&=\big((\tau_{-\alpha}\psi_p)\,,\,(\Int\,\mathfrak{K}(T))\,(\tau_{-\beta}\psi_q)\big)_{L^2(\X)}\\
	&=\big(\psi_p\,,\,\Int\big[(\tau_{\alpha-\beta}\otimes\bb1)\mathfrak{K}(T)\big]\psi_q\big)_{L^2(\X)}=:\big [\mathring{\mathfrak{m}}_{\B}[T]_{\alpha-\beta}\big ]_{p,q}.
\end{align*}

Taking now $T=H_{\B}\in\mathbb{B}\big(P_{\B}L^2(\X)\big)$ we obtain the following infinite matrix with rapid off-diagonal decay:
\begin{align}\label{F-PN-V-10}
	&\big[\mathfrak{M}_{\B}[H_{\B}]_{\alpha,\beta}\big]_{p,q}=\big((\tau_{-\alpha}\psi_p)\,,\,\big(\Int\,\mathfrak{K}(\widehat{H}_{\B})\big)\,(\tau_{-\beta}\psi_q)\big)_{L^2(\X)}\\ \nonumber
	&=\int_{\T_*}d\theta\,e^{i<\theta,\alpha-\beta>}\,\Big(\underset{k_0\leq k\leq k_0+N}{\sum}\lambda_k(\theta)\big(\hat{\psi}_q(\theta),\hat{\pi}_k(\theta)\hat{\psi}_p(\theta)\big)_{\mathscr{F}_{\theta}}\Big)\,,
\end{align}
and the following rapidly decaying sequence (as stated in point 3.\ref{point3b} in Theorem \ref{T-I}):
\begin{align}\label{F-PN-2}
	\big[[\mathring{\mathfrak{m}}_{\B}[\widehat{H}_{\B}]_{\gamma}\big]_{p,q}\,&=\,\int_{\T_*}d\theta\,e^{i<\theta,\gamma>}\,\Big(\underset{k_0\leq k\leq k_0+N}{\sum}\lambda_k(\theta)\big(\hat{\psi}_q(\theta),\hat{\pi}_k(\theta)\hat{\psi}_p(\theta)\big)_{\mathscr{F}_{\theta}}\Big)\\ \nonumber
	&=\,\int_{\T_*}d\theta\,e^{i<\theta,\gamma>}\,\big(\hat{\psi}_q(\theta),\widehat{H}_{\B}(\theta)\hat{\psi}_p(\theta)\big)_{\mathscr{F}_{\theta}}.
\end{align}
If we denote by $\mathring{U}(\gamma)$ the restriction to $\ell^2(\Gamma)$ of the translation with $-\gamma\in\Gamma$ and by:
\beq\label{F-PN-3}
\widehat{\mm}_{\B}[\widehat{H}_{\B}](\theta):=\big(\hat{\psi}_q(\theta),\widehat{H}_{\B}(\theta)\hat{\psi}_p(\theta)\big)_{\mathscr{F}_{\theta}}
\eeq
the last equality in \eqref{F-PN-2} allows us to view the infinite matrix \eqref{F-PN-V-10} as associated with  the following operator acting in $\ell^2(\Gamma)\otimes\Co^{\nB}$:
\beq\label{F-f-4}
\tOp\big(\widehat{\mm}_{\B}[\widehat{H}_{\B}]\big)\,:=\,\underset{\gamma\in\Gamma}{\sum}\mathring{U}(-\gamma)\otimes\Big[\mathcal{F}_{\T_*}\big(\widehat{\mm}_{\B}[\widehat{H}_{\B}](\theta)\big)\Big](\gamma)\,\in\,\mathbb{B}\big(\ell^2(\Gamma)\big)\otimes_H\mathscr{M}(\nB)\,,
\end{equation}
where $\otimes_H$ is the so-called \textit{spatial tensor product} of the two $C^*$-algebras associated with  their canonical representations on the Hilbert spaces $\ell^2(\Gamma)$ and $\Co^{\nB}$. \\
One may compare the formula \eqref{F-f-4} with the Weyl quantization of a $\Gamma_*$-periodic symbol that does not depend on the configuration variables $x\in\X$:
\beq
\Op(F)=\int_{\X}dz\big(\mathcal{F}_{\X^*}F\big)(z)U(-z)=\underset{\gamma\in\Gamma}{\sum}\big(\mathcal{F}_{\T_*}F\big)(\gamma)U(-\gamma).
\eeq
In conclusion, we may view $\mM_{\B}[H_{\B}]$ as a quantization of $\widehat{\mm}_{\B}[\widehat{H}_{\B}]\circ\p_*:\X^*\rightarrow\mathscr{M}(\nB)$ through a Weyl calculus on the lattice $\Gamma$ with matrix valued symbols.

\subsection{Weyl symbols for the non perturbed decomposition}\label{SS-3.5}

We close this section with an analysis of the symbols of some $\B$-operators considered as pseudo-differential operators. We shall start from their  distribution kernels defined via the Parseval frame \eqref{DF-freePfr}.

\begin{proposition}\label{R-p-symb} There exists $p_\B$ in $S^{-\infty}(\Xi)_\Gamma$ such that $P_\B\equiv\Op(p_{\B})$.
\end{proposition}
\begin{proof}
	In order to obtain the explicit relation between the $\Gamma$-periodic tempered distribution $p_{\B}$ (such that $P_\B\equiv\Op(p_{\B})$) and the integral kernel $\mathfrak{K}[P_\B]$ (in the above Corollary \ref{C-PN-2}) we notice that:
	\begin{align*}
P_{\B}=\Op(p_{\B})&=\Int\,\mathfrak{K}[P_{\B}]
	\end{align*}
	and for $A=0$, Equation \eqref{F-KerOpA} and the formula of $\Upsilon\in\mathbb{B}(\R^{2d})$ in Remark \ref{R-KerOpA} imply  that: 
	\begin{align}\label{F-p-symb}
		&p_{\B}(z,\zeta)=\big(\bb1\otimes\mathcal{F}_{\X}\big)\big(\mathfrak{K}[P_{\B}])\circ\Upsilon^{-1}\big)(z,\zeta)=\\
		\nonumber
		&=\hspace*{-0.2cm}\int_{\X}\hspace*{-0.2cm}dv\,e^{i<\zeta,v>}\,\Big(\underset{\gamma\in\Gamma}{\sum}\,\underset{p\in\underline{\nB}}{\sum}\big(\tau_{-\gamma}\psi_p(z+v/2)\big)\,\overline{\big(\tau_{-\gamma}\psi_p(z-v/2)\big)}\Big).
	\end{align}
\end{proof}

The definitions given in Subsection \ref{SS-free-dyn} imply the equality $H_\B=P_\B H^\circ P_\B$, and the fact that $H_{\B}$ has a symbol $h_\B$ in $S^{-\infty}(\Xi)_\Gamma$, $H_\bot:=H^\circ-H_{\B}$ has a symbol $h_\bot\in{S}^p_1(\Xi)_\Gamma$ and we have the formulas:
\beq\label{F-f-5}
H_{\B}=\Op(h_{\B}),\quad\,h_\B=p_\B\sharp{h}^\circ\sharp{p}_\B,\quad\,H_\bot=\Op(h_\bot),\quad\,h_\bot:=h^\circ-h_{\B}\,.
\eeq
Due to  Hypothesis \ref{H-E0} (that amounts to a simple shift of the energy origin) the interval $[0,E_0]\subset\R$ is always non-void and it is a spectral gap for 
$H_\bot$. In what follows, we will need the notation introduced in \eqref{dec-id}, \eqref{eq:3.1}, \eqref{eq:3.2} and \eqref{H-infinit}. 

The operator $P_\B$ coincides with the orthogonal projection $P_{\ker H_\bot}$ on the kernel of $H_\bot$. 
Since the interval $(E_-,E_+)$ is a spectral gap for $H_\bot$, we consider a contour $\mathscr{C}_0$ which does not contain $0$ in its interior but contains the interval $[E_0,(E_-+E_+)/2]$, and stays at a positive distance from the spectrum of $H_\bot$. Then we have the usual formulas:
\beq \label{DF-p0}
\big(\Op(h_\bot)-\zz\bb1\big)^{-1}=\Op\big((h_\bot-\zz)^-\big),\quad{p}_0=-(2\pi i)^{-1}\int_{\mathscr{C}_0}d\zz\,\big(h_\bot-\zz\big)^-,\quad\,h_0=p_0\sharp{h^\circ}\sharp{p}_0\,,
\eeq
implying that $P_0$ and $H_0$ have symbols of class $S^{-\infty}(\Xi)_{\Gamma}$.\\
Finally we also have the equalities 
\beq \label{DF-pa} p_\infty=1-p_0-p_\B \mbox{ and  } h_\infty=p_\infty\sharp{h}^\circ\sharp{p}_\infty=1-h_0-h_\B\,,
\eeq
implying that $H_\infty=H^\circ-H_{\B}-H_0$ and $P_{\infty}$ have symbols $h_\infty$ in $S^p_1(\Xi)_\Gamma$  and $p_\infty$ in $S^0_1(\Xi)_\Gamma$.

\section{The perturbed decomposition induced by $\B$.}
\label{S-3}

Our main purpose now is to use some kind of local gauge choices inspired from \cite{Lu} and define a family of ``modified magnetic translations" (see also \cite{Z}) generalizing the \textit{Zak translations} (from Definition \ref{D-Ztr-cB}), in order to obtain a ``magnetic" version of the family $\blPsi_{\B}$ in \eqref{DF-freePfr} and through it the desired projection $P^\epsilon_{\B}$ in Theorem \ref{T-I} as a kind of $\epsilon$-dependent ``magnetic version" of the orthogonal projection $P_{\B}$. An important technical difficulty is to elaborate a procedure replacing the usual Gram-Schmidt orthonormalization method in the case of frames. We begin with some preliminary technical results concerning the decomposition of the Hamiltonian $H^\epsilon$ extending that of $H^\circ$.

\subsection{The perturbed Hamiltonian decomposition}
\label{SS-4.1}

The symbol decomposition $h^\circ=h_0+h_{\B}+h_\infty$ discussed above and the properties of magnetic quantization allow us to consider a similar operator decomposition when the perturbing magnetic field given in \eqref{HF-Beps} is added (we use the notations in \eqref{N-A-eps}):
\beq \label{F-descOepscch}
\Op^{\epsilon}(h^\circ)\,=\,\Op^{\epsilon}(h_0)+\Op^{\epsilon}(h_\B)+\Op^{\epsilon}(h_\infty).
\eeq

We shall use the notations $H^{\epsilon}_0\equiv\Op^{\epsilon}(h_0)$,  $H^{\epsilon}_{\B}\equiv\Op^{\epsilon}(h_\B)$ and $H^{\epsilon}_\infty$ for the self-adjoint closure of $\Op^{\epsilon}(h_\infty)$.
We also consider  the lower-semibounded self-adjoint operator $$ H^{\epsilon}_\bot:=H_0^{\epsilon}+H^{\epsilon}_\infty$$ equal to the closure of $\Op^{\epsilon}(h_\bot)$ with symbol $h_\bot\in S^p_1(\Xi)_\Gamma$ given in \eqref{F-f-5}.

An important aspect is that these {perturbed} Hamiltonians no longer commute among themselves, as it was the case in Subsection \ref{SS-free-dyn}. Another important observation, based on the results in \cite{AMP,CP-1} is that their spectra, considered as subsets of $\R$, do not change much as subsets of $\R$ and we have the following result. 

\begin{remark}
	Under the hypothesis of Theorem \ref{T-I}, there exists $\delta_0 >0$ and for any $\delta \in (0,\delta_0]$ there exists $\epsilon_0>0$ such that for any $\epsilon\in [0,\epsilon_0]$ we have the inclusions:
	\begin{equation}\begin{split}\label{F-m-sp-est}
	&\sigma\big(H_0^{\epsilon}\big)\subset[-\delta\,,\,\delta]\cup[E_0-\delta,E_-+\delta],\quad 
	\sigma\big(H_\B^{\epsilon}\big)\subset[-\delta\,,\,\delta]\cup[E'_--\delta,E'_++\delta],\\
		&\sigma\big(H^{\epsilon}_\infty\big)\subset[-\delta\,,\,\delta]\cup[E_+-\delta,+\infty),\\ 
		&\sigma\big(H^{\epsilon}_\bot\big)\subset[-\delta\,,\,\delta]\cup[E_0-\delta,E_-+\delta]\cup[E_+-\delta,+\infty).
	\end{split}
    \end{equation}
\end{remark}

Actually, much more precise regularity estimates {for the spectral edges seen as functions of $\epsilon$} {were} obtained {in} \cite{CP-2}{, but they are not needed here.}

\subsection{Magnetic perturbation of the Parseval frame.}\label{SS-m-P-frame}

Recalling \eqref{N-A-eps}, we introduce the following notations:
\beq \label{dhc5}
\Lambda^{\epsilon}_{\gamma}(x):=\tLambda^{\epsilon}(x,\gamma),\quad\forall\gamma\in\Gamma,\qquad\tOmega^\epsilon(x,y,z):=\tLambda^{\epsilon}(x,y)\tLambda^{\epsilon}(y,z)\tLambda^{\epsilon}(z,x).
\eeq

\subsubsection{The inhomogeneous Zak translations and the perturbed Parseval frame}\label{SSS-iHom-Z-trsl}

Guided by the same intuition as for the construction of the modified Wannier functions in {\cite{Ne-RMP,CHP-1} } and by the analysis in \cite{CMM} and using the notations \eqref{dhc5}, let us  introduce the following system of functions:
\beq \label{DF-m-W-f}
\widetilde{\blPsi}^{\epsilon}_{\B}:=\big\{\tpsi^\epsilon_{\gamma,p}:=\Lambda^{\epsilon}_\gamma\big(\tau_{-\gamma}\psi_p\big)\equiv \mathfrak{T}^{\epsilon}_\gamma (\psi_p),\,\gamma\in\Gamma,\,p\in\underline{n_\B}\big\}\in L^2(\X)
\eeq 
and let $\widetilde{\mathfrak{L}}^{\epsilon}_{\B}$ be the closed subspace that they generate. 

Having in mind the arguments and constructions in Subsection 3.4 in \cite{CIP}, let us consider for any $N\in\mathbb{N}\setminus\{0\}$ the following integral kernels of finite-rank integral operators  (with the notations defined after \eqref{DF-Int}):
\begin{align}\label{F-I-0}
	\mathfrak{K}^{\epsilon}_{\B,N}(x,y)\,:=\,\hspace*{-0.4cm}\underset{|\gamma|\leq\,N,\,p\in\underline{n_\B}}{\sum}\hspace*{-0.3cm}\mathfrak{T}^{\epsilon}_\gamma\psi_p(x)\,\overline{\mathfrak{T}^{\epsilon}_\gamma\psi_p(y)}\,\equiv\,\underset{|\gamma|\leq\,N}{\sum}\mathfrak{K}^{\epsilon}_{\B,\gamma}(x,y)
\end{align}
and notice that $\big[\widetilde{\mathfrak{L}}^{\epsilon}_{\B}\big]^{\bot}\subset\Ker\big[\Int\,\mathfrak{K}^{\epsilon}_{\B,N}\big]$ and $\Rge\big[\Int\,\mathfrak{K}^{\epsilon}_{\B,N}\big]\subset\widetilde{\mathfrak{L}}^{\epsilon}_{\B}$ for any $N\in\mathbb{N}\setminus\{0\}$. 

\begin{lemma}
With the notations in \eqref{dhc5}, given any $n\in\mathbb{N}$ there exists $C_n>0$ such that for any $\epsilon\in[0,\epsilon_0]$ and any $(\gamma,p)\in\Gamma\times\underline{n_{\B}}$, we have:
	\[
	\forall(\hat{x},\hat{y})\in\mathcal{E}\times\mathcal{E}:\quad\big|\mathfrak{T}^{\epsilon}_\gamma\psi_p(\hat{x}+\alpha)\,\overline{\mathfrak{T}^{\epsilon}_\gamma\psi_p(\hat{y}+\beta)}\big|\,\leq\,C_n<\alpha-\beta>^{-n},\qquad\forall(\alpha,\beta)\in\Gamma\times\Gamma.
	\]
\end{lemma}
\begin{proof}
Using \eqref{DF-triv-sect} and \eqref{F-PN-1} we can write:
\begin{align*}
	&\big|\mathfrak{T}^{\epsilon}_\gamma\psi_p(\hat{x}+\alpha)\,\overline{\mathfrak{T}^{\epsilon}_\gamma\psi_p(\hat{y}+\beta)}\big|=\\ 
	&\ =(2\pi)^{-2d}\Big|\Lambda^\epsilon(\hat{x}+\alpha,\gamma)\overline{\Lambda^\epsilon(\hat{y}+\beta,\gamma)}\Big(\int_{\T_*}d{\theta}\,e^{i<{\theta},\alpha-\gamma>}\,\hat{\psi}_{p}(\hat{x};{\theta})\Big)\Big(\int_{\T_*}d{\tilde{\bz}^*}\,e^{-i<{\tilde{\bz}^*},\beta-\gamma>}\,\overline{\hat{\psi}_{p}(\hat{y};{\tilde{\bz}^*})}\Big)\Big|\\
	&\leq\,C^\prime_n<\alpha-\gamma>^{-n}<\beta-\gamma>^{-n}\big(\underset{\theta\in\T_*}{\sup}\,\underset{\hat{x}\in\E}{\sup}\,\big|\big(\partial_{\theta}^n\hat{\psi}_{p}\big)(\hat{x};\theta)\big|\big)^2\,\leq\,{C}_n<\alpha-\beta>^{-n}.
\end{align*}
\end{proof}
\begin{proposition}\label{C-off-d-decay}
For any $N\in\mathbb{N}$, the integral kernel $\mathfrak{K}^\epsilon_{\B,N}$ has rapid off-diagonal decay and defines a bounded self-adjoint integral operator.
\end{proposition}

In order to control the limit for $N\nearrow\infty$ we shall use the Cotlar-Stein Lemma (see for example Lemma 18.6.5 in \cite{H-3}) and thus we need to estimate the products:
\[\begin{split}
\big[\Int\big(\mathfrak{K}^{\epsilon}_{\B,\alpha}\big)\big]^*\cdot\Int\big(\mathfrak{K}^{\epsilon}_{\B,\beta}\big)\,&=\,\Int\big[\big(\mathfrak{K}^{\epsilon}_{\B,\alpha}\big)^\dagger\diamond\big(\mathfrak{K}^{\epsilon}_{\B,\beta}\big)\big]\,,\\
\Int\big(\mathfrak{K}^{\epsilon}_{\B,\alpha}\big)\cdot\Int\big[\big(\mathfrak{K}^{\epsilon}_{\B,\beta}\big)\big]^*\,&=\,\Int\big[\big(\mathfrak{K}^{\epsilon}_{\B,\alpha}\big)\diamond\big(\mathfrak{K}^{\epsilon}_{\B,\beta}\big)^\dagger\big]\,.
\end{split}\]
\begin{align*}
&\big[\big(\mathfrak{K}^{\epsilon}_{\B,\alpha}\big)^\dagger\diamond\big(\mathfrak{K}^{\epsilon}_{\B,\beta}\big)\big]=\big[\big(\mathfrak{K}^{\epsilon}_{\B,\alpha}\big)\diamond\big(\mathfrak{K}^{\epsilon}_{\B,\beta}\big)\big]=\big[\big(\mathfrak{K}^{\epsilon}_{\B,\alpha}\big)\diamond\big(\mathfrak{K}^{\epsilon}_{\B,\beta}\big)^\dagger\big]\,;\\
&\big[\big(\mathfrak{K}^{\epsilon}_{\B,\alpha}\big)\diamond\big(\mathfrak{K}^{\epsilon}_{\B,\beta}\big)^\dagger\big](x,y)=\\
&\quad=\underset{p\in\underline{n_\B}}{\sum}\ \underset{q\in\underline{n_\B}}{\sum}\int_{\X}dz\,\Big(\mathfrak{T}^{\epsilon}_\alpha\psi_p(x)\,\overline{(\mathfrak{T}^{\epsilon}_\alpha\psi_p)(z)}\Big((\mathfrak{T}^{\epsilon}_\beta\psi)_q(z)\Big)\,\overline{\mathfrak{T}^{\epsilon}_\beta\psi_q(y)}\Big)\,.
\end{align*}

The following computations prove that for any $n\in\mathbb{N}$ there exists a constant $C_n>0$ such that, for any $\epsilon\in[0,\epsilon_0]$, we have:
\begin{align*}
&\int_{\X}dz\,\overline{(\mathfrak{T}^{\epsilon}_\alpha\psi_p)(z)}\,(\mathfrak{T}^{\epsilon}_\beta\psi)_q(z)=\int_{\X}dz\,\overline{\Lambda^{\epsilon}_\alpha\psi_p(z-\alpha)}\,\Lambda^{\epsilon}_\beta\psi_q(z-\beta)\\\nonumber
&\quad=\int_{\X}dz\,\Big(<z-\alpha>^{n}\overline{\Lambda^{\epsilon}_\alpha\psi_p(z-\alpha)}\Big)\,\Big(<z-\beta>^n\Lambda^{\epsilon}_\beta\psi_q(z-\beta)\Big)<z-\alpha>^{-n}<z-\beta>^{-n}\\
&\quad\leq\,C_n<\alpha-\beta>^{-n}\underset{p\in\underline{\nB}}{\max}[\lnu_{n+(d+1)/2,0}(\psi_p)]^2.
\end{align*}
Via the Schur test and Proposition \ref{C-off-d-decay}, the sums over $(p,q)$ being finite, the above estimation gives the necessary condition for applying the Cotlar-Stein Lemma and obtaining the following result.
\begin{proposition}\label{P-II-1}
The limit $\underset{N\nearrow\infty}{\lim}\Int\,\mathfrak{K}^{\epsilon}_{\B,N}\,=:\,\widetilde{Q}^{\epsilon}_\B\in\mathbb{B}\big(L^2(\X)\big)$ exists in the strong operator topology on $\mathbb{B}\big(L^2(\X)\big)$. Moreover, the limit $\underset{N\nearrow\infty}{\lim}\mathfrak{K}^{\epsilon}_{\B,N}=:\mathfrak{K}^{\epsilon}_{\B}$ also exists uniformly on compact sets in $BC^\infty(\X\times\X)$, it belongs to $\mathring{\mathscr{S}}_{\Delta}(\X\times\X)$ (see Notation \ref{N-per-distr}) and  $\mathfrak{K}^\epsilon_\B$ equals the distribution kernel of $\widetilde{Q}^{\epsilon}_\B$.
\end{proposition}
An important technical difficulty is that in general the operator $\widetilde{Q}^\epsilon_{\B}$ is not the orthogonal projection on $\widetilde{\mathfrak{L}}^\epsilon_{\B}$. In fact we have the following result.

\begin{proposition}\label{P-II-2}
The operator $\widetilde{Q}^{\epsilon}_\B:=\Int\,\mathfrak{K}^{\epsilon}_{\B}$ is self-adjoint and satisfies the following identity (in $\mathbb{B}\big(L^2(\X)\big)$): $$[\widetilde{Q}^{\epsilon}_\B]^2=\widetilde{Q}^{\epsilon}_\B\,+\,\mathscr{O}(\epsilon).$$
\end{proposition}
\begin{proof}
Let us fix some $N\in\mathbb{N}$ and compute:
\beq \begin{aligned}\label{F-II-0}
&\big[(\mathfrak{K}^{\epsilon}_{\B,N})\diamond(\mathfrak{K}^{\epsilon}_{\B,N})\big](x,y)=\\
&\qquad=\underset{|\alpha|\leq\,N,\,p\in\underline{n_\B}}{\sum}\ \underset{|\beta|\leq\,N,\,q\in\underline{n_\B}}{\sum}\int_{\X}dz\,\mathfrak{T}^{\epsilon}_\alpha\psi_p(x)\,\overline{\mathfrak{T}^{\epsilon}_\alpha\psi_p(z)}\,\mathfrak{T}^{\epsilon}_\beta\psi_q(z)\,\overline{\mathfrak{T}^{\epsilon}_\beta\psi_q(y)}.
\end{aligned}
\eeq

Let us make a deeper study of the $z$-integral appearing in the above sums, following the analysis in Appendix \ref{SSS-PN-VI-2}:

\begin{align}\nonumber
\int_{\X}dz\,\overline{\mathfrak{T}^{\epsilon}_\alpha\psi_p(z)}\,\mathfrak{T}^{\epsilon}_\beta\psi_q(z)&=\int_{\X}dz\,\overline{\Lambda^{\epsilon}_\alpha(z)\psi_p(z-\alpha)}\,\Lambda^{\epsilon}_\beta(z)\psi_q(z-\beta)\\ \nonumber
&=\tLambda^{\epsilon}(\alpha,\beta)\int_{\X}dz\,\tOmega^{\epsilon}(\alpha,z,\beta)\overline{\psi_p(z-\alpha)}\,\psi_q(z-\beta)\\ \nonumber
&=\tLambda^{\epsilon}(\alpha,\beta)\big((\tau_{-\alpha}\psi_p)\,,\,\bb1\,(\tau_{-\beta}\psi_q)\big)_{L^2(\X)}\,+\,\epsilon\,\big[[\tilde{\mathfrak{m}}^{\epsilon}_{\B,1}]_{\alpha,\beta}\big]_{p,q}\\ \label{F-II-1}
&=\tLambda^{\epsilon}(\alpha,\beta)\big[\mathfrak{M}_{\B}[\bb1]_{\alpha,\beta}\big]_{p,q}\,+\,\epsilon\,\big[[\tilde{\mathfrak{m}}^{\epsilon}_{\B,1}]_{\alpha,\beta}\big]_{p,q}\,,
\end{align}
where we used \eqref{C-PN-VI-3} for the main term and put into evidence the remainder:
\begin{align}\label{F-PN-VII-2}
&\big[[\tilde{\mathfrak{m}}^{\epsilon}_{\B,1}]_{\alpha,\beta}\big]_{p,q}:=\\ \nonumber
&\ =-i\tLambda^{\epsilon}(\alpha,\beta)\int_{\X}dz\,\Big(\int_{<\alpha,z,\beta>}\,B^\epsilon\Big)\,\int_0^1ds\exp\Big(-is\epsilon\int_{<\alpha,z,\beta>}\,B^\epsilon\Big)\overline{\psi_p(z-\alpha)}\,\psi_q(z-\beta), 
\end{align}
with the following estimation implied by the arguments in Appendix \ref{SSS-PN-VI-2} (making use of Notation \ref{N-B}):
\[
\begin{split}
&\forall N\in\mathbb{N},\ \exists\,C_N(\BB)>0:\quad\big|\big[[\tilde{\mathfrak{m}}^{\epsilon}_{\B,1}]_{\alpha,\beta}\big]_{p,q}\big|\leq\\
&\ \leq\,{C_N(\BB)} \,<\alpha-\beta>^{-N}\,\big[\epsilon\,\big(\underset{p\in\underline{\nB}}{\max}\lnu_{N+d+2,0}(\psi_p)^2\big)\,+\,\mathscr{O}(\epsilon^2)\big],\ \forall(p,q)\in\underline{n_{\B}}\times\underline{n_{\B}},\ \forall\epsilon\in[0,\epsilon_0].
\end{split}
\]
When inserting into \eqref{F-II-0} the first term in \eqref{F-II-1}, we notice that:
\begin{align*}
&\underset{|\alpha|\leq\,N,\,p\in\underline{n_\B}}{\sum}\ \underset{|\beta|\leq\,N,\,q\in\underline{n_\B}}{\sum}\hspace*{-0.3cm}\tLambda^{\epsilon}(\alpha,\beta)\big[\mathfrak{M}_{\B}[\bb1]_{\alpha,\beta}\big]_{p,q}\mathfrak{T}^{\epsilon}_\alpha\psi_p\otimes\overline{\mathfrak{T}^{\epsilon}_\beta\psi_q}=\\\nonumber
&\quad=\underset{|\alpha|\leq\,N,\,p\in\underline{n_\B}}{\sum}\ \underset{|\beta|\leq\,N,\,q\in\underline{n_\B}}{\sum}\hspace*{-0.3cm}\tLambda^{\epsilon}(\alpha,\beta)\,\tLambda^{\epsilon}_\alpha(x)\,[\tLambda^{\epsilon}_\beta(y)^{-1}]\big[\mathfrak{M}_{\B}[\bb1]_{\alpha,\beta}\big]_{p,q}\psi_p(x-\alpha)\,\overline{\psi_q(y-\beta)}
\end{align*}
\begin{align} \nonumber
&=\underset{|\alpha|\leq\,N,\,p\in\underline{n_\B}}{\sum}\ \underset{|\beta|\leq\,N,\,q\in\underline{n_\B}}{\sum}\hspace*{-0.3cm}\tLambda^{\epsilon}(x,y))\tOmega^\epsilon(\beta,x,\alpha)\,\tOmega^\epsilon(x,\beta,y)\Big[\big[\mathfrak{M}_{\B}[\bb1]_{\alpha,\beta}\big]_{p,q}[\tau_\alpha\psi_p]\otimes\overline{[\tau_{\beta}\psi_q]}\Big](x,y)\\ \nonumber
&\quad=\underset{|\alpha|\leq\,N,\,p\in\underline{n_\B}}{\sum}\ \underset{|\beta|\leq\,N,\,q\in\underline{n_\B}}{\sum}\hspace*{-0.3cm}\tLambda^{\epsilon}(x,y))\Big\{\Big[\big[\mathfrak{M}_{\B}[\bb1]_{\alpha,\beta}\big]_{p,q}[\tau_\alpha\psi_p]\otimes\overline{[\tau_{\beta}\psi_q]}\,+\,\epsilon\,\big[[\tilde{\mathfrak{m}}^{\epsilon}_{\B,2}]_{\alpha,\beta}\big]_{p,q}\Big](x,y)\Big\}
\end{align}
with $\tilde{\mathfrak{m}}^{\epsilon}_{\B,2}$ a remainder similar with $\tilde{\mathfrak{m}}^{\epsilon}_{\B,1}$.\\
For the main term above, Corollary \ref{C-PN-2} implies the formulas
\begin{align*}
&\tLambda^{\epsilon}\underset{|\alpha|\leq\,N_1}{\sum}\,\underset{p\in\underline{n_\B}}{\sum}\Big[\underset{|\beta|\leq\,N_2}{\sum}\,\underset{q\in\underline{n_\B}}{\sum}\Big(\big[\mathfrak{M}_{\B}[\bb1]_{\alpha,\beta}\big]_{p,q}[\tau_\alpha\psi_p]\otimes\overline{[\tau_{\beta}\psi_q]}\Big)\Big]
\\ \nonumber
&\quad=\tLambda^{\epsilon}\underset{|\alpha|\leq\,N_1}{\sum}\,\underset{p\in\underline{n_\B}}{\sum}[\tau_\alpha\psi_p]\,\otimes\,\Big[\underset{|\beta|\leq\,N_2}{\sum}\,\underset{q\in\underline{n_\B}}{\sum}\big([\tau_\alpha\psi_p]\,,\,[\tau_{\beta}\psi_q]\big)_{L^2(\X)}\overline{[\tau_{\beta}\psi_q]}\Big]
\end{align*}
\begin{align*}
\ \underset{N_2\nearrow\infty}{\longrightarrow}\ \tLambda^{\epsilon}\underset{|\alpha|\leq\,N_1}{\sum}\,\underset{p\in\underline{n_\B}}{\sum}[\tau_\alpha\psi_p]\otimes\overline{[\tau_\alpha\psi_p]}\ \underset{N_1\nearrow\infty}{\longrightarrow}\ \Int\,\tLambda^{\epsilon}\mathfrak{K}[P_{\B}].
\end{align*}
Finally let us notice that:
\begin{align*}
&[\tLambda^{\epsilon}\mathfrak{K}[P_{\B}]\,-\,\mathfrak{K}^\epsilon_{\B}](x,y)\,=\,\hspace*{-0.4cm}\underset{\alpha
\in\Gamma,\,p\in\underline{n_\B}}{\sum}\hspace*{-0.3cm}\Big[\tLambda^\epsilon(x,y)[\tau_\alpha\psi_p(x)]\,\overline{[\tau_\alpha\psi_p(y)]}\,-\,\mathfrak{T}^{\epsilon}_\alpha\psi_p(x)\,\overline{\mathfrak{T}^{\epsilon}_\alpha\psi_p(y)}\Big]\\
&\quad=\,\tLambda^\epsilon(x,y)\hspace*{-0.4cm}\underset{\alpha
\in\Gamma,\,p\in\underline{n_\B}}{\sum}\hspace*{-0.3cm}\big(1\,-\,\tOmega^\epsilon(x,\alpha,y)\big)[\tau_\alpha\psi_p(x)]\,\overline{[\tau_\alpha\psi_p(y)]}\,,
\end{align*}
in order to conclude that the arguments in Appendix \ref{SSS-PN-VI-2} imply the estimation:
	\[
	\begin{split}
		\forall N\in\mathbb{N},\ \exists\,{C_N(\BB)}>0:\ &\big|[\tLambda^{\epsilon}\mathfrak{K}[P_{\B}]\,-\,\mathfrak{K}^\epsilon_{\B}](x,y)\big| \leq\\
		&\leq \,{C_N(\BB)} \,<x-y>^{-N}\,\big[\epsilon\,\big(\underset{p\in\underline{\nB}}{\max}\lnu_{N+d+2,0}(\psi_p)^2\big)\,+\,\mathscr{O}(\epsilon^2)\big],\ \forall\epsilon\in[0,\epsilon_0].
	\end{split}
	\]

For the second term in \eqref{F-II-1} we have to study the limit:
\begin{align*}
\underset{N\nearrow\infty}{\lim}\underset{|\alpha|\leq\,N,\,|\beta|\leq\,N}{\sum}\ \underset{(p,q)\in\underline{n_\B}^2}{\sum}\hspace*{-0.3cm}\big[[\tilde{\mathfrak{m}}^{\epsilon}_{\B}]_{\alpha,\beta}\big]_{p,q}\,\mathfrak{T}^{\epsilon}_\alpha\psi_p(x)\,\overline{\mathfrak{T}^{\epsilon}_\beta\psi_q(y)}\,\equiv\,\underset{N\nearrow\infty}{\lim}\underset{|\alpha|\leq\,N,\,|\beta|\leq\,N}{\sum}\,\tilde{\mathfrak{K}}^{\epsilon}_{\alpha,\beta}(x,y)\,,
\end{align*}
and we use once again the Cotlar Stein Lemma with the countable family of indices $\Gamma\times\Gamma$. In fact we notice that we have:
\begin{align*}
&\big[\tilde{\mathfrak{K}}^{\epsilon}_{\alpha,\beta}\diamond\tilde{\mathfrak{K}}^{\epsilon}_{\alpha',\beta'}\big](x,y)=\\
&\quad=\hspace*{-0.3cm}\underset{(p,q)\in\underline{n_\B}^2}{\sum}\ \underset{(p',q')\in\underline{n_\B}^2}{\sum}\hspace*{-0.3cm}\big[[\tilde{\mathfrak{m}}^{\epsilon}_{\B}]_{\alpha,\beta}\big]_{p,q}\big[[\tilde{\mathfrak{m}}^{\epsilon}_{\B}]_{\alpha',\beta'}\big]_{p',q'}\Big[\delta_{\alpha',\beta}\delta_{p',q}+\epsilon\,\big[[\tilde{\mathfrak{m}}^{\epsilon}_{\B}]_{\alpha',\beta}\big]_{p',q}\Big]\mathfrak{T}^{\epsilon}_\alpha\psi_p(x)\,\overline{\mathfrak{T}^{\epsilon}_{\beta'}\psi_{q'}(y)}\,,
\end{align*}

and thus, for any $(m_1,m_2,m_3)\in\mathbb{N}^3$ there exists $C>0$ such that:
\begin{align*}
	\Big\|\Int\big[\tilde{\mathfrak{K}}^{\epsilon}_{\alpha,\beta}\diamond\tilde{\mathfrak{K}}^{\epsilon}_{\alpha',\beta'}\big]\Big\|_{\mathbb{B}(L^2(\X))}\,&\leq\,C<\alpha-\beta>^{-m_1}<\alpha'-\beta'>^{-m_2}\big[\delta_{\alpha',\beta}+\epsilon<\alpha'-\beta>^{-m_3}\big]
\end{align*}
and choosing asequately the powers $(m_1,m_2,m_3)$ one may obtain that for any $N\in\mathbb{N}$ there exists $C_N>0$ such that:
\begin{align*}
	\Big\|\Int\big[\tilde{\mathfrak{K}}^{\epsilon}_{\alpha,\beta}\diamond\tilde{\mathfrak{K}}^{\epsilon}_{\alpha',\beta'}\big]\Big\|_{\mathbb{B}(L^2(\X))}\,&\leq\, C_N<\alpha-\alpha'>^{-N}<\beta-\beta'>^{-N}.
\end{align*}

\end{proof}

\begin{corollary}\label{C-PN-1}
 For some $\epsilon_0>0$ as stipulated in \eqref{HF-Beps} there exists some $C_Q>0$ such that for any $\epsilon\in[0,\epsilon_0]$ one has the inclusion:
$$
\sigma(\widetilde{Q}^{\epsilon}_\B)\subset(-C_Q\epsilon,C_Q\epsilon)\bigcup(1-C_Q\epsilon,1+C_Q\epsilon).
$$
\end{corollary}
\begin{proof}
As $\widetilde{Q}^{\epsilon}_\B$ is a bounded self-adjoint operator,  Proposition \ref{P-II-2} above implies that if $\lambda\in\R$ is in its spectrum, then it obeys $|\lambda^2-\lambda|\leq\,C_Q\epsilon$ for some $C_Q>0$, that can be obtained from \eqref{F-PN-VII-2}, and any $\epsilon\in[0,\epsilon_0]$.
\end{proof}
We shall call $\widetilde{Q}^\epsilon_{\B}$ \textit{the 'quasi-projection' associated with the frame} $\widetilde{\blPsi}^\epsilon_{\B}$.

\subsubsection{The magnetic symbol of the ``quasi-projection"}

\begin{proposition}\label{R-PN-IV-1}
There exists $\tilde{q}^\epsilon_{\B}\in{S}^{-\infty}(\Xi)$ such that 
$
\widetilde{Q}_\B^{\epsilon}=\Op^\epsilon(\tilde{q}^\epsilon_{\B}).
$
\end{proposition}
\begin{proof}
Proposition \ref{P-II-1} implies that $\widetilde{Q}^\epsilon_{\B}=\Int\,\mathfrak{K}^\epsilon_{\B}=\Op^{\epsilon}(\tilde{q}^\epsilon_{\B})$ with 
$\mathfrak{K}^\epsilon_{\B}\in\mathring{\mathscr{S}}_{\Delta}(\X\times\X)$ (as in Notation \ref{N-per-distr}). 
Having in view Formula \eqref{F-KerOpA} from the appendices, the equalities $$\widetilde{Q}^\epsilon_{\B}=\Int\,\mathfrak{K}^\epsilon_{\B}=\Op^{\epsilon}(\tilde{q}^\epsilon_{\B})$$ imply that:
\beq\label{F-IV-1}
\tLambda^\epsilon\,\big[\big((\bb1\otimes\mathcal{F}_{\X^*})\tilde{q}^\epsilon_{\B}\big)\circ\Upsilon\big]\,=\,\mathfrak{K}^\epsilon_{\B},\quad\text
{i.e.:}
\quad\tilde{q}^\epsilon_{\B}=\big(\bb1_{\X}\otimes\mathcal{F}_{\X}\big)\big[[(\tLambda^\epsilon)^{-1}\mathfrak{K}^\epsilon_{\B}]\circ\Upsilon^{-1}\big]\,.
\eeq
We know that:
\begin{align*}
\mathfrak{K}^\epsilon_{\B}(x,y)&=\underset{\gamma\in\Gamma}{\sum}\ \underset{p\in\underline{n_\B}}{\sum}\big[\mathfrak{T}^{\epsilon}_\gamma\,\psi_{p}\big](x)\,\overline{\big[\mathfrak{T}^{\epsilon}_\gamma\,\psi_{p}\big](y)}\,,
\end{align*}
and conclude that:
\begin{align}\label{F-PN-IV-3}
&\tilde{q}^\epsilon_{\B}(z,\zeta)=\\ \nonumber
&\quad=\int_{\X}dv\,e^{i<\zeta,v>}\Big[\underset{\gamma\in\Gamma}{\sum}\ \underset{p\in\underline{n_\B}}{\sum}\Omega^\epsilon(z+v/2,\gamma,z-v/2)\,\big[\tau_{-\gamma}\psi_{p}\big](z+v/2)\,\overline{\big[\tau_{-\gamma}\psi_{p}\big](z-v/2)}\Big]\,.
\end{align}
\end{proof}

\begin{lemma}
With the above notations and hypothesis, for any continuous semi-norm $\lnu:S^{-\infty}(\Xi)\rightarrow\R_+$ there exists some constant $C_{\nu}>0$, such that:
\[
\lnu\big(\widetilde{q}^\epsilon_{\B}\,-\,p_{\B}\big)\,\leq\,C_\nu\,\epsilon,\quad\forall\epsilon\in[0,\epsilon_0].
\]
\end{lemma}
\begin{proof}
We go back to the formulas \eqref{F-IV-1} - \eqref{F-PN-IV-3} and \eqref{F-p-symb} in order to make explicit the difference of symbols we want to estimate and obtain terms of the form:
\begin{align}\nonumber
	\blacklozenge\hspace*{1cm}&<\zeta>^{2n}\big(\widetilde{q}^\epsilon_{\B}\,-\,p_{\B}\big)\big](z,\zeta)=<\zeta>^{2n}\int_{\X}dv\,e^{i<\zeta,v>}\,\times\\ &\hspace*{2cm}\times\underset{\gamma\in\Gamma}{\sum}\ \underset{p\in\underline{n_\B}}{\sum}\big[\Omega^\epsilon(z+v/2,\gamma,z-v/2)-1\big]\,\big[\tau_{-\gamma}\psi_{p}\big](z+v/2)\,\overline{\big[\tau_{-\gamma}\psi_{p}\big](z-v/2)}\nonumber \\ \label{F-PN-IV-2-0}
&=(-i\epsilon)\int_{\X}dv\,e^{i<\zeta,v>}<i\partial_v>^{2n}\Big[\big(\int_{<z+v/2,\gamma,z-v/2>}B^\epsilon\big)\,\times\\ \nonumber
&\times\,\Big[\int_0^1ds\exp\Big(-is\epsilon\int_{<z+v/2,\gamma,z-v/2>}B^\epsilon\Big)\Big]\underset{\gamma\in\Gamma}{\sum}\ \underset{p\in\underline{n_\B}}{\sum}\big[\tau_{-\gamma}\psi_{p}\big](z+v/2)\,\overline{\big[\tau_{-\gamma}\psi_{p}\big](z-v/2)}\Big],
\end{align}
\beq
\begin{aligned}
& \blacklozenge\hspace*{0.5cm}<\zeta>^{2n}\big[\partial_z^a\partial_{\zeta}^b\big(\widetilde{q}^\epsilon_{\B}\,-\,p_{\B}\big)\big](z,\zeta)\,=\,<\zeta>^n\Big[\partial_z^a\partial_{\zeta}^b\int_{\X}dv\,e^{i<\zeta,v>}\,\times\\ 
&\hspace*{0.5cm}\times\underset{\gamma\in\Gamma}{\sum}\ \underset{p\in\underline{n_\B}}{\sum}\big[\Omega^\epsilon(z+v/2,\gamma,z-v/2)-1\big]\,\big[\tau_{-\gamma}\psi_{p}\big](z+v/2)\,\overline{\big[\tau_{-\gamma}\psi_{p}\big](z-v/2)}\,\Big]
\\ \label{F-PN-IV-2} 
&=\hspace*{-0.2cm}\int_{\X}\hspace*{-0.2cm}dv\,e^{i<\zeta,v>}<i\partial_v>^{2n}\Big[(iv)^b\underset{a^\prime\leq{a}}{\sum}C^a_{a^\prime}\,\big[\big(\partial_z^{a^{\prime}}\tOmega^\epsilon\big)(z+v/2,\gamma,z-v/2)\big]\,\times\\ 
&\hspace*{0.5cm}\times\,\big[\partial_z^{a-a^\prime}\big(\psi_p(z-\gamma+v/2)\,\overline{\psi_p(z-\gamma-v/2)}\big)\big]\Big], 
\end{aligned}
\eeq
and also
\begin{align}\label{F-D-Omega}
	\big(\partial^c_v\partial_z^{a}\tOmega^\epsilon\big)(z+v/2,\gamma,z-v/2)=\partial^c_v\partial_z^{a}\Big[\exp\Big(-i\epsilon\int_{<z+v/2,\gamma,z-v/2>}B^\epsilon\Big)\Big].
\end{align}
 \eqref{C-PN-IV-1} in Appendix \ref{SSS-PN-VI-1} gives us an upper bound that together with the last lines in \eqref{F-PN-IV-2-0} and \eqref{F-PN-IV-2} imply the conclusion of the Lemma.
\end{proof}

\subsubsection{The ``magnetic" modified projection for the isolated Bloch family.}

\begin{definition}\label{D-II-1}
	For $\epsilon\in[0,\epsilon_0]$ with $\epsilon_0>0$ small enough and $C_Q>0$ as in Corollary \ref{C-PN-1}, we may choose a circle $\mathscr{C}_\epsilon\subset\Co$ centred in $1$ with a radius $2C_Q\epsilon<1/2$ and define:
	\[ \label{dhc22}
	P_\B^{\epsilon}:= \frac{i}{2\pi}\oint_{\mathscr{C}_\epsilon}d\zz\,\big(\widetilde{Q}_\B^{\epsilon}-\zz\bb1\big)^{-1}
	\]
	as the spectral projection of $\widetilde{Q}_\B^{\epsilon}$ on the interval $(1-C_Q\epsilon,1+C_Q\epsilon)$. Let us also introduce the notation:
	\[
	\mathfrak{L}^{\epsilon}_{\B}:=\Rge\,P^{\epsilon}_{\B}
	\]
	and call it the ``\emph{magnetic modified space}"  associated with  the isolated Bloch family $\B$.
\end{definition}
\begin{remark}\label{R-PN-0}
Since $\widetilde{Q}_\B^{\epsilon}$ is a self-adjoint operator  as in Corollary \ref{C-PN-1} while $P^\epsilon_{\B}$ is its spectral projection on the spectral island close to $1\in\Co$, it follows that $\mathfrak{L}^{\epsilon}_{\B}\subset\widetilde{\mathfrak{L}}^{\epsilon}_{\B}$ and  $\|\widetilde{Q}_\B^{\epsilon}-P^\epsilon_{\B}\|_{\mathbb{B}(L^2(\X))}\leq{C}\epsilon$ for some constant $C>0$ (that can be estimated using the constant $C_Q$ in  Corollary \ref{C-PN-1}) and any $\epsilon\in[0,\epsilon_0]$. 
\end{remark}

\begin{lemma}\label{R-PN-1}
	With the above notations, for any continuous semi-norm $\lnu:S^{-\infty}(\Xi)\rightarrow\R_+$ there exists some constant $C_{\nu}>0$, depending on a finite number of derivatives of $B^\epsilon\in\Fb^2(\X)$, such that:
	\[
	\lnu\big(\widetilde{q}^\epsilon_{\B}-p^\epsilon_{\B}\big)\,\leq\,C_\nu\,\epsilon,\quad\forall\epsilon\in[0,\epsilon_0].
	\]
\end{lemma}
\begin{proof}
Let us notice the following identity:
\beq\label{F-PN-V-1}(\widetilde{Q}_\B^{\epsilon}-\zz\bb1)\big (\widetilde{Q}_\B^{\epsilon}+(\zz-1)\bb1 \big )=\big (\widetilde{Q}_\B^{\epsilon}\big )^2-\widetilde{Q}_\B^{\epsilon} +\zz(1-\zz)\bb1
\eeq
implying that
\begin{align*}
(\widetilde{Q}_\B^{\epsilon}-\zz \bb1)^{-1}&=\big (\widetilde{Q}_\B^{\epsilon}+(\zz-1)\bb1 \big ) \Big ( \big (\widetilde{Q}_\B^{\epsilon}\big )^2-\widetilde{Q}_\B^{\epsilon} +\zz(\zz-1)\bb1\Big )^{-1} \\
&=: \zz^{-1}(1-\zz)^{-1}\, \widetilde{Q}_\B^{\epsilon}  +\Big (\big (\widetilde{Q}_\B^{\epsilon}\big )^2-\widetilde{Q}_\B^{\epsilon}\Big ) \, F(\widetilde{Q}_\B^{\epsilon},\zz)
\end{align*}
where the function $F$ can easily be read off by expanding the above inverse. Using the magnetic Beals criterion and the arguments in section 6.2 of \cite{IMP-2}, we have that  $F(\widetilde{Q}_\B^{\epsilon},\zz)$ is a magnetic pseudo-differential operator with a symbol in $S^0$. The operator $\big (\widetilde{Q}_\B^{\epsilon}\big )^2-\widetilde{Q}_\B^{\epsilon}$ has a symbol in $S^{-\infty}$ which is also proportional with $\epsilon$ as can be seen from Proposition \ref{P-II-2}. Thus $\Big (\big (\widetilde{Q}_\B^{\epsilon}\big )^2-\widetilde{Q}_\B^{\epsilon}\Big ) \, F(\widetilde{Q}_\B^{\epsilon},\zz)$ has a symbol in $S^{-\infty}$, uniformly in $\zz$, and proportional with $\epsilon$. 
The final step is to integrate with respect to $\zz$. 
\end{proof}

\begin{remark}\label{R-p-eps-B}
The above Lemma \ref{R-PN-1} and Proposition \ref{R-PN-IV-1} imply that $P^\epsilon_{\B}=\Op^\epsilon(p^\epsilon_{\B})$ with $p^\epsilon_{\B}\in{S}^{-\infty}(\Xi)$ and for any continuous semi-norm $\lnu:S^{-\infty}(\Xi)\rightarrow\R_+$ there exists some constant $C_{\nu}>0$ such that:
\[
\lnu\big(p^\epsilon_{\B}\,-\,p_{\B}\big)\,\leq\,C_\nu\,\epsilon,\quad\forall\epsilon\in[0,\epsilon_0].
\]
\end{remark}

\subsubsection{Commutator with the Hamiltonian.}\label{SS-comm-H-eps}

As remarked before the definition \eqref{F-PN-1}, the functions $\psi_p$ for $p\in\underline{n_{\B}}$ belong to $\mathscr{S}(\X)$ and we conclude that the family $\big\{\mathfrak{T}^{\epsilon}_{\gamma}\psi_p,\ \gamma\in\Gamma,\ p\in\underline{n_{\B}}\big\}$ is contained in $\mathcal{D}(H^{\epsilon})$ for any $\epsilon\in[0,\epsilon_0]$ for some $\epsilon_0>0$ as discussed in subsection \ref{SS-m-P-frame}.
Taking into account that $h^\circ\in{S}^p_1(\Xi)_\Gamma$ and $p^\epsilon_{\B}\in{S}^{-\infty}(\Xi)$ for any $\epsilon\in[0,\epsilon_0]$ with some $\epsilon_0>0$, let us compute:
\begin{align*}
	H^{\epsilon}P^\epsilon_{\B}-P^\epsilon_{\B}\,H^{\epsilon}&=\Op^\epsilon(h^\circ\sharp^\epsilon{p}^\epsilon_{\B}-{p}^\epsilon_{\B}\sharp^\epsilon{h})=\Op^\epsilon(h^\circ\sharp{p}_{\B}-{p}_{\B}\sharp{h^\circ})+\mathcal{O}(\epsilon).
\end{align*}
\[
\Op^\circ\big(h^\circ\sharp{p}_{\B}-p_{\B}\sharp{h}^\circ\big)=H^\circ\,P_{\B}\,-\,p_{\B}\,H^\circ\,=\,0
\]
and conclude that:
\beq \label{P-IV-2}
\big[\,H^\epsilon\,,\,P^\epsilon_{\B}\,\big]\,=\,\mathcal{O}(\epsilon).
\eeq

\begin{proposition}\label{P-est-Jepsilon}
	Let $H^\circ$ in \eqref{D-Hcirc} satisfy Hypothesis \ref{H-isBf} and consider a perturbation by a magnetic field satisfying Hypothesis \ref{HF-Beps}. Then there exists $\delta_0>0$ so that the interval $$J^{\delta}_\B:=\big(E_-+{ 2}\delta\,,\,E_+-{2}\delta\big)$$ is not empty for any $\delta \in (0,\delta_0]$ and for any test function $\varphi\in\,C^\infty_0(\R)$ with $\supp\varphi\subset J^{\delta}_\B$, there exist $\epsilon_0 >0$ and $C>0$ such that for any $\epsilon\in[0,\epsilon_0]$ we have:
	$$\big\|P^{\epsilon}_\B\,\varphi\big(H^{\epsilon}\big)\,-\,\varphi\big(H^{\epsilon}\big)\big\|_{\mathbb{B}(L^2(\X))}\,\leq\,C\,\epsilon.$$
\end{proposition}

\begin{proof}
	The hypothesis $J:=\supp\varphi\subset(E_-+2\delta,E_+-2\delta)$ implies that: 
	\beq \label{F-Ecirc-h-J}
	E_J(H^\circ)=E_{J}(H_\bot)\oplus\,E_{J}(H_\B )=E_{J}(H_{\B})\subset\,E_{\{0\}}(H_\bot)=P_\B. 
	\eeq
	If we use Proposition 6.33 in \cite{IMP-2} and the notations:
	\[\begin{split}
		&\varphi\big(H^{\epsilon}\big)\,=\,\Op^{\epsilon,\cc}\big(\varphi^{\epsilon}[h]\big),\quad\varphi^{\epsilon}[h]\in\,S^{-p}_{1}(\X^*\times\X),\\
		&\varphi\big(H^{\circ}\big)\,=\,\Op^\circ\big(\varphi^{\circ}[h]\big),\quad\varphi^{\circ}[h]\in\,S^{-p}_{1}(\X^*\times\X)\,,
	\end{split}\]
	we notice that \eqref{F-Ecirc-h-J} implies the equality:
	\[
	P_\B\,\varphi\big(H^{\circ}\big)\,=\,\varphi\big(H^{\circ}\big)\quad\text{i.e.:}\quad\,p_\B\sharp^{B^\circ}\varphi^{\circ}[h]\,-\,\varphi^{\circ}[h]\,=\,0\,.
	\]
	Thus, the usual estimation of the magnetic perturbation  on symbols implies that:
	\[\begin{split}
		\big\|P^{\epsilon}_\B\,\varphi\big(H^{\epsilon}\big)\,-\,\varphi\big(H^{\epsilon}\big)\big\|_{\mathbb{B}(L^2(\X))}\,=\,\mathscr{O}(\epsilon).
	\end{split}\]
\end{proof}

\subsection{A Parseval frame for the magnetic modified space}\label{SS-psGS}

Now we shall elaborate a procedure of transforming any "quasi-Parseval" frame into a Parseval one, replacing the Gram-Schmidt orthogonalization procedure and aply it for the closed subspace $\mathfrak{L}^\epsilon_{\B}=P^\epsilon_{\B}L^2(\X)$ with the perturbed frame $\widetilde{\blPsi}^\epsilon_{\B}$. One will notice that this procedure is very different from Gram-Schmidt one. 

By usual holomorphic functional calculus with the bounded self-adjoint operator $\widetilde{Q}^{\epsilon}_\B$, if we consider the function $\zz\mapsto \zz^{-1/2}=e^{-2^{-1}{\rm Ln}(z)}$ that is holomorphic on any disk around $1$ not containing $0$, we can define:
\beq \label{dhc21}
\Theta^{\epsilon}_\B\,:=\,\frac{i}{2\pi}\oint_{\mathscr{C}}d\zz\,\zz^{-1/2}\,\big(\widetilde{Q}_\B^{\epsilon}-\zz\bb1\big)^{-1},
\eeq
as a bounded self-adjoint operator commuting with $\widetilde{Q}^{\epsilon}_\B$. It  satisfies the equalities: 
\beq \label{F-QThetaP-eps-c}
P^{\epsilon}_\B\,=\,\Theta^{\epsilon}_\B\,\widetilde{Q}^{\epsilon}_\B\,\Theta^{\epsilon}_\B\,=\, [\Theta^{\epsilon}_\B]^2\,\widetilde{Q}_\B^{\epsilon}\,=\, \widetilde{Q}_\B^{\epsilon}\,[\Theta^{\epsilon}_\B]^2.
\eeq 

\begin{lemma}\label{L-theta-tq}
{The operator $\Theta^{\epsilon}_\B$ is a magnetic pseudo-differential operator with a symbol $\theta^\epsilon_{\B}\in S^{-\infty}(\Xi)$ and such that for $\widetilde{q}^\epsilon_{\B}\in S^{-\infty}(\Xi)$ the symbol of $\widetilde{Q}^\epsilon_{\B}$ given by Proposition \ref{R-PN-IV-1} and for any continuous semi-norm $\lnu:S^{-\infty}(\Xi)_\Gamma\rightarrow\R_+$ there exists some constant $C_{_\nu}>0$ such that:
$$
\lnu(\theta^\epsilon_{\B}-\widetilde{q}^\epsilon_{\B})\,\leq\,C_{_\nu}\epsilon,\quad\forall\epsilon\in[0,\epsilon_0].
$$}
\end{lemma}
\begin{proof}
We have the identity 
\[
\Theta^{\epsilon}_\B = \widetilde{Q}_\B^{\epsilon} \,\frac{i}{2\pi}\oint_{\mathscr{C}}d\zz\,\zz^{-3/2}\,\big(\widetilde{Q}_\B^{\epsilon}-\zz\bb1\big)^{-1}.
\]
Proposition \ref{P-II-1} implies that $\widetilde{Q}^{\epsilon}_{\B}$ is a magnetic pseudo-differential operator with symbol of class $S^{-\infty}$. 
Via the magnetic Beals criterion we have that the above complex integral defines a magnetic pseudo-differential operator with a symbol in $S^0$, which by multiplication with $\widetilde{Q}^{\epsilon}_{\B}$ becomes an operator with a symbol 
$\theta^\epsilon_{\B}$ of class $S^{-\infty}$.

Because $P_\B^\epsilon$ is a projection, the arguments above imply that we may write:
\beq \label{ddc1}
\begin{aligned}
\Theta^{\epsilon}_\B -P_\B^\epsilon&= \,\frac{i}{2\pi}\oint_{\mathscr{C}}d\zz\,\zz^{-1/2}\,\Big (\big(\widetilde{Q}_\B^{\epsilon}-\zz\bb1\big)^{-1} -\big(P_\B^{\epsilon}-\zz\bb1\big)^{-1}\Big )\\&=\,\frac{i}{2\pi}\oint_{\mathscr{C}}d\zz\,\zz^{-1/2}\,\big(\widetilde{Q}_\B^{\epsilon}-\zz\bb1\big)^{-1} \big ( P_\B^{\epsilon}-\widetilde{Q}_\B^{\epsilon}\big )\big(P_\B^{\epsilon}-\zz\bb1\big)^{-1}
\end{aligned}
\eeq
and by the same Beals \& bootstrap argument we notice  that the right hand side of \eqref{ddc1} has a magnetic symbol in $\sigma^\epsilon_{\B}\in{S}^{-\infty}(\Xi)$ with semi-norms of order $\epsilon$, i.e. for any continuous semi-norm $\nu:S^{-\infty}(\Xi)\rightarrow\R_+$ there exists some $C_\nu>0$ such that $\nu(\sigma^\epsilon_{\B})\leq{C}_\nu\,\epsilon$.

Coupling this with Lemma \ref{R-PN-1} we also get that $\Theta^{\epsilon}_\B -\widetilde{Q}_\B^\epsilon$ has the same property.
\end{proof}

\begin{proposition}\label{P-IV-1}
For any $(\alpha,p,\epsilon)\in\Gamma\times\underline{n_{\B}}\times[0,\epsilon_0]$, the function $\psi^\epsilon_{\alpha,p}:=\Theta^{\epsilon}_\B\mathfrak{T}^{\epsilon}_\alpha\,\psi_{p}\in{L^2}(\X)$ belongs to $\mathscr{S}(\X)$ with estimations uniform in $\epsilon\in[0,\epsilon_0]$.
\end{proposition}
\begin{proof}
From the construction in Subsection \ref{SS-B-P-fr} we know that $\psi_{p}\in\mathscr{S}(\X)$. For any  $n\in\mathbb{N}$ there exists $C_n>0$ such that $|\psi_p(\hat{x}+\gamma)|\leq{C_n}<\gamma>^{-n}$. Since  \[\X\ni{x}\mapsto\Lambda^\epsilon_\gamma(x)=\exp\Big(-i\epsilon\int_{[x,\gamma]}A^\epsilon\Big)\in\mathbb{U}(1)\] is a smooth function of modulus 1, we have that $\mathfrak{T}^{\epsilon}_\alpha\,\psi_{p}\in\mathscr{S}(\X)$ and given any $n\in\mathbb{N}$ there exists $C_n>0$ such that $|\mathfrak{T}^{\epsilon}_\alpha\,\psi_{p}(\hat{x}+\gamma)|\leq{C_n}<\gamma-\alpha>^{-n}$. Finally our previous arguments implied that $\Theta_\B^{\epsilon}$ have integral kernels of class $\mathring{\mathscr{S}}_{\Delta}(\X\times\X)$ with fast off-diagonal decay.
\end{proof}

\begin{proposition}\label{P-III-1}
If we define for any $N\in\mathbb{N}\setminus\{0\}$ the integral kernel:
\[
\mathfrak{K}_N[P^{\epsilon}_\B]:=\underset{|\alpha|\leq\,N}{\sum}\ \underset{p\in\underline{n_\B}}{\sum}\big[\Theta^{\epsilon}_\B\mathfrak{T}^{\epsilon}_\alpha\,\psi_{p}\big]\otimes\overline{\big[\Theta^{\epsilon}_\B\mathfrak{T}^{\epsilon}_\alpha\,\psi_{p}\big]}\,,
\]
and its associated integral operator $P^{\epsilon}_{\B,N}:=\Int\,\mathfrak{K}_N[P^{\epsilon}_\B]$, then the following limits exist and verify the equalities:
\begin{itemize} 
	\item $\mathfrak{K}=\underset{N\nearrow\infty}{\lim}\mathfrak{K}_N[P^{\epsilon}_\B]\in\mathring{\mathscr{S}}(\X\times\X)\bigcap{C}^\infty(\X\times\X)$ uniformly on compact sets in $\X\times\X$; 
	\item $\underset{N\nearrow\infty}{\lim}P^{\epsilon}_{\B,N}=Int\,\mathfrak{K}=P^{\epsilon}_{\B}\in\mathbb{B}\big(L^2(\X)\big)$ for the strong operator topology.
	\end{itemize}
\end{proposition}
\begin{proof}
Let us denote by $\mathfrak{K}[P^{\epsilon}_\B]_\gamma:=\underset{p\in\underline{n_\B}}{\sum}\big[\Theta^{\epsilon}_\B\mathfrak{T}^{\epsilon}_\gamma\,\psi_{p}\big]\otimes\overline{\big[\Theta^{\epsilon}_\B\mathfrak{T}^{\epsilon}_\gamma\,\psi_{p}\big]}$ and by $\mathfrak{K}[\Theta^{\epsilon}_{\B}]\in\mathring{\mathscr{S}}(\X\times\X)$ the integral kernel of $\Theta^{\epsilon}_{\B}$. Then:
\begin{align}
\mathfrak{K}[P^{\epsilon}_\B]_\alpha\diamond\mathfrak{K}[P^{\epsilon}_\B]_\beta=\underset{(p,q)\in\underline{n_\B}^2}{\sum}\big(\Theta^{\epsilon}_\B\mathfrak{T}^{\epsilon}_\alpha\,\psi_{p}\,,\,\Theta^{\epsilon}_\B\mathfrak{T}^{\epsilon}_\beta\,\psi_{q}\big)_{L^2(\X)}\big[\Theta^{\epsilon}_\B\mathfrak{T}^{\epsilon}_\alpha\,\psi_{p}\big]\otimes\overline{\big[\Theta^{\epsilon}_\B\mathfrak{T}^{\epsilon}_\beta\,\psi_{q}\big]}
\end{align}
and considering each 'coefficient':
\begin{align}
\big(\Theta^{\epsilon}_\B\mathfrak{T}^{\epsilon}_\alpha\,\psi_{p}\,,\,\Theta^{\epsilon}_\B\mathfrak{T}^{\epsilon}_\beta\,\psi_{q}\big)_{L^2(\X)}=&\int_{\X}\int_{\X}\int_{\X}dx\,dy\,dz\,\mathfrak{K}[\Theta^{\epsilon}_{\B}](x,y)\mathfrak{K}[\Theta^{\epsilon}_{\B}](x,z)\,\times\\
&\times\,\overline{\Lambda^{\epsilon}_\alpha(y)}\overline{\psi_p(y-\alpha)}\Lambda^{\epsilon}_\beta(z)\psi_q(z-\beta)
\end{align}
and putting into evidence the decay of the functions under the integrals,  we finally obtain the estimation:
\begin{align*}
&<x-y>^{-n_1}<x-z>^{-n_2}<y-\alpha>^{-n_3}<z-\beta>^{-n_4}\leq\\
&\hspace*{3cm}\leq\,C<x-y>^{-n_1+n_2}<y-z>^{-n_2+n_3}<z-\alpha>^{-n_3+n_4}<\alpha-\beta>^{-n_4}\,,
\end{align*}
that allows the use of the Cotlar-Stein Lemma.
\end{proof}

The previous arguments show that we can write:
\[\begin{split}
	\mathfrak{K}[P^{\epsilon}_\B](x,y)&=\hspace*{-0.3cm}\underset{(\alpha,p)\in\Gamma\times\underline{n_\B}}{\sum}\hspace*{-0.3cm}\big[\Theta^{\epsilon}_\B\mathfrak{T}^{\epsilon}_\alpha\,\psi_{p}\big](x)\,\overline{\big[\Theta^{\epsilon}_\B\mathfrak{T}^{\epsilon}_\alpha\,\psi_{p}\big](y)}.
\end{split}\]

\begin{proposition}
The family 
\beq \label{DF-Pfr-B-eps}
\blPsi^{\epsilon}_{\B}:=\big\{\psi^{\epsilon}_{\alpha,p}:=\Theta^{\epsilon}_\B\mathfrak{T}^{\epsilon}_\alpha\,\psi_{p},\quad\forall(\alpha,p)\in\Gamma\times\underline{n_\B}\big\}
\eeq
is a Parseval frame for the subspace $\mathfrak{L}^{\epsilon}_{\B}\subset\,L^2(\X)$.
\end{proposition}
\begin{proof}
Each $\Theta^{\epsilon}_\B\mathfrak{T}^{\epsilon}_\alpha\,\psi_{p}$ belongs to $\mathfrak{L}^{\epsilon}_{\B}=\Rge\,P^{\epsilon}_{\B}$.
If $f\in\mathfrak{L}^{\epsilon}_{\B}$, we have the equalities:
\begin{align*}
f=P^{\epsilon}_{\B}f=\hspace*{-0.5cm}\underset{(\alpha,p)\in\Gamma\times\underline{n_\B}}{\sum}\hspace*{-0.3cm}\big(\Theta^{\epsilon}_\B\mathfrak{T}^{\epsilon}_\alpha\,\psi_{p}\,,\,f\big)_{L^2(\X)}\,\Theta^{\epsilon}_\B\mathfrak{T}^{\epsilon}_\alpha\,\psi_{p}.
\end{align*}
\end{proof}
\begin{remark}\label{R-PN-V-2}
The abstract construction presented in the end of Appendix \ref{A-frames} implies the existence of an injective $C^*$-algebra morphism $\mW^\epsilon_{\B}:\mathbb{B}\big(P^\epsilon_{\B}L^2(\X)\big)\rightarrow\mathbb{B}\big(\ell^2(\Gamma)\otimes\Co^{\nB}\big)$  associated with  the Parseval frame $\blPsi^\epsilon_{\B}$ and given by the explicit formula:
\[
\mW^\epsilon_{\B}T\,=\,\mathfrak{C}_{\blPsi^\epsilon_{\B}}T\mathfrak{C}_{\blPsi^\epsilon_{\B}}^*.
\]
Moreover, using the Parseval frame $\blPsi^\epsilon_{\B}$ in \eqref{DF-Pfr-B-eps} we have the following ``perturbed" version of the definition \eqref{C-PN-VI-3}:
\beq \label{C-PN-VI-3-1}
\big[\mathfrak{M}^\epsilon_{\B}[T]_{\alpha,\beta}\big]_{p,q}\,:=\,\big({\cal{e}}_\alpha\otimes{\cal{e}}_p\,,\,\mathfrak{W}^\epsilon_{\B}[T]\,{\cal{e}}_\beta\otimes{\cal{e}}_q\big)_{\ell^2(\Gamma)^{\nB}}\,=\,\big(\psi^\epsilon_{\alpha,p}\,,\,T\,\psi^\epsilon_{\beta,q}\big)_{L^2(\X)}.
\eeq
\end{remark}
\begin{remark}\label{R-III-1}
We have the formula $P^\epsilon_{\B}=\Id_{\mathfrak{L}^\epsilon_{\B}}=\mathfrak{C}_{\blPsi^\epsilon_{\B}}\mathfrak{C}_{\blPsi^\epsilon_{\B}}^*$ and taking $\epsilon=0$ we obtain the formulas:
\[
P_{\B}\,=\,\Int\,\mathfrak{K}_{\B},\quad\mathfrak{K}_{\B}:=\hspace*{-0.4cm}\underset{(\gamma,p)\in\Gamma\times\underline{n_{\B}}}{\sum}(\tau_{-\gamma}\psi_p)\otimes(\overline{\tau_{-\gamma}\psi_p)}\,.
\]
\end{remark}

\subsection{The effective Hamiltonian of the isolated Bloch family in a regular magnetic field}\label{SS-PN-V-2}

Once we have put into evidence the orthogonal projection $P^\epsilon_{\B}$ that is almost $H^\epsilon$-invariant, with an error of order $\epsilon$ and have as Parseval frame the system $\blPsi^\epsilon_{\B}$\,, let us consider the \textit{``reduced magnetic Hamiltonian"}  associated with  the closed subspace $\mathfrak{L}^\epsilon_{\B}=P^\epsilon_{\B}\big[L^2(\X)\big]$:
\beq \label{DF-m-band-H}
\ham^\epsilon_{\B}\,:=\,P^\epsilon_{\B}\,H^\epsilon\,P^\epsilon_{\B}
\eeq
and its associated infinite matrix $\mathfrak{M}^\epsilon_{\B}[\ham^\epsilon_{\B}]$ with respect to this Parseval frame as in \eqref{C-PN-VI-3-1}.

\begin{proposition}\label{P-f-2}
	Suppose given a magnetic field as in \eqref{HF-Beps} and a distribution kernel $\mathfrak{K}\in\mathring{\mathscr{S}}_{\Delta}
	(\X\times\X)$ (see the Notation \ref{N-per-distr}) satisfying the periodicity property $$(\tau_{\gamma}\otimes\tau_{\gamma})\mathfrak{K}=\mathfrak{K} \mbox{  for any } \gamma\in\Gamma\,.$$ The infinite matrix  associated with  the operator $T^\epsilon:=\Int\,\tLambda^{\epsilon}\mathfrak{K}$ with respect to the Parseval frame $\blPsi^{\epsilon}_{\B}$ in \eqref{DF-Pfr-B-eps}, is of the form 
	\begin{align*}
	\big[\mathfrak{M}^\epsilon_{\B}[T^\epsilon]_{\alpha,\beta}\big]_{p,q}:&=\big(\psi^\epsilon_{\alpha,p}\,,\,\big[\Int\,\tLambda^{\epsilon}\mathfrak{K}\big]\psi^\epsilon_{\beta,q}\big)_{L^2(\X)}\\ &=\tLambda^{\epsilon}(\alpha,\beta)\Big[\big[[\mathring{\mathfrak{m}}_{\B}(\mathfrak{K})]_{\alpha-\beta}\big ]_{p,q}+\epsilon\big[\widetilde{\mathfrak{M}}^{\epsilon}_{\B}[\mathfrak{K}]_{\alpha,\beta}\big]_{p,q}\Big]\,,
	\end{align*}
	where:
	\begin{itemize} 
	\item $\mathring{\mathfrak{m}}_{\B}(\mathfrak{K})\in{\cal{s}}(\Gamma;\MmN)$ is independent of $\epsilon$ and is defined by
	$$\big[[\mathring{\mathfrak{m}}_{\B}]_{\gamma}(\mathfrak{K})\big ]_{p,q}:=\big(\psi_p\,,\,\Int\big[(\tau_{\gamma}\otimes\bb1)\mathfrak{K}\big]\psi_q\big)_{L^2(\X)},$$ for $\gamma\in\Gamma$ and $1\leq p,q\leq n_{\B}$;
	\item the  family 
	$\big\{\widetilde{\mathfrak{M}}^{\epsilon}_{\B}[\mathfrak{K}],\,\epsilon\in[0,\epsilon_0]\big\}$ belongs to a bounded set of  $\mathscr{M}^\circ_\Gamma[\MmN]$\,.
	\end{itemize}
\end{proposition}
\begin{proof}
Using \eqref{DF-Pfr-B-eps} we notice that we have the following decomposition:
\beq \label{F-B-3}
\begin{aligned} 
&\big(\Theta^\epsilon_{\B}\,\mathfrak{T}^{\epsilon}_\alpha\psi_p\,,\,[\Int\,\tLambda^{\epsilon}\mathfrak{K}]\Theta^\epsilon_{\B}\,\mathfrak{T}^{\epsilon}_\beta\, \psi_q\big)_{L^2(\X)}=\big\langle\,\mathfrak{K}\,,\,\big(\overline{\Lambda^{\epsilon}_\alpha}\otimes\Lambda^{\epsilon}_\beta\big)\,\tLambda^{\epsilon}\,\big(\overline{\tau_{-\alpha}\psi_p}\otimes\tau_{-\beta}\psi_q\big)\big\rangle_{\mathscr{S}(\X\times\X)}+\,\\ 
	&\ +\,\big([\Theta^\epsilon_{\B}-\bb1]\,\tpsi^\epsilon_{\alpha,p}\,,\,[\Int\,\tLambda^{\epsilon}\mathfrak{K}]\Theta^\epsilon_{\B}\,\tpsi^\epsilon_{\beta,q}\big)_{L^2(\X)}\,+\,\big(\tpsi^\epsilon_{\alpha,p}\,,\,[\Int\,\tLambda^{\epsilon}\mathfrak{K}][\Theta^\epsilon_{\B}-\bb1]\,\tpsi^\epsilon_{\beta,q}\big)_{L^2(\X)}.
\end{aligned}
\eeq

In order to deal with the first term in the r.h.s. of  \eqref{F-B-3}, we use Stokes' formula in order to study the $\epsilon$ dependence:
\begin{align}\nonumber
	\big[\big(\overline{\Lambda^{\epsilon}_\alpha}\otimes\Lambda^{\epsilon}_\beta\big)\,\tLambda^{\epsilon}](x,y)&=\exp\Big[-i\Big(\int_{[\alpha,x]}A^{\epsilon}+\int_{[y,\beta]}A^{\epsilon}+\int_{[x,y]}A^{\epsilon}\Big]\\ \label{F-PN-VII-1}
	&=\tLambda^{\epsilon}(\alpha,\beta)\,\exp\Big[-i\Big(\int_{<\alpha,x,y>}B^{\epsilon}\Big)\Big]\,\exp\Big[-i\Big(\int_{<\alpha,y,\beta>}B^{\epsilon}\Big)\Big].
\end{align}
Using the computation in Subsection  \ref{SSS-PN-VI-1} we may conclude that
\begin{align}\label{F-PN-VII-1-1}
	\Big|\big[\big(\overline{\Lambda^{\epsilon}_\alpha}\otimes\Lambda^{\epsilon}_\beta\big)\,\tLambda^{\epsilon}](x,y)\,-\,\tLambda^{\epsilon}(\alpha,\beta)\Big|\leq\,C\epsilon\, \big[<x-\alpha><x-y>+<y-\beta><x-y>\big]\,,
\end{align}
and by iteration that all the derivatives of $\big(\overline{\Lambda^{\epsilon}_\alpha}\otimes\Lambda^{\epsilon}_\beta\big)\,\tLambda^{\epsilon}$ may be bounded by polynomials of the form:
$$
\epsilon^M<x-\alpha>^{M_1}<y-\beta>^{M_2}<x-y>^{M_3}.
$$
Denoting $\mathfrak{k}_N(x,y):=<x-y>^N$ and $\mathfrak{K}^\prime_N:=\mathfrak{k}_N\mathfrak{K}\in\mathring{\mathscr{S}}_{\Delta}(\X\times\X)$ we notice that:
\begin{align}\nonumber
	&\big\langle\,\mathfrak{K}\,,\,\big(\overline{\Lambda^{\epsilon}_\alpha}\otimes\Lambda^{\epsilon}_\beta\big)\,\tLambda^{\epsilon}\,\big(\overline{\tau_{-\alpha}\psi_p}\otimes\tau_{-\beta}\psi_q\big)\big\rangle_{\mathscr{S}(\X\times\X)}=\big\langle\,\mathfrak{K}\,,\,\tLambda^{\epsilon}(\alpha,\beta)\,\big(\overline{\tau_{-\alpha}\psi_p}\otimes\tau_{-\beta}\psi_q\big)\big\rangle_{\mathscr{S}(\X\times\X)}\,+\\ \label{F-B-I-1}
	&\hspace*{3cm}+\big\langle\,\mathfrak{K}^\prime_N\,,\,\mathfrak{k}_{-N}\Big[\big[\big(\overline{\Lambda^{\epsilon}_\alpha}\otimes\Lambda^{\epsilon}_\beta\big)\,\tLambda^{\epsilon}]\,-\,\tLambda^{\epsilon}(\alpha,\beta)\Big]\big(\overline{\tau_{-\alpha}\psi_p}\otimes\tau_{-\beta}\psi_q\big)\big\rangle_{\mathscr{S}(\X\times\X)}.
\end{align}
Taking now into account the estimation  \eqref{F-PN-VII-1-1} and the following remarks concerning its derivatives and noticing that:
\begin{align*}
	<x-y>^{-N}=<(x-\alpha)-(y-\beta)+(\alpha-\beta)>^{-N}\leq\,C_N<x-\alpha>^N<y-\beta>^N<\alpha-\beta>^{-N}
\end{align*}
we conclude that, for any $N\in\mathbb{N}$,  there exists $\Phi_{N,\alpha,\beta}\in\mathscr{S}(\X\times\X)$ uniformly in $(\alpha,\beta)\in\Gamma\times\Gamma$ such that:
\begin{align*}
	\Big[\mathfrak{k}_{-N}\Big[\big[\big(\overline{\Lambda^{\epsilon}_\alpha}\otimes\Lambda^{\epsilon}_\beta\big)\,\tLambda^{\epsilon}]\,-\,\tLambda^{\epsilon}(\alpha,\beta)\Big]\big(\overline{\tau_{-\alpha}\psi_p}\otimes\tau_{-\beta}\psi_q\big)\Big](x,y)=\epsilon\,<\alpha-\beta>^{-N}\Phi_{N,\alpha,\beta}(x,y).
\end{align*}
Moreover, we may obtain an expansion in powers of $\epsilon$ convergent in matrix norm for $\epsilon\in[0,\epsilon_0]$ for some small enough $\epsilon_0>0$, where the main contribution to \eqref{F-B-I-1} is given by:
\begin{align*}
	\big\langle\,\mathfrak{K}\,,\,\tLambda^{\epsilon}(\alpha,\beta)\,\big(\overline{\tau_{-\alpha}\psi_p}\otimes\tau_{-\beta}\psi_q\big)\big\rangle_{\mathscr{S}(\X\times\X)}&=\tLambda^{\epsilon}(\alpha,\beta)\,\big((\tau_{-\alpha}\psi_p)\,,\,(\Int\,\mathfrak{K})\,(\tau_{-\beta}\psi_q)\big)_{L^2(\X)},
\end{align*}
the remainder being one of the terms in the sum of contributions representing $\big[\widetilde{\mathfrak{M}}^{\epsilon}_{\B}[\mathfrak{K}]_{\alpha,\beta}\big]_{p,q}$.
Now let us deal with the two remainder terms in \eqref{F-B-3} and notice that we have to control the following type of norms:
$
\big\|[\Theta^\epsilon_{\B}-\bb1]\,\tpsi^\epsilon_{\alpha,p}\big\|_{L^2(\X)}$, for $(\alpha,p)\in\Gamma\times\underline{\nB}$\,.  
We shall do that by treating separately the two differences:
\beq \label{F-PN-VII-3}
\big\|[\Theta^\epsilon_{\B}-\widetilde{Q}^\epsilon_{\B}]\,\tpsi^\epsilon_{\alpha,p}\big\|_{L^2(\X)},\qquad\big\|[\widetilde{Q}^\epsilon_{\B}-\bb1]\,\tpsi^\epsilon_{\alpha,p}\big\|_{L^2(\X)}.
\eeq
For the first difference we use Lemma \ref{L-theta-tq}, while
for the second norm in \eqref{F-PN-VII-3} let us start from Proposition \ref{P-II-2} and compute using \eqref{F-II-1}:
\begin{align*}
\big[[\widetilde{Q}^\epsilon_{\B}-\bb1]\,\tpsi^\epsilon_{\alpha,p}\big](x)&=\underset{\gamma\in\Gamma,\,q\in\underline{n_\B}}{\sum}\hspace*{-0.3cm}\mathfrak{T}^{\epsilon}_\gamma\psi_q(x)\,\int_{\Z}dz\,\overline{\mathfrak{T}^{\epsilon}_\gamma\psi_q(z)}\,\big(\mathfrak{T}^{\epsilon}_\alpha\psi_{p}(z)\big)\,-\,\mathfrak{T}^{\epsilon}_\alpha\psi_{p}(x)\\
&=\epsilon\hspace*{-0.3cm}\underset{\gamma\ne\alpha,\,q\in\underline{n_\B}}{\sum}\hspace*{-0.3cm}\big[[\tilde{\mathfrak{m}}^{\epsilon}_{\B}]_{\alpha,\beta}\big]_{p,q}\,\mathfrak{T}^{\epsilon}_\gamma\psi_q(x).
\end{align*}
Appendix \ref{SSS-PN-VI-2} may be used to obtain precise estimations and an $\epsilon$-expansion  for the remainder $\big[[\tilde{\mathfrak{m}}^{\epsilon}_{\B}]_{\alpha,\beta}\big]_{p,q}$.		
\end{proof}

\paragraph{The infinite matrix of the effective Hamiltonian.}
Let us study the matrix of the {\it ``reduced magnetic Hamiltonian"} $\mathfrak{M}^{\epsilon}_{\B}[\ham^\epsilon_{\B}]$, using Proposition \ref{P-f-2}, formula \eqref{F-PN-V-10} and the equality:
\begin{align*}
	&\big((\tau_{-\alpha}\psi_p)\,,\,\big[\Int\,\mathfrak{K}[h^\circ]\big]\,(\tau_{-\beta}\psi_q)\big)_{L^2(\X)}=\big((\tau_{-\alpha}\psi_p)\,,\,\big[\Int\,\mathfrak{K}[h_{\B}]\big]\,(\tau_{-\beta}\psi_q)\big)_{L^2(\X)}.
\end{align*}
We get
\beq\label{F-PN-V-11}
\begin{aligned}
\big[\mathfrak{M}^{\epsilon}_{\B}[\ham^\epsilon_{\B}]_{\alpha,\beta}\big]_{p,q}&=\big(\mathfrak{T}^{\epsilon}_\alpha\psi_p\,,\,H^\epsilon\mathfrak{T}^{\epsilon}_\beta\, \psi_q\big)_{L^2(\X)}=\big(\mathfrak{T}^{\epsilon}_\alpha\psi_p\,,\,\big[\Int\,\tLambda^{\epsilon}\mathfrak{K}[h^\circ]\big]\mathfrak{T}^{\epsilon}_\beta\, \psi_q\big)_{L^2(\X)}\\ 
&\ =\tLambda^{\epsilon}(\alpha,\beta)\Big[\big[\mathring{\mathfrak{m}}_{\B}[H_{\B}]_{\alpha-\beta}\big ]_{p,q}+\epsilon\big[\widetilde{\mathfrak{M}}^{\epsilon}_{\B}\big[\mathfrak{K}[h^\circ]\big]_{\alpha,\beta}\big]_{p,q}\Big]\,,
\end{aligned}
\eeq
with $\quad\widetilde{\mathfrak{M}}^{\epsilon}_{\B}\big[\mathfrak{K}[h^\circ]\big]_{\alpha,\beta}\big]_{p,q}=$:
\begin{align}\nonumber
&=\epsilon^{-1}\Lambda^\epsilon(\alpha,\beta)\Big\langle\,\mathfrak{K}[h^\circ]\,,\,\Big[1-\Omega^\epsilon(\alpha,\cdot,\cdot)\big(1\otimes\Omega^\epsilon(\alpha,\beta,\cdot)\big)\Big]\,\big(\overline{\tau_{-\alpha}\psi_p}\otimes\tau_{-\beta}\psi_q\big)\Big\rangle_{\mathscr{S}(\X\times\X)}\\ \nonumber
&\ =-i\Big\langle\,\mathfrak{K}[h^\circ]\,,\,\Big[\Big(\int_{<\alpha,\cdot,\cdot>}\,B^\epsilon\Big)+\Big(1\otimes\int_{<\alpha,\beta,\cdot>}\,B^\epsilon\Big)\Big]\big(\overline{\tau_{-\alpha}\psi_p}\otimes\tau_{-\beta}\psi_q\big)\Big\rangle_{\mathscr{S}(\X\times\X)}+\,\mathcal{O}(\epsilon)\\ \label{F-PN-V-12}
&\qquad\,=:\,\big[\mathfrak{m}^{\epsilon}_{\B}\big[\mathfrak{K}[h^\circ]\big]_{\alpha,\beta}\big]_{p,q}\,+\,\mathcal{O}(\epsilon).
\end{align}
with $\mathfrak{m}^{\epsilon}_{\B}[\mathfrak{K}[h^\circ]]$ given explicitely in Formula \eqref{F-PN-VI-1}.\\

We have thus obtained the proof of points \ref{point1} and \ref{point2} in Theorem \ref{T-I}.

\subsection{Connection with the Peierls-Onsager substitution procedure.}

Let us very briefly summarize the main results of this section, giving an informal short version of the statement  2 in  Theorem \ref{T-I}.

	Given a $\Gamma$-periodic pseudo-differential operator with an isolated Bloch family $\B$, the dynamics of the states in the subspace  associated with  the family $\B$ is given by an infinite Toeplitz matrix $\mM_{\B}[H_{\B}]$ (see \eqref{F-PN-V-10}) defined by a rapidly decaying $\Gamma$-indexed sequence $\mathring{\mm}_{\B}[H_{\B}]$ (see \eqref{F-PN-2}) (i.e. $\mM_{\B}[H_{\B}]_{\alpha,\beta}=\mathring{\mm}_{\B}[H_{\B}]_{\alpha-\beta}\in\mathscr{M}(\nB)$). Superposing a weak, regular long-range magnetic field, one can still isolate a "quasi-invariant" subspace $\mathfrak{L}^\epsilon_{\B}$ (Definition \ref{D-II-1}), with an error of the order of magnitude of the magnetic field, and an {\it 'effective Hamiltonian"} approximating the real dynamics and given by an infinite $\Gamma$-indexed matrix. Let us emphasize here that the singular part of this approximation, coming from the long-range character of the magnetic field, is treated as a magnetic quantization of the free matrix valued symbol $\mathring{\mm}_{\B}[H_{\B}]$ in the spirit of the Peierls-Onsager procedure, isolating also a 'regular' part that is small in norm of the order of magnitude of the magnetic field (see \eqref{F-PN-V-11} and \eqref{F-PN-V-12}):
	\beq \label{F-PO}
	\mM^\epsilon_{\B}[\ham^\epsilon_{\B}]_{\alpha,\beta}=\tLambda_{\Gamma}^\epsilon(\alpha,\beta)\mathring{\mm}_{\B}[H_{\B}]_{\alpha-\beta}\,+\,\epsilon\tLambda_{\Gamma}^\epsilon(\alpha,\beta)\widetilde{\mm}^\epsilon_{\B}\big[\mathfrak{K}[h^\circ]\big]_{\alpha,\beta}\,+\,\mathcal{O}(\epsilon^2)\,,
	\eeq
	(where we have denoted by $\tLambda_{\Gamma}$ the restriction of $\tLambda:\X\times\X\rightarrow\mathbb{S}^1$ to $\Gamma\times\Gamma$). 	
	In fact, the 'singular part' in \eqref{F-PO} may be written as:	
	\begin{align}
		\tLambda_{\Gamma}^\epsilon(\alpha,\beta)\mathring{\mm}_{\B}[H_{\B}]_{\alpha-\beta}&=\tLambda_{\Gamma}^\epsilon(\alpha,\beta)\underset{\gamma\in\Gamma}{\sum}[\mathring{U}(-\gamma)]_{\alpha,\beta}\otimes\big(\mathcal{F}_{\T_*}\widehat{\mm}_{\B}[\widehat{H}_{\B}]\big)(\gamma)\\
		&=\underset{\gamma\in\Gamma}{\sum}\big[\tLambda_{\Gamma}^\epsilon\,\mathring{U}(-\gamma)\big]_{\alpha,\beta}\otimes\big(\mathcal{F}_{\T_*}\widehat{\mm}_{\B}[\widehat{H}_{\B}](\gamma)\,.
	\end{align}
	This may be considered as the matrix of the following operator in $\ell^2(\Gamma)\otimes\Co^{\B}$:
	\beq
	\tOp^{\epsilon}\big(\widehat{\mm}_{\B}[\widehat{H}_{\B}]\big)\,:=\,\underset{\gamma\in\Gamma}{\sum}\big[\tLambda_{\Gamma}^\epsilon\,\mathring{U}(-\gamma)\big]\otimes\big(\mathcal{F}_{\T_*}\widehat{\mm}_{\B}[\widehat{H}_{\B}]\big)(\gamma)\,,
	\eeq
	that we may compare with the unperturbed one in \eqref{F-f-4}.

\section{Dynamics of the perturbed isolated Bloch family.}
\label{S-4}

Finally we proceed to compare the real perturbed Hamiltonian $H^\epsilon$ (see \eqref{D-Heps}) with the \textit{``magnetic perturbation of the isolated Bloch family Hamiltonian"} \eqref{DF-m-band-H}.

\subsection{Reduction to the subspace ${P^{\epsilon}_\B\,L^2(\X)}$}

\subsubsection{Main statement} This subsection is devoted to the proof of point \ref{point2}{\it i} in Theorem \ref{T-I} that is  contained in the following statement:
\begin{theorem}\label{T-I-1} Given a Hamiltonian $H^\circ$ as in \eqref{D-Hcirc} satisfying Hypotheses \ref{H-isBf}, a magnetic field as in \eqref{HF-Beps} and the perturbed Hamiltonian $H^{\epsilon}$ as in \eqref{D-Heps}, there exists $\delta_0>0$ such that for any $\delta\in(0,\delta_0]$, if we denote by $J^\delta_\B:=\big(E_-+{2}\delta\,,\,E_+-{2}\delta\big)$, there exist constants $\epsilon_0 >0$ and $C>0$ such that, for any $\epsilon\in[0,\epsilon_0]$,
	there exists an orthogonal projection $P^{\epsilon}_\B$ and an {effective} magnetic Hamiltonian $\ham^{\epsilon}_\B:=P^{\epsilon}_\B\,H^{\epsilon}\,P^{\epsilon}_\B$ commuting with $P^{\epsilon}_\B$, satisfying:
	\begin{enumerate}
		\item For any $\lambda\in J^\delta_\B$ the operator $(\bb1-P^{\epsilon}_\B)( H^{\epsilon}-\lambda\bb1)(\bb1-P^{\epsilon}_\B)$ is invertible in the subspace $(\bb1-P^{\epsilon}_\B)L^2(\X)$ and its inverse {denoted by} $[R^\bot_{\epsilon}(\lambda)]$ is norm-bounded  in $\mathbb{B}\big((\bb1-P^{\epsilon}_\B)L^2(\X)\big)$ for $\lambda$ in any compact subinterval $J\subset\,J^\delta_\B$. 
		\item We have the {equality}:
		\[J\cap\sigma(H^{\epsilon}) =  J\cap\sigma\Big(\big \{\ham^{\epsilon}_\B\,-\,P^{\epsilon}_\B H^{\epsilon}[R^\bot_{\epsilon,\cc}(\lambda)]H^{\epsilon}P^{\epsilon}_\B\big \}\big|_{P^{\epsilon}_{\B}\,L^2(\X)}\Big)\,.\]
		\item For any $\lambda\in\,J^\delta_\B\setminus\sigma(H^{\epsilon})$, considering the orthogonal decomposition $$L^2(\X)=P^{\epsilon}_{\B}\,L^2(\X)\,\oplus\,(\bb1-P^{\epsilon}_{\B})L^2(\X)$$ and  denoting by $[R^\sim_{\epsilon}(\lambda)]$ the inverse {in $P^{\epsilon}_\B\,L^2(\X)$} of the operator $$\ham^{\epsilon}_\B\,-\,P^{\epsilon}_\B H^{\epsilon}[R^\bot_{\epsilon}(\lambda)]H^{\epsilon}P^{\epsilon}_\B\,-\,\lambda\,P^{\epsilon}_\B$$ we have the {Feshbach}-Schur block decomposition:
		\begin{align*}
			&\big(H^{\epsilon}-\lambda\bb1\big)^{-1}\\ 
			&\qquad = 
			\left(\begin{array}{cc}
				[R^\sim_{\epsilon}(\lambda)] & -[R^\sim_{\epsilon}(\lambda)]H^{\epsilon}[R^\bot_{\epsilon}(\lambda)] \\
				-[R^\bot_{\epsilon}(\lambda)] H^{\epsilon}[R^\sim_{\epsilon}(\lambda)] & [R^\bot_{\epsilon}(\lambda)]+ [R^\bot_{\epsilon}H^{\epsilon}[R^\sim_{\epsilon}(\lambda)] H^{\epsilon}[R^\bot_{\epsilon}(\lambda)]
			\end{array}\right)
		\end{align*}
		
		and the estimate:
		\beq \label{F-est-FS-epsilon}
		\big\|P^{\epsilon}_\B H^{\epsilon}[R^\bot_{\epsilon}( \lambda)]H^{\epsilon}P^{\epsilon}_\B\big\|_{\mathbb{B}(L^2(\X))}\,\leq \,C  \epsilon^2\,\big\|R^\bot_{\epsilon}(\lambda)\big\|_{\mathbb{B}(L^2(\X))}.
		\eeq
	\end{enumerate}
\end{theorem}

\begin{remark}\label{R-ext-T-I}
	We can extend the estimate \eqref{F-est-FS-epsilon} and point 3  of the above theorem {to}  any complex number $\lambda$ in $\{\zz\in \mathbb C\,,\, \Re\hspace*{-1pt}{\cal{e}}\zz\in\mathring{J}_\B\}$, with uniform {estimates} with respect to $\Im\hspace*{-1pt}\mathcal{m}\,\lambda$.
\end{remark}

The main {ingredient} in the proof of Theorem \ref{T-I} is the abstract spectral  result presented in Paragraph~\ref{A-SchurC} (that we also {used} in \cite{CHP-2, CHP-3}), together with the magnetic pseudodifferential calculus. We choose $P^\epsilon_{\B}$ as in \eqref{D-II-1} and $\ham^\epsilon_{\B}$ as in \eqref{DF-m-band-H}.

\subsubsection{Schur complement and reduction to a quasi-invariant subspace.}\label{A-SchurC}

We shall consider the following abstract setting already introduced and used in \cite{CHP-2,CHP-4}.

\begin{hypothesis}\label{H-II}
	In a separable complex Hilbert space $\mathcal{H}$ we consider a family of pairs $(H_\kappa,P_\kappa)$ indexed by $\kappa\in[0,\kappa_0]$ for some $\kappa_0>0$, where $H_\kappa:\mathcal{D}(H_\kappa)\rightarrow\mathcal{H}$ is a lower-semibounded self-adjoint operator and $\Pi_\kappa=\Pi_\kappa^*=\Pi_\kappa^2$ is an orthogonal projection, such that, for  $\Pi_\kappa^\bot:=\bb1-\Pi_\kappa$, we have the properties:
	\begin{enumerate}
		\item there exists $C>0$ such that for any $\kappa\in[0,\kappa_0]$  we have that  $\Pi_\kappa\mathcal{H}\subset\mathcal{D}(H_\kappa)$ and $ \|\Pi_\kappa^\bot H_\kappa\Pi_\kappa \|_{\mathbb B(\mathcal{H})}\, \leq\,C\,\kappa$\,; 
		\item there exists an interval $J\subset\mathbb{R}$ with non-void interior, such that for any $\kappa\in[0,\kappa_0]$ and any $t\in J$ the operator $\Pi_\kappa^\bot H_\kappa\Pi_\kappa^\bot\,-\,t\Pi_\kappa^\bot$ is invertible  as operator in $\Pi_\kappa^\bot\mathcal{H}$ with the inverse being uniformly bounded on $J$.
	\end{enumerate}
\end{hypothesis} 

We notice that under Hypothesis \ref{H-II} we have the identity:
\[\nonumber 
\Pi_\kappa H_\kappa R^\bot_\kappa(t)H_\kappa\Pi_\kappa=\Pi_\kappa H_\kappa \Pi_\kappa^\bot R^\bot_\kappa(t)\Pi_\kappa^\bot H\Pi_\kappa\,\in\,\mathbb{B}(\mathcal{H})
\] 
and the estimate (for some $C>0$):
\beq \label{F-est-FS}
\big\|\Pi_\kappa H_\kappa R^\bot_\kappa(t)H_\kappa\Pi_\kappa\big\|_{\mathbb{B}(\mathcal{H})}\,\leq\,C \kappa^2\,\big\|R^\bot_\kappa(t)\big\|_{\mathbb{B}(\mathcal{H})}.
\eeq

A simple algebraic computation allows us to prove the following statement about the Schur complement  (\cite{CH}).

\begin{proposition}\label{P-SchurC} 
	Under Hypothesis \ref{H-II} we have that:
	\begin{itemize}
		\item $t\in J\cap\sigma(H)$ if and only if $t\in J\cap\sigma\big(\Pi_\kappa H_\kappa\Pi_\kappa\,-\,\Pi_\kappa HR^\bot_\kappa(t)H_\kappa\Pi_\kappa\big)$ \,;
		\item if we denote by $R^\sim_\kappa(t):=\big(\Pi_\kappa(H_\kappa-t)\Pi_\kappa\,-\,\Pi_\kappa H_\kappa R^\bot_\kappa(t)H_\kappa\Pi_\kappa\big)^{-1}$ as operator in $\Pi_\kappa\mathcal{H}$, we have the identity
		\vspace*{-1cm}
		
		\[\begin{split}
			\big(H_\kappa-t\bb1_{\mathcal{H}}\big)^{-1}&=\left(\begin{array}{cc}
				\Pi_\kappa(H_\kappa-t)\Pi_\kappa & \Pi_\kappa H_\kappa\Pi_\kappa^\bot \\
				\Pi_\kappa^\bot H_\kappa\Pi_\kappa & \Pi_\kappa^\bot(H_\kappa-t)\Pi_\kappa^\bot
			\end{array}\right)^{-1}\\ &=\  
			\left(\begin{array}{cc}
				R^\sim_\kappa(t) & -R^\sim_\kappa(t)H_\kappa R^\bot_\kappa(t) \\
				-R^\bot_\kappa(t) HR^\sim_\kappa(t) & R^\bot_\kappa(t)+ R^\bot_\kappa(t)H_\kappa R^\sim_\kappa(t) H_\kappa R^\bot_\kappa(t)
			\end{array}\right).
		\end{split}\] 
	\end{itemize}
\end{proposition}

\begin{corollary}\label{C-SchurC} 
	Under  Hypothesis \ref{H-II}, the operator $\ham_{\kappa}:=\Pi_{\kappa}H_\kappa\Pi_{\kappa}$ defines  a bounded self-adjoint operator acting in $\Pi_{\kappa}\mathcal{H}$ and there exists $C>0$ such that, for all $\kappa\in [0,\kappa_0]$, 
	$$
	\max\left\{\underset{\lambda\in J\cap\sigma(H_\kappa)}{\sup}\dist\Big(\lambda\,,\,\sigma\big(\ham_{\kappa}\big)\Big)\ ,\ \underset{\lambda\in J\cap\sigma(\ham_{\kappa})}{\sup}\dist\Big(\lambda\,,\,\sigma(H_\kappa)\Big)\right\}\,\leq\,C\,\kappa^2\,\big\|R^\bot_\kappa(t)\big\|_{\mathbb{B}(\mathcal{H})}.
	$$
\end{corollary}

\subsubsection{Verification of Hypothesis \ref{H-II} and end of the proof of Theorem \ref{T-I}}

Let us recall the definitions and notations in \ref{SS-free-dyn}, \ref{SS-3.5} and \ref{SS-4.1} and using the context of Paragraph  \ref{A-SchurC}, let us take $\mathcal{H}=L^2(\X)$, $\kappa=\epsilon$, $\Pi_\kappa=P^{\epsilon}_\B$ and $H_\kappa=H^{\epsilon}$ for $P^\epsilon_{\B}$ as in \eqref{D-II-1}, for any $\epsilon\in[0,\epsilon_0]$. 

\paragraph{Verification of condition 1.} 

The previous Propositions \ref{P-IV-1} and \ref{P-III-1} clearly imply that $P^\epsilon_{\B}L^2(\X)\subset\mathcal{D}(H^\epsilon)$ while Formula \eqref{P-IV-2} clearly implies that:
\[
(\bb1-P^\epsilon_{\B})H^\epsilon{P}^\epsilon_{\B}=(\bb1-P^\epsilon_{\B})\big[H^\epsilon\,,\,P^\epsilon_{\B}\big]=\epsilon(\bb1-P^\epsilon_{\B})\,\mathfrak{X}^\epsilon_{\B}+\mathcal{O}(\epsilon^2)\,,
\]
and thus the desired estimation in the first condition of Hypothesis \ref{H-II}.

\paragraph{Verification of condition 2.} 

Let us consider the product $P^{\epsilon}_\bot H^{\epsilon}P^{\epsilon}_\bot$. Using Remark \ref{R-p-symb} we notice that $H_\bot=P_\bot HP_\bot$ has two spectral gaps: $[0,E_0]$ and $[E_-,E_+]$, with $0<E_0\leq E_-<E_+<\infty$. We shall
prove that condition 2 in Hypothesis \ref{H-II} is verified for some interval $J\subset[E_-,E_+]$ with non-void interior and the  pair $\big(H^{\epsilon}\,,\,P^{\epsilon}_\B\big)$ in $\mathcal{H}=L^2(\X)$ where we consider, as explained above, $\kappa=\epsilon$. Then, let us introduce:
\[
P^{\epsilon}_\bot:=\bb1-P^\epsilon_{\B}
\]
and study the difference:
\begin{align*}
P^{\epsilon}_\bot H^{\epsilon}P^{\epsilon}_\bot\,-\,H^\epsilon_\bot\,&=\,H^\epsilon-H^\epsilon_\bot+P^\epsilon_{\B}H^\epsilon{P}^\epsilon_{\B}-\big(P^\epsilon_{\B}H^\epsilon+H^\epsilon{P}^\epsilon_{\B}\big)\\
&=\,H^\epsilon-H^\epsilon_\bot+P^\epsilon_{\B}H^\epsilon{P}^\epsilon_{\B}+\big[H^\epsilon\,,\,P^\epsilon_{\B}\big]-2H^\epsilon{P}^\epsilon_{\B}\\
&=\,H^\epsilon-H^\epsilon_\bot-P^\epsilon_{\B}H^\epsilon{P}^\epsilon_{\B}+\big[H^\epsilon\,,\,P^\epsilon_{\B}\big]-2P^\epsilon_\bot{H}^\epsilon{P}^\epsilon_{\B}\\
&=\,H^\epsilon_{\B}-P^\epsilon_{\B}H^\epsilon{P}^\epsilon_{\B}-\epsilon\big(\bb1+2P^\epsilon_{\B}\big)\mathfrak{X}^\epsilon_{\B}+\mathcal{O}(\epsilon^2).
\end{align*}
Then we notice that:
\begin{align*}
	H^\epsilon_{\B}-P^\epsilon_{\B}H^\epsilon{P}^\epsilon_{\B}&=\Op^\epsilon\big(h_{\B}-p^\epsilon_{\B}\sharp^\epsilon{h}\sharp^\epsilon{p}^\epsilon_{\B}\big)=\Op^\epsilon\big(p_{\B}\sharp^\circ{h}\sharp^\circ{p}_{\B}-p^\epsilon_{\B}\sharp^\epsilon{h}\sharp^\epsilon{p}^\epsilon_{\B}\big)\\
	&=\Op^\epsilon\big(p_{\B}\sharp^\circ{h}\sharp^\circ{p}_{\B}-p_{\B}\sharp^\epsilon{h}\sharp^\epsilon\p_{\B}+p_{\B}\sharp^\epsilon{h}\sharp^\epsilon{p}_{\B}-p^\epsilon_{\B}\sharp^\epsilon{h}\sharp^\epsilon{p}^\epsilon_{\B}\big)\\
	&=\Op^\epsilon\big((p_{\B}\sharp^\circ{h})\sharp^\circ{p}_{\B}-(p_{\B}\sharp^\circ{h})\sharp^\epsilon{p}_{\B}\big)+\Op^\epsilon\big((p_{\B}\sharp^\circ{h}-p_{\B}\sharp^\epsilon{h})\sharp^\epsilon{p}_{\B}\big)\\
	&\quad+\Op^\epsilon\big((p_{\B}-p^\epsilon_{\B})\sharp^\epsilon{h}\sharp^\epsilon{p}^\epsilon_{\B}\big)+\Op^\epsilon\big({p}_{\B}\sharp^\epsilon{h}\sharp^\epsilon(p_{\B}-p^\epsilon_{\B})\big)
\end{align*}
and using also Proposition \ref{P-replP3_5} we conclude that:
\beq\begin{split} \label{F-dif-h-bot}
	&H^\epsilon_{\B}-P^\epsilon_{\B}H^\epsilon{P}^\epsilon_{\B}=\\ &\ =\epsilon\Op^\epsilon\big(\mathcal{r}_{\epsilon}(p_{\B},h_{\B})+\mathcal{r}_{\epsilon}(h,p_{\B})\sharp^\circ{p}_{\B}+(p_{\B}-p^\epsilon_{\B})\sharp^{\circ}p_{\B}+p_{\B}\sharp^{\circ}(p_{\B}-p^\epsilon_{\B})\big)\,+\,\mathcal{O}(\epsilon^2)\,.
\end{split}\eeq

The arguments in Subsection \ref{SS-free-dyn} imply that: $\sigma\big(H_\bot\big)\subset\{0\}\cup[E_0,E_- ]\cup [ E_+,+\infty)$. Thus we can use {the last inclusion listed in} \eqref{F-m-sp-est} in order to conclude that there exists $\delta_0>0$ such that for any $\delta\in(0,\delta_0]$ there exist constants $\epsilon_0(\delta)>0$ and $C>0$ such that for any $\epsilon\in[0,\epsilon_0]$:
$$
J:=\big(E_-+\delta\,,\,E_+-\delta\big)\subset\,\mathbb{R}\setminus\sigma\big(H^{\epsilon}_\bot\big).
$$ 
Thus we have proved the following statement:
\begin{proposition}\label{P-red-magn-Hamilt-0}
	Let $H^\circ$ in \eqref{D-Hcirc} satisfy Hypothesis \ref{H-isBf} and a perturbation by a magnetic field verifying  \eqref{HF-Beps}. With $J_{\B}:=[E_-,E_+]\subset\R$ as in Hypothesis \ref{H-isBf},   there exists $\delta_0>0$  such that for any $\delta\in(0,\delta_0]$ the interval $J^{\delta}_\B:=\big(E_-+\delta\,,\,E_+-\delta\big)$ is not void and there exists a constant $\epsilon_0(\delta)>0$ such that for any $\epsilon\in[0,\epsilon_0]$ and for any $t\in J^{\delta}_\B$ the operator   $(\bb1-P^{\epsilon}_\B)\big(H^{\epsilon}_\bot\,-\,t\bb1\big)(\bb1-P^{\epsilon}_\B)$ is invertible as a self-adjoint operator in  $(\bb1-P^{\epsilon}_\B)L^2(\X)$. This  inverse is denoted by $[R^\bot_{\epsilon}(t)]$ and is uniformly bounded on $J^{\delta}_\B$. 
\end{proposition}

Putting this result together with \eqref{F-dif-h-bot}, we obtain {that the} Condition 2 in Hypothesis~\ref{H-II} is satisfied, and Theorem \ref{T-I} is then a direct consequence of Proposition \ref{P-SchurC} and Corollary~\ref{C-SchurC}.

\subsection{The effective time evolution}\label{S-ev}

The proof of point \ref{point2}\textit{ii} in Theorem \ref{T-I} will follow from the following theorem.
\begin{theorem}\label{C-T-I} 
	Under the hypotheses of Theorem \ref{T-I}, given any  compact interval $J\subset J^\delta_\B$ (with $J^\delta_\B$ as in the Theorem \ref{T-I}), if we denote by $E_J( H^{\epsilon})$ the spectral projection of $H^\epsilon$ for the interval $J\subset\R$, there exist $C>0$ {and $\epsilon_0 >0$ } such that for any $\epsilon\in [0,\epsilon_0]$ and for all $v$ in the range of $E_J( H^{\epsilon})$, we have the estimation: 
	\beq \label{F-ev-est}
	\big\|e^{-itH^{\epsilon}}v - e^{-it\ham^{\epsilon}_{\B}}v\big\|_{L^2(\X)}\,\leq\,C\,\Big[\epsilon\,+\,{(1+|t|) ^3}\,\epsilon^2\Big]\,\|v\|_{L^2(\X)},\ \forall t\in\R\,.
	\eeq
\end{theorem}
\begin{proof}
	In order to prove \eqref{F-ev-est}, we consider states $v$ with energies in a compact interval $J\subset J_\B$, and we fix two cut-off functions $\varphi$ and $\widetilde{\varphi}$ in $C^\infty_0(\mathbb{R})$ that are equal to $1$ on $J$, have their support included in $J^\delta_\B$ and verify the equality $\varphi=\widetilde{\varphi}\,\varphi$. 
	For any $v\in E^{\epsilon}_h(J)\mathcal{H}$ we may write:
	$$
	e^{-itH^{\epsilon}}v=E^{\epsilon}_h(J) e^{-itH^{\epsilon}}{\varphi}\big(H^{\epsilon}\big)E^{\epsilon}_h(J)v=\widetilde{\varphi}(H^{\epsilon})e^{-itH^{\epsilon}}{\varphi}\big(H^{\epsilon}\big)\widetilde{\varphi}(H^{\epsilon})v.
	$$
	For any $t\in\mathbb{R}$ let us define $\varphi_t\in C^\infty_0(\R)$ by $\varphi_t(s):=e^{-its}\varphi(s)$, so that the above equality becomes:
	\beq \label{F-t-evol}
	e^{-itH^{\epsilon}}v=\widetilde{\varphi}(H^{\epsilon}) \varphi_t\big(H^{\epsilon}\big)\widetilde{\varphi}(H^{\epsilon})v.
	\eeq
	
	We use now the Helffer-Sj\"{o}strand formula (see \cite{HS,D-95}). 
	For that we fix an auxiliary cut-off function $\chi\in C^\infty_0(\R;[0,1])$ having support in $\{|t|\leq 2\}$ and being equal to 1 on $\{|t|\leq1\}$ and define (notice a slight difference with \cite{D-95} in the choice of such a quasi-analytic extension):
	\[\nonumber 
	\overset{\frown}{\varphi}_{t,N}:\Co\rightarrow\R,\quad\overset{\frown}{\varphi}_{t,N}(x+iy):=\underset{0\leq k\leq N}{\sum}(\partial^k\varphi_t)(x)(iy)^k(k!)^{-1}\chi(y).
	\]
	Then for any $N\in\mathbb{N}$, the support of $\overset{\frown}{\varphi}_{t,N}$ is a compact set contained in $\supp\varphi_t\times[-2,2]$ and $\overset{\frown}{\varphi}_{t,N}$ is smooth. Moreover:
	\begin{equation}\label{dhc26}\begin{split}
		\frac{\partial\overset{\frown}{\varphi}_{t,N}}{\partial\overline{\z}}(\z)&=\frac{1}{2}\left(\frac{\partial\overset{\frown}{\varphi}_{t,N}}{\partial x}+i\frac{\partial\overset{\frown}{\varphi}_{t,N}}{\partial y}\right)\\
		&\hspace*{-2cm}=\frac{1}{2}\Big[i\underset{0\leq k\leq N}{\sum}(\partial^k\varphi_t)(x)(iy)^k(k!)^{-1}[(\partial\chi)(y)]+(\partial^{N+1}\varphi_t)(x)(iy)^N(N!)^{-1}\chi(y)\Big]\,,
	\end{split}\end{equation}
	and we see that for any $x\in\R$:
	$$
	\underset{y\rightarrow0}{\lim}\,\, \, |y|^{-N}\, \big | (\partial_{\overline{\zz}}\overset{\frown}{\varphi}_{t,N})(x+iy)\big | = \big |(\partial^{N+1}\varphi_t)(x)(N!)^{-1}\big |<\infty.
	$$
	
	We notice the important fact that $(\partial^k\varphi_t)(x)$ is a polynomial of degree $k$ in $t\in\R$ whose coefficients are smooth complex functions of $x\in\R$ having all their support contained in $\supp(\varphi)$. As functions of $t$, these terms can grow at most like $<t>^N$.
	
	Let us use \eqref{F-t-evol} and Proposition \ref{P-est-Jepsilon} in order to get:
	\begin{equation*}
		\begin{array}{ll}
			e^{-itH^{\epsilon}}v&=\widetilde{\varphi}(H^{\epsilon}) \varphi_t\big(H^{\epsilon}\big)\widetilde{\varphi}(H^{\epsilon})v=P^{\epsilon}_\B\,\widetilde{\varphi}(H^{\epsilon}) \varphi_t\big(H^{\epsilon}\big)\widetilde{\varphi}(H^{\epsilon})\,P^{\epsilon}_\B\,v\,+\,\epsilon \,X_{\epsilon}v\\
			&=P^{\epsilon}_\B {\varphi}_t\big(H^{\epsilon}\big)P^{\epsilon}_\B v\,+\,\epsilon \,X_{\epsilon}v\\
			&={\pi^{-1}} \iint\,\big(-\frac i2  d\zz d\bar \zz\big)\,\big(\partial_{\overline{\zz}}\overset{\frown}{\varphi}_{t,N}\big)(\zz,\overline{\zz})\,P^{\epsilon}_\B \big(H^{\epsilon}-\zz\bb1\big)^{-1}P^{\epsilon}_\B v\,+\,\epsilon \,X_{\epsilon}v\\
			&={\pi^{-1}} \iint\,\big(-\frac i2  d\zz d\bar \zz\big)\,\big(\partial_{\overline{\zz}}\overset{\frown}{\varphi}_{t,N}\big)(\zz,\overline{\zz})\,P^{\epsilon}_\B R^\sim_{\epsilon}(\zz)P^{\epsilon}_\B v\,+\,\epsilon \,X_{\epsilon}v
		\end{array}
	\end{equation*}
	with $\|X_{\epsilon}\|_{\mathbb{B}(L^2(\X))}\leq 1$ uniformly for $\epsilon\in[0,\epsilon_0]$.
	For $\epsilon$ small enough we also have the identity 
	$$e^{-it\ham^{\epsilon}_{\B}}v=e^{-it\ham^{\epsilon}_{\B}}P^{\epsilon}_\B\,\widetilde{\varphi}(H^{\epsilon})\,v\,+\,\mathscr{O}(\epsilon)\,v= \varphi_t\big(\ham^{\epsilon}_{\B}\big)P^{\epsilon}_\B\,v\,+\,\mathscr{O}(\epsilon)\,v.$$
	Hence we may compute:
	\begin{align*}
		\big\|e^{-itH^{\epsilon}}v\,-\,e^{-it\ham^{\epsilon}_{\B}}v\big\|_{L^2(\X)}\,&{=}\,\big\|P^{\epsilon}_\B {\varphi}_t\big(H^{\epsilon}\big)P^{\epsilon}_\B v\,-\,P^{\epsilon}_\B \varphi_t\big(\ham^{\epsilon}_{\B}\big)P^{\epsilon}_\B v\big\|_{L^2(\X)}\,+\,\mathscr{O}(\epsilon)\|v\|_{L^2(\X)}\\
		&\hspace*{-5cm}\leq \iint\big(\frac{d\zz d\bar \zz}{2\pi}  \big)\,\big|\big(\partial_{\overline{\zz}}\overset{\frown}{\varphi}_{t,N}\big)(\zz,\overline{\zz})\big|\,\Big\|P^{\epsilon}_{\B}\Big[R^{\epsilon}(\zz)\,-\, \big(\ham^{\epsilon}_{\B}-\zz P^{\epsilon}_{\B}\big)^{-1}\Big]\,P^{\epsilon}_\B v\Big\|_{L^2(\X)}+\,\mathscr{O}(\epsilon)\|v\|_{L^2(\X)}.
	\end{align*}
	Using Theorem \ref{T-I} and Remark \ref{R-ext-T-I},  we notice the identity $P^{\epsilon}_{\B}R^{\epsilon}(\zz)P^{\epsilon}_{\B}=P^{\epsilon}_{\B}R^\sim_{\epsilon}(\zz)P^{\epsilon}_{\B}$ and the estimate:
	\begin{align*}
		&\Big\|P^{\epsilon}_{\B}\Big[R^\sim_{\epsilon}(\zz)\,-\, \big(\ham^{\epsilon}_{\B}-\zz P^{\epsilon}_{\B}\big)^{-1}\Big]\,P^{\epsilon}_\B \Big\|_{L^2(\X)}\\
		&\leq\Big\|P^{\epsilon}_\B H^{\epsilon}[R^\bot_{\epsilon}(\zz)]H^{\epsilon}P^{\epsilon}_\B\Big\|_{L^2(\X)}\leq C(\delta)\,(\Im\hspace*{-1pt}{\cal{m}}\zz)^{-2}\, \epsilon^2\,\big\|R^\bot_{\epsilon}(\zz)\big\|_{\mathbb{B}(L^2(\X))}.
	\end{align*}
	
	Finally, by taking $N=2$ in the definition of $\overset{\frown}{\varphi}_{t,N}$ and using \eqref{dhc26}, we have the bound:
	\begin{align*}
		&\iint\,\big(-\frac i2  d\zz d\bar \zz\big)\,\big|\big(\partial_{\overline{\zz}}\overset{\frown}{\varphi}_{t,2}\big)(\zz,\overline{\zz})\big|\,|\Im\hspace*{-1pt}{\cal{m}}\zz|^{-2}=\iint_{\supp(\widetilde{\varphi}_{t,2})}\,dx\,dy\,\big|\big(\partial_{\overline{\zz}}\overset{\frown}{\varphi}_{t,2}\big)(x+iy)\big|\,|y|^{-2}\\
		&\hspace*{0.5cm}\leq\int_{\supp(\varphi)}\,dx\,\left[\underset{0\leq k\leq 2}{\sum}\big[\underset{x\in\R}{\sup}\big|\big(\partial^{k}\varphi_t\big)(x)\big|\big]\int_{1}^2\,dy\,y^{k-2}\right]\\
		&\hspace*{0.5cm}+\int_{\supp(\varphi)}\,dx\,\big[\underset{x\in\R}{\sup}\big|\big(\partial^{3}\varphi_t\big)(x)\big|\big]\int_{0}^2\,dy.
	\end{align*}
	
	One concludes that there exist  {$C>0$ and $\epsilon_0 >0$ } such that for any $t\in\R$ and for any $\epsilon\in[0,\epsilon_0]$ we have the estimate:
	\begin{align*}
		\big\|e^{-itH^{\epsilon}}v\,-\,e^{-it\ham^{\epsilon}_{\B}}v\big\|_{L^2(\X)}\,&\leq\\
		&\hspace*{-4cm}\leq\,\Big[C\,\epsilon^2\, <t>^3 \big(\hspace*{-0.3cm}\underset{ {\scriptsize \begin{array}{c} \Re\hspace*{-1pt}\mathcal{e}\zz\in \supp(\varphi),\\|\Im\hspace*{-1pt}{\cal{m}}\zz|\leq2\end{array}}}{\sup}\hspace*{-0.3cm}\big\|[R^\bot_{\epsilon}(\zz)]\big\|_{\mathbb{B}(L^2(\X))}\big)\,+\,\mathscr{O}(\epsilon)\Big]\,\|v\|_{L^2(\X)}\, ,\quad\forall v\in E^{\epsilon}_h(J)\mathcal{H}.
	\end{align*}
\end{proof}

\section{Magnetic fields having a non-zero constant component.}\label{S-5}

We dedicate this last section to a more refined analysis of Formula \eqref{F-PO} in order to obtain a better understanding of the 'singular term' $\tLambda_{\Gamma}^\epsilon\mM_{\B}[H_{\B}]$ and of the first order correction in \eqref{F-PO}. For that we shall consider that our perturbing magnetic field is a small perturbation of a weak constant magnetic field $B^\bullet:=dA^\bullet$.

We shall consider a magnetic field $B(x)=\epsilon\big(B^\bullet+\cc{B}^\epsilon(x)\big)$ as in Hypothesis \ref{H-Bepsc}. The vector potential defining $B(x)$ will be taken of the form $A(x)\equiv\epsilon\big(A^\bullet(x)+\cc\,A^\epsilon(x)\big)=\epsilon{A}^{\cc}(x)$.

\subsection{The constant perturbing magnetic field.}

In this subsection we work under Hypothesis \ref{H-Bepsc} with the parameter $\cc=0$, i.e. with a constant magnetic field of the form $B=\epsilon{B}^\bullet$. We keep using the notations \eqref{N-A-eps} but the perturbing magnetic field is now constant. Thus the entire analysis and all the results in the previous section are true, but we shall put into evidence some important specific features valid in this particular case. 
 
The main new feature when dealing with a constant magnetic field comes from the Zak translations \eqref{DF-m-W-f}.

\subsubsection{The Zak translations in constant magnetic field}\label{SS-Z-trsl}

\begin{definition}\label{D-Ztr-cB}
	Given a constant magnetic field of the form $\epsilon\,B^\bullet$, we call \emph{ Zak translations} with vectors from $\Gamma$  the following family of twisted unitary translations on $L^2(\X)$ (using the notation in \eqref{N-A-eps})
	\beq \label{DF-Z-trsl}	{\mathfrak{T}^{\epsilon}_\gamma f (x)} :=\tLambda^\epsilon(x,\gamma)\,f(x-\gamma)=\Lambda^\epsilon_\gamma(x)\,f(x-\gamma),\quad\
	\forall\gamma\in\Gamma.
	\eeq
\end{definition}
\noindent For our constant magnetic field situation, we recall the following results in \cite{CHN}:

\begin{proposition}\label{Zak-transl}
	{The}  family  in \eqref{DF-Z-trsl} satisfies the following properties:
	\begin{enumerate}
		\item The map $\mathfrak{T}^{\epsilon}:\Gamma\ni\gamma\rightarrow\mathfrak{T}^{\epsilon}_\gamma\in\mathbb{U}\big(L^2(\X)\big)$ defines a projective representation with 2-cocycle $[\tLambda^{\epsilon}]^{-1}:\Gamma\times\Gamma\rightarrow\mathbb{S}^1$ i.e.
		$$\mathfrak{T}^{\epsilon}_\alpha\mathfrak{T}^{\epsilon}_\beta= \tLambda^\epsilon(\beta,\alpha)\, \mathfrak{T}^{\epsilon}_{\alpha+\beta}\,.$$
		\item The operator $\mathfrak{Op}^{\epsilon}(F)$ commutes with all the $\{\mathfrak{T}^{\epsilon}_\gamma\}_{\gamma\in\Gamma}$ if and only if $F\in\mathscr{S}^\prime(\Xi)$ is $\Gamma$-periodic with respect to the variable in $\X$.
	\end{enumerate}
\end{proposition}

\subsubsection{The magnetic isolated Bloch family in constant magnetic field.}

The essential consequence following from point 1 in Proposition \ref{Zak-transl} above is that we have the commutation relation:
\[
\forall(\alpha,\gamma)\in\Gamma\times\Gamma:\quad\mathfrak{T}^{\epsilon}_\alpha\,\big[\Int\,\mathfrak{K}^{\epsilon}_{\B,\gamma}\big]\,=\,\big[\Int\,\mathfrak{K}^{\epsilon}_{\B,\gamma}\big]\mathfrak{T}^{\epsilon}_\alpha
\]
implying that $\widetilde{Q}^\epsilon_{\B}$ and all its assoociated functional calculus are in the commutant of the family $\{\mathfrak{T}^{\epsilon}_\gamma\}_{\gamma\in\Gamma}$. It follows that $\Theta^\epsilon$ and also $P^\epsilon_{\B}$ commute with the family $\{\mathfrak{T}^{\epsilon}_\gamma\}_{\gamma\in\Gamma}$ and we evidently have the following statement true.

\begin{proposition}\label{P-P-eps-B-const}
In the case of a constant perturbing magnetic field $B:=\epsilon{B}^\bullet$, the space $\mathfrak{L}^\epsilon_{\B}$ is equal to the closure of the linear span over $\Co$ of the family $\{\mathfrak{T}^{\epsilon}_\gamma\psi^\epsilon_p\}_{(\gamma,p)\in\Gamma\times\nB}$ with:
\beq \label{DF-PN-VI-4}
\psi^\epsilon_p\,=\,\Theta^\epsilon\psi_p\,,\quad\forall{p}\in\nB\,,
\eeq
and
\beq \label{F-P-eps-B-const}
P^\epsilon_{\B}\,=\,\Int\Big[\underset{(p,\gamma)\in\nB\times\Gamma}{\sum}\big(\mathfrak{T}^{\epsilon}_\gamma\psi^\epsilon_p\big)\otimes\big(\overline{\mathfrak{T}^{\epsilon}_\gamma\psi^\epsilon_p}\big)\Big].
\eeq
\end{proposition}

\begin{corollary}\label{C-PN-VII-1}
	The family $\blPsi^\epsilon_{\B}:=\big\{\mathfrak{T}^\epsilon_\alpha\psi^\epsilon_p\ \forall(\alpha,p)\in\Gamma\times\underline{n_\B}\big\}$ is a strongly localized Parseval frame for the closed subspace $\mathcal{L}^\epsilon_\B$ that it generates in $L^2(\X)$.
\end{corollary}
 Notice that this is the equivalent of \eqref{DF-Pfr-B-eps} for the situation when the perturbing magnetic field is constant.\\

We end this subsection with the following technical result about the family of functions $\{\psi^\epsilon_p\}_{p\in\nB}$ introduced in \eqref{DF-PN-VI-4}.
	\begin{lemma}\label{L-1}
		There exists $\epsilon_0>0$ such that for any $N>0$ there exists $C_N>0$ such that for all $(\alpha,p)\in\Gamma\times\underline{n_\B}$ and $\epsilon\in[0,\epsilon_0]$, if we denote by $Q$ the operator of multiplication  with the variable $x\in\X$ in $L^2(\X)$, we have the estimations
		$$
		\big\|<Q-\alpha>^N\Big (\mathfrak{T}^\epsilon_\alpha\psi^\epsilon_p-\mathfrak{T}^\epsilon_\alpha\psi_p\Big )\big\|_{L^2(\X)}\leq C_N\,\epsilon \, \Vert <Q-\alpha>^{N+1}\psi_p\Vert_{L^2(\X)}.
		$$
	\end{lemma}
	\begin{proof}
		Because $\Theta^\epsilon_\B$ commutes with the Zak translations, it is enough to prove this for $\alpha=0$. Since $P_\B \psi_p=\psi_p$, using \eqref{dhc21} we may write 
		\beq \label{dhc20}
		\begin{aligned}
			\Theta^\epsilon_\B\,\psi_p-\psi_p&=-(2\pi i)^{-1}\oint_{\mathscr{C}}d\zz\,\zz^{-1/2}\,\Big (\big(\widetilde{Q}_\B^\epsilon-\zz\bb1\big)^{-1} -\big(P_\B-\zz\bb1\big)^{-1}\Big )\, \psi_p\\
			\\&=(2\pi i)^{-1}\oint_{\mathscr{C}}d\zz\,\zz^{-1/2}\,\big(\widetilde{Q}_\B^\epsilon-\zz\bb1\big)^{-1} \Big (\widetilde{Q}_\B^\epsilon-P_\B\Big )\, \big( P_\B-\zz\bb1\big)^{-1}\, \psi_p.
		\end{aligned}
		\eeq
		{Using Lemma \ref{R-PN-1}, Remark \ref{R-PN-0} and the fact that $|\Lambda^\epsilon(x,y)-1|\leq C\, \epsilon$, we obtain that,  for any $N\geq 1$, there exists $C_N$ and $\epsilon_0 >0$ such that, for any $f$ with rapid decay and any $\epsilon \in [0,\epsilon_0]$,  we have
		$$\Vert <Q>^N\, \big (\widetilde{Q}_\B^\epsilon-P_\B\big )f\Vert_{L^2(\X)}\leq C_N\, \epsilon\, \Vert <Q>^{N+1} f\Vert_{L^2(\X)}.$$
		}
		Also, denoting by $X$ either $P_\B$ or $\widetilde{Q}_\B^\epsilon$ we have, for any  $N\geq 1$, 
		$$\sup_{\zz\in \mathscr{C}}\Big \Vert <\cdot >^N \big (X-\zz\bb1\big)^{-1}\, <\cdot >^{-N}\Big \Vert =C'_N<\infty,$$
		which may be proved by showing that all possible multiple commutators between the position operator and the resolvent can be extended to bounded operators. Implementing all this in \eqref{dhc20} the proof is finished.
	\end{proof}	
	
\subsubsection{The magnetic infinite matrices in a constant perturbing magnetic field}\label{SSS-const-mf-Pfr}
	
We shall put now into evidence the important feature characterizing the 'effective magnetic Hamiltonian' for the isolated Bloch family in a constant magnetic perturbing field. 

For a linear operator $T\in\,\mathbb{B}\big(P^\epsilon_\B\,L^2(\X)\big)$ we use the notation 
	\beq \label{DF-cmf-Bmatrix}
	\mathfrak{M}^\epsilon_\B[T]\,:=\,\big(\mathfrak{T}^\epsilon_\alpha\,\psi^\epsilon_{p}\,,\,T\,\mathfrak{T}^\epsilon_\beta\,\psi^\epsilon_{q}\big)_{L^2(\X)}
	\eeq
	for its infinite matrix defined by the Parseval frame in Corollary \ref{C-PN-VII-1} and notice that this is simply the matrix in \eqref{C-PN-VI-3-1} for the situation when the perturbing magnetic field is constant. The commutation properties discussed in this section for the case of a constant perturbing magnetic field have the following interesting consequence concerning the structure of the matrices $\mM^\epsilon_{\B}[\Op^\epsilon(F)]$ for a $\Gamma$-periodic symbol $F$.
	
	\begin{proposition}\label{P-B2}
		Given a constant magnetic field $B:=\epsilon{B}^\bullet=\epsilon\,dA^\bullet$ and a magnetic operator $\Int\,\tLambda^\epsilon\mathfrak{K}$ with $\mathfrak{K}\in\mathring{\mathscr{S}}_{\Delta}(\X\times\X)$ such that $(\tau_\gamma\otimes\tau_\gamma)\mathfrak{K}=\mathfrak{K}$ for all $\gamma\in\Gamma$, there exists $\mathring{\mathfrak{m}}^\epsilon_\B[\mathfrak{K}]\in{\cal{s}}(\Gamma;\MmN)$ such that:
		\begin{equation}\nonumber 
		\mathfrak{M}^\epsilon_\B[\Int\,\tLambda^\epsilon\mathfrak{K}]_{\alpha,\beta}\,=\,\tLambda^\epsilon(\alpha,\beta)\,\mathring{\mathfrak{m}}^\epsilon_\B[\mathfrak{K}]_{\alpha-\beta}\,.
		\end{equation}
	\end{proposition}
	\begin{proof}
		Taking into account Proposition \ref{Zak-transl},  let us compute:
		\begin{align*}
			\big[\mathfrak{M}^\epsilon_\B[\Int\,\tLambda^\epsilon\mathfrak{K}]_{\alpha,\beta}\big]_{p,q}&=\big(\mathfrak{T}^\epsilon_\alpha\,\psi^\epsilon_p\,,\,[\Int\,\tLambda^\epsilon\,\mathfrak{K}[F]]\,\mathfrak{T}^\epsilon_\beta\,\psi^\epsilon_q\big)_{L^2(\X)}\\
			&=\tLambda^\epsilon(\alpha,\beta)\big(\mathfrak{T}^\epsilon_{\alpha-\beta}\,\psi^\epsilon_p\,,\,[\Int\,\tLambda^\epsilon\,\mathfrak{K}[F]]\,\psi^\epsilon_q\big)_{L^2(\X)}
		\end{align*}
		and identify:
		\[\nonumber 
		\big [\mathring{\mathfrak{m}}^\epsilon_\B[\mathfrak{K}]_\gamma \big ]_{p,q}\,:=\,\big(\mathfrak{T}^\epsilon_{\gamma}\,\psi^\epsilon_p\,,\,[\Int\,\tLambda^\epsilon\,\mathfrak{K}(][F]]\,\psi^\epsilon_q\big)_{L^2(\X)}.
		\]
	\end{proof}
	
	\begin{definition}\label{D-m-eps-B}
	 For any $\epsilon\in[0,1]$ we denote by $\vec{\mathscr{M}}^\epsilon_\Gamma[\mathscr{M}_{n_\B}]$ the complex linear space of the matrices $\mathfrak{M}^\epsilon_\B$ in $\mathscr{M}^\circ_\Gamma[\mathscr{M}_{n_\B}]$ satisfying the following property (as in Proposition \ref{P-B2}):
			$$
			\exists\, \mathring{\mathfrak{m}}^\epsilon_\B\in{\cal{s}}(\Gamma;\MmN)\ \text{such that}\ [\mathfrak{M}^\epsilon_\B]_{\alpha,\beta}\,=\,\tLambda^\epsilon(\alpha,\beta)\,[\mathring{\mathfrak{m}}^\epsilon_\B]_{\alpha-\beta}.
			$$
	\end{definition}
	
	\begin{proposition}\label{P-matr}
		The space $\vec{\mathscr{M}}^\epsilon_{\Gamma}[\MmN]$ is a normed subalgebra of $\mathring{\mathscr{M}}_{\Gamma}[\MmN]$ with involution and unity.
	\end{proposition}
	\begin{proof}
		We reproduce here the proof of (3.28) in \cite{CHN}.
		Given $(S,T)\in\big[\mathscr{M}^\circ_{\Gamma}[\MmN]_\epsilon\big]^2$ we know that $$S_{\alpha,\beta}=\tLambda^\epsilon(\alpha,\beta)\mathring{S}_{\alpha-\beta} \mbox{ and } T_{\alpha,\beta}=\tLambda^\epsilon(\alpha,\beta)\mathring{T}_{\alpha-\beta} \mbox{ with }  (\mathring{S},\mathring{T})\in\big[{\cal{s}}\big(\Gamma;\MmN\big)\big]^2\,.$$
		Hence
		\begin{align*}
			(ST)_{\alpha,\beta}:&=\underset{\gamma\in\Gamma}{\sum}S_{\alpha,\gamma}\,T_{\gamma,\beta}=\underset{\gamma\in\Gamma}{\sum}\tLambda^\epsilon(\alpha,\gamma)\mathring{S}_{\gamma-\alpha}\tLambda^\epsilon(\gamma,\beta)\mathring{T}_{\beta-\gamma}=\underset{\gamma^\prime\in\Gamma}{\sum}\tLambda^\epsilon(\alpha,\gamma^\prime)\tLambda^\epsilon(\gamma^\prime+\alpha,\beta)\mathring{S}_{\gamma^\prime}\mathring{T}_{\beta-\alpha-\gamma^\prime}\\
			&=\tLambda^\epsilon(\alpha,\beta)\underset{\gamma\in\Gamma}{\sum}\tLambda^\epsilon(\alpha,\gamma)\Lambda^\epsilon(\gamma,\beta)\mathring{S}_{\gamma}\mathring{T}_{\beta-\alpha-\gamma}=\tLambda^\epsilon(\alpha,\beta)\underset{\gamma\in\Gamma}{\sum}\tLambda^\epsilon(\gamma,\beta-\alpha)\mathring{S}_{\gamma}\mathring{T}_{\beta-\alpha-\gamma}\,,
		\end{align*}
		and we can define  $\mathring{ST}\in{\cal{s}}\big(\Gamma;\MmN\big)$ by the formula:
		\[\nonumber 
		[\mathring{ST}]_\alpha:=\underset{\gamma\in\Gamma}{\sum}\tLambda^\epsilon(\gamma,\alpha)\,\mathring{S}_{\gamma}\cdot \mathring{T}_{\alpha-\gamma}\,.
		\]
	\end{proof}
	
	\begin{proposition}\label{P-4.30}
		There exists some $\epsilon_0>0$ and for any $n\in\mathbb{N}$ there exists some $C_n>0$ such that:
		\beq
		<\gamma>^n\big|\big[\mathring{\mathfrak{m}}^\epsilon_\B[\widehat{H}_\B]\big]_\gamma-\big[\mathring{\mathfrak{m}}_\B[\widehat{H}_\B]_\gamma\big]_{p,q}\big|\,\leq\,C_n\,\epsilon\,,\quad\forall\epsilon\in[0,\epsilon_0]\,.
		\eeq
	\end{proposition}
	\begin{proof}
		Starting from Definition \ref{D-m-eps-B} and  \eqref{F-PN-V-10}--\eqref{F-PN-2} we can write:
		\beq
		\big[\mathring{\mathfrak{m}}^\epsilon_\B[\widehat{H}_\B]_\gamma\big]_{(p,q)}=\big(\mathfrak{T}^\epsilon_{\gamma}\,\psi^\epsilon_p\,,\,\ham^{\epsilon}_\B\,\psi^\epsilon_q\big)_{L^2(\X)}=\big(\mathfrak{T}^\epsilon_{\gamma}\,\psi^\epsilon_p\,,\,H^{\epsilon}\,\psi^\epsilon_q\big)_{L^2(\X)},
		\eeq
		and 
		\beq
		\big[\mathring{\mathfrak{m}}_\B[\widehat{H}_\B]_\gamma\big]_{(p,q)}=\mathfrak{M}_\B[H_\B]_{(\gamma,p),(0,q)}=\big(\tau_{-\gamma}\,\psi_p\,,\,H_\B\,\psi_q\big)_{L^2(X)}\,.
		\eeq
		Then  we have to estimate:
		\begin{align*}
			&<\gamma>^n\big|\big[\mathring{\mathfrak{m}}^\epsilon_\B[\widehat{H}_\B]\big]_\gamma-\big[\mathring{\mathfrak{m}}_\B[\widehat{H}_\B]_\gamma\big]_{p,q}\big|=<\gamma>^n\big|\big(\mathfrak{T}^\epsilon_{\gamma}\,\psi_p^\epsilon\,,\,H^{\epsilon}\,\psi^\epsilon_q\big)_{L^2(\X)}-\big(\tau_{-\gamma}\psi_p\,,\,H_\B\,\psi_q\big)_{L^2(X)}\big|\\ \nonumber
			&\hspace*{12pt}\leq\, <\gamma>^n\big|\big(\Lambda^\epsilon_{\gamma}\,\tau_{-\gamma}\Theta^\epsilon_\B\,\psi_p\,,\,\Op^{\epsilon}(p^\epsilon_\B\sharp^\epsilon\,h\sharp^\epsilon\,p^\epsilon_\B-h_\B)\,\Theta^\epsilon_\B\,\psi_q\big)_{L^2(\X)}\big|+\\
			&\hspace*{24pt}+<\gamma>^n\,\big|\big(\Lambda^\epsilon_{\gamma}\,\tau_{-\gamma}\Theta^\epsilon_\B\,\psi_p\,,\,\Op^{\epsilon}(h_\B)\,\Theta^\epsilon_\B\,\psi_q\big)_{L^2(\X)}-\big(\tau_{-\gamma}\psi_p\,,\,\Op^{A^\circ}(h_\B)\,\psi_q\big)_{L^2(X)}\big|.
		\end{align*}

		First, we notice that Proposition \ref{P-replP3_5}  and Remark \ref{R-p-eps-B} imply that $p^\epsilon_\B\sharp^\epsilon\,h\sharp^\epsilon\,p^\epsilon_\B-h_\B=\mathscr{O}(\epsilon)$ as symbols of bounded operators and the rapid decay of the functions $\psi_p$ for $p\in\underline{n_\B}$ allows to control the factor $<\gamma>^n$ by the decay of $\tau_{-\gamma}\psi_p$. In order to estimate $\Lambda^\epsilon_{\gamma}\,\tau_{-\gamma}\Theta^\epsilon_\B\,\psi_p$ we use Lemma \ref{L-1}. Thus we obtain the estimate:
		\begin{align*}
			<\gamma>^n&\big|\big[\mathring{\mathfrak{m}}^\epsilon_\B[\widehat{H}_\B]\big]_\gamma-\big[\mathring{\mathfrak{m}}_\B[\widehat{H}_\B]_\gamma\big]_{p,q}\big|=<\gamma>^n\big(\tau_{-\gamma}\psi_p\,,\,\Int\big[\big(\Lambda^\epsilon_{\gamma}\,\Lambda^\epsilon-1\big)\mathfrak{K}^\circ(h_\B)\big]\,\psi_q\big)_{L^2(\X)}\hspace*{-0.2cm}+\mathscr{O}(\epsilon).
		\end{align*}
		In order to finish the proof we notice that:
		\begin{align*}
			<\gamma>^n&\big(\tau_{-\gamma}\psi_p\,,\,\Int\big[\big(\Lambda^\epsilon_{\gamma}\,\Lambda^\epsilon-1\big)\mathfrak{K}^\circ(h_\B)\big]\,\psi_q\big)_{L^2(\X)}=\\
			&=<\gamma>^n\int_{\X}dx\int_{\X}dy\,\overline{\psi_p(x-\gamma)}\,\big(i\epsilon\langle\,B^\bullet,(\gamma\wedge x+y\wedge x)\rangle\big)\,\times\\
			&\hspace{3.5cm}\times\,\Big(\int_0^1ds\,\exp\big(-is\epsilon\langle\,B^\bullet,x\wedge\gamma+x\wedge y\rangle\big)\Big)\,\mathfrak{K}^\circ(h_\B)(x,y)\,\psi_q(y).
		\end{align*}
		As $\gamma\wedge x=(\gamma-x)\wedge x$ and $y\wedge x=y\wedge(x-y)$ we obtain that, for any  $M_1, M_2$ and $M_3$ in $\mathbb{N}$, we can write:
		\beq\begin{split}
			\Big|<\gamma>^n&\big(\tau_{-\gamma}\psi_p\,,\,\Int\big[\big(\Lambda^\epsilon_{\gamma}\,\Lambda^\epsilon-1\big)\mathfrak{K}^\circ(h_\B)\big]\,\psi_q\big)_{L^2(\X)}\Big|\leq\,C\,\epsilon\int_{\X}dx\int_{\X}dy\,\times\\
			&\hspace*{12pt}\times\,<x-\gamma>^{M_1}|\psi_p(x-\gamma)|\,<x-y>^{M_2}|\mathfrak{K}^\circ(h_\B)(x,y)|\,<y>^{M_3}|\psi_q(x)|\,\times\\
		&\hspace*{12pt}\times\,<\gamma>^n<x-\gamma>^{-M_1}\big(<x-\gamma><x>+<x-y><y>\big)<x-y>^{-M_2}<y>^{-M_3}.
		\end{split}\eeq
		Finally we notice that $<\gamma>^n\leq\,C_n<\x-\gamma>^n<x>^n$ and $<x>^{n+1}\leq\,C^\prime_{n}<x-y>^{n+1}<y>^{n+1}$ and we only have to choose $M_1\geq n+1$, $M_2\geq n+2$ and $M_3\geq n+2$ in order to obtain a bound of order $\epsilon$ for the above scalar product and finish the proof.
	\end{proof}
	
\subsection{Weak fluctuations on a constant perturbing magnetic field.}\label{SSS-nonconst-mf-Pfr}

It is evident that Hypothesis \ref{H-Bepsc} implies Hypothesis \ref{HF-Beps} and thus all the results obtained in Sections \ref{S-3} and \ref{S-4} in particular Theorem \ref{T-I} remain true, but we shall reformulate them in a way that puts into evidence the presence of the non-vanishing constant part of the perturbing magnetic field. 

Working under Hypothesis \ref{H-Bepsc} on the perturbing magnetic field, the total magnetic field for the magnetic pseudo-differential calculus is $B^\circ+\epsilon\big(B^\bullet+\cc\,B^\epsilon\big)$ and the notations introduced in \eqref{N-A-eps} have to be modified into:
\begin{equation}\begin{split}\label{N-PN-IV-1}
		&\Lambda^{\epsilon,\cc}\,\equiv\,\Lambda^{A^\circ+\epsilon\,A^\cc},\ \Lambda^\epsilon\,\equiv\,\Lambda^{A^\circ+\epsilon{A}^\bullet},\ \tLambda^{\epsilon,\cc}\equiv\Lambda^{\cc\epsilon{A}^\epsilon},\ \tLambda^\epsilon\equiv\Lambda^{\epsilon{A}^\bullet}\\ &\Op^{\epsilon,\cc}\,\equiv\,\Op^{A^\circ+\epsilon{A}^\cc},\ \Op^\epsilon\equiv\Op^{A^\circ+\epsilon{A}^\bullet}\\ 
		&\Omega^{\epsilon,\cc}\,\equiv\,\Omega^{B^\circ+\epsilon\,B^\cc},\ \Omega^\epsilon\,\equiv\,\Omega^{B^\circ+\epsilon{B}^\bullet},\ \tOmega^{\epsilon,\cc}\equiv\Omega^{\cc\epsilon{B}^\epsilon},\\ 
		&\sharp^{\epsilon,\cc}\,\equiv\,\sharp^{B^\circ+\epsilon{B}^\cc},\ \sharp^\epsilon\,\equiv\,\sharp^{B^\circ+\epsilon{B}^\bullet}.
\end{split}\end{equation}

In agreement with our new notations \eqref{N-PN-IV-1} we shall modify the main notations from Sections \ref{S-3} and \ref{S-4} changing the exponents $\epsilon\in[0,\epsilon_0]$ with the pair $(\epsilon,\cc)\in[0,\epsilon_0]\times[0,1]$. 

\begin{remark}
In agreement with the above convention the projection $P^\epsilon_{\B}$ becomes $P^{\epsilon,\cc}_{\B}$ and the reduced Hamiltonian $\ham^\epsilon_{\B}$ becomes $\ham\ec_{\B}$. The important new features appear in the formulation of Statement  2 in Theorem \ref{T-I} and are due to the replacement of the family of functions $\widetilde{\blPsi}^\epsilon_{\B}$ in \eqref{DF-m-W-f} by the family:
 \beq \label{DF-m-W-f-c}
\widetilde{\blPsi}\ec_{\B}\,:=\,\big\{\tpsi\ec_{\alpha,p}:=\widetilde{\Lambda}^{\epsilon,\cc}_\alpha\,\mathfrak{T}^\epsilon_\alpha\,\psi^\epsilon_{p}, \alpha\in\Gamma,\,p\in\underline{n_\B}\big\}.
\eeq
We shall sometimes use the notation $$\psi^\epsilon_{\alpha,p}:=\mathfrak{T}^\epsilon_\alpha\,\psi^\epsilon_{p}\,.$$
An important feature will be the presence of the objects $P^{\epsilon,0}_{\B}$ and $\ham^{\epsilon,0}_{\B}$ referring to the projection and reduced Hamiltonian associated with the constant part of the perturbing magnetic field.
\end{remark}

Related to  \eqref{DF-m-W-f-c} we shall also use the notations:
\[
\mathfrak{T}\ec_\gamma\,:=\,\tLambda\ec_\gamma\, \mathfrak{T}^\epsilon_{\gamma}\,,\qquad\forall\gamma\in\Gamma
\]
where the  $\mathfrak{T}^\epsilon_{\gamma}$ are the Zak translations defined in \eqref{DF-Z-trsl}.
	
	We repeat now the arguments and constructions in Subsection  \ref{SS-m-P-frame}, starting with an operator $\widetilde{Q}^{\epsilon,\cc}_\B$ defined by the regular integral kernel (with rapid off-diagonal decay):
	\[
	\mathfrak{K}[\widetilde{Q}^{\epsilon,\cc}_\B](x,y):=\underset{(\alpha,p)\in\Gamma\times\underline{n_\B}}{\sum}\,\widetilde{\Lambda}^{\epsilon,\cc}_\alpha\,\mathfrak{T}^\epsilon_\alpha\,\psi^\epsilon_{p}(x)\,\overline{\widetilde{\Lambda}^{\epsilon,\cc}_\alpha\,\mathfrak{T}^\epsilon_\alpha\,\psi^\epsilon_{p}(y)}.
	\]
	We notice that:
	\begin{align}\label{F-PN-VII-4}
		\widetilde{\Lambda}^{\epsilon,\cc}_\alpha(x)\,\overline{\widetilde{\Lambda}^{\epsilon,\cc}_\alpha(y)}&=\widetilde{\Lambda}^{\epsilon,\cc}(x,\alpha)\widetilde{\Lambda}^{\epsilon,\cc}(\alpha,y)=\widetilde{\Lambda}^{\epsilon,\cc}(x,y)\big[\exp\big(-i\int_{<x,\alpha,y>}B^\epsilon\big)\big]\\ \nonumber
		&=\widetilde{\Lambda}^{\epsilon,\cc}(x,y)\Big[1\,-\,i\epsilon\cc\,\Big(\int_{<x,\alpha,y>}B^\epsilon\Big)\Big(\int_0^1ds\,\exp\big(-is\epsilon\cc\int_{<x,\alpha,y>}B^\epsilon\big)\Big)\Big],\\ \nonumber
		\Big|\int_{<x,\alpha,y>}B^\epsilon\Big|&\leq\,C<x-\alpha><y-\alpha>,\qquad\forall\epsilon\in[0,\epsilon_0]\,,
	\end{align}
	and conclude that $[\widetilde{Q}^{\epsilon,\cc}_\B]^2=\widetilde{Q}^{\epsilon,\cc}_\B\,+\,\mathscr{O}(\epsilon\cc)$ (as bounded operators on $L^2(\X)$). Thus, for some $\epsilon_0>0$ and any $(\cc,\epsilon)\in[0,1]\times[0,\epsilon_0]$
	$$
	\sigma(\widetilde{Q}^{\epsilon,\cc}_\B)\subset(-\cc\epsilon,\cc\epsilon)\bigcup(1-\cc\epsilon,1+\cc\epsilon)\,,
	$$
	and we may define the following objects similar to those in Subsection  \ref{SS-m-P-frame}:
	\[\begin{split}
		&P^{\epsilon,\cc}_\B\,:=\,-(2\pi i)^{-1}\oint_{\mathscr{C}}d\zz\,\big(\widetilde{Q}^{\epsilon,\cc}_\B-\zz\bb1\big)^{-1},\\
		&\Theta^{\epsilon,\cc}_\B\,:=\,-(2\pi i)^{-1}\oint_{\mathscr{C}}d\zz\,\sqrt{\zz^{-1}}\,\big(\widetilde{Q}^{\epsilon,\cc}_\B-\zz\bb1\big)^{-1},
	\end{split}\]
	so that we have the following formula:
	\[\begin{split}
		\mathfrak{K}[P^{\epsilon,\cc}_\B](x,y)&=\hspace*{-0.3cm}\underset{(\alpha,p)\in\Gamma\times\underline{n_\B}}{\sum}\hspace*{-0.3cm}\Theta^{\epsilon,\cc}_\B\big[\widetilde{\Lambda}^{\epsilon,\cc}_\alpha\,\big(\mathfrak{T}^\epsilon_\alpha\,\psi^\epsilon_{p}\big)\big](x)\,\overline{\Theta^{\epsilon,\cc}_\B\big[\widetilde{\Lambda}^{\epsilon,\cc}_\alpha\,\big(\mathfrak{T}^\epsilon_\alpha\,\psi^\epsilon_{p}\big)\big](y)}\,.
	\end{split}\]
	
Finally we obtain the following Parseval frame of the Hilbert subspace $P^{\epsilon,\cc}_\B\,L^2(\X)$:
	\beq \label{DF-Pfr-B-epsc}
	\blPsi\ec_{\B}:=\big\{\psi^{\epsilon,\cc}_{\alpha,p}:=\Theta^{\epsilon,\cc}_\B\big[\widetilde{\Lambda}^{\epsilon,\cc}_\alpha\,\big(\mathfrak{T}^\epsilon_\alpha\,\psi^\epsilon_{p}\big)\big],\ \forall(\alpha,p)\in\Gamma\times\underline{n_\B}\big\}\,,
	\eeq
	and given any bounded linear operator $T\in\mathbb{B}\big(P^{\epsilon,\cc}_\B\,L^2(\X)\big)$ one may define its \textit{magnetic $\B$-matrix} as the matrix $\mathring{\mathfrak{M}}^{\epsilon,\cc}_\B[T]\in\mathscr{M}_\Gamma[\MmN]$ given by:
	\[
	\big[\mathfrak{M}^{\epsilon,\cc}_\B[T]_{\alpha,\beta}\big]_{p,q}\,:=\,\big(\psi^{\epsilon,\cc}_{\alpha,p}\,,\,T\,\psi^{\epsilon,\cc}_{\beta,q}\big)_{L^2(\X)}.
	\]
	
It is obvious from the above formulas (compared with the similar ones in Subsection ~\ref{SS-m-P-frame}) that all the estimations obtained in Subsection  \ref{SS-m-P-frame} for the operators $\widetilde{Q}^\epsilon_{\B},P^\epsilon_{\B},\Theta^\epsilon_{\B},\ham^\epsilon_{\B}$ remain true for $\widetilde{Q}\ec_{\B},P\ec_{\B},\Theta\ec_{\B},\ham\ec_{\B}$ with the change of $\epsilon$ into $\cc\epsilon$. The important difference comes at the end of Subsection \ref{SS-PN-V-2}, when, due to the presence of the functions $\mathfrak{T}^\epsilon_\alpha\,\psi^\epsilon_{p}$ in \eqref{DF-m-W-f-c} that replace the functions $\tau_{-\gamma}\psi_p$ in \eqref{DF-m-W-f}, Formula \eqref{F-PN-V-11} becomes:
\begin{align}\label{F-PN-V-11-c}
	&\big[\mathfrak{M}\ec_{\B}[\ham\ec_{\B}]_{\alpha,\beta}\big]_{p,q}=\big(\mathfrak{T}^{\epsilon}_\alpha\psi_p\,,\,\big[\Int\,\tLambda\ec\tLambda^{\epsilon}\mathfrak{K}^\circ(h)\big]\mathfrak{T}^{\epsilon}_\beta\, \psi_q\big)_{L^2(\X)}=\\ \nonumber
	&\ =\tLambda\ec\tLambda^{\epsilon}(\alpha,\beta)\Big[\big[\mathring{\mathfrak{m}}^\epsilon_{\B}[\widehat{H}]_{\alpha-\beta}\big ]_{p,q}+\cc\epsilon\big[\widetilde{\mathfrak{M}}\ec_{\B}[\ham\ec_{\B}]_{\alpha,\beta}\big]_{p,q}\Big]\,,
\end{align}
where, for $\mathfrak{M}^{\epsilon,0}_{\B}[\ham^{\epsilon,0}_{\B}]_{\alpha,\beta}=\tLambda^{\epsilon}(\alpha,\beta)\mathring{\mathfrak{m}}^\epsilon_{\B}[\widehat{H}]_{\alpha-\beta}$,  we may use the general result in \eqref{F-PN-V-11}:
\begin{align}
	&\mathfrak{M}^{\epsilon,0}_{\B}[\ham^{\epsilon,0}_{\B}]_{\alpha,\beta}=\tLambda^{\epsilon}(\alpha,\beta)\big[\mathring{\mathfrak{m}}_{\B}[H_{\B}]_{\alpha-\beta}+\epsilon\,\widetilde{\mathfrak{M}}^{\epsilon}_{\B}[\mathfrak{K}^\circ(h)]_{\alpha,\beta}\big]
\end{align}
with the particular choice $B=\epsilon{B}^\bullet$ and also the estimation in Proposition \ref{P-4.30}. Theorem~\ref{T-II} clearly follows from the above results.

	From \eqref{F-PN-VII-4} we also deduce that:
	\[
	\mathfrak{K}[\widetilde{Q}^{\epsilon,\cc}_\B]=\widetilde{\Lambda}^{\epsilon,\cc}\,\big(\mathfrak{K}[P^{\epsilon,0}_\B]\,+\,\epsilon\cc\mathfrak{K}^{\epsilon,\cc}_2\big)\,,
	\]
	where  $P^{\epsilon,0}_{\B}$ is the same as $P^\epsilon_{\B}$ in \eqref{F-P-eps-B-const} in Proposition \ref{P-P-eps-B-const} and $\mathfrak{K}^{\epsilon,\cc}_2$ is a $BC^\infty(\X\times\X)$  integral kernel having rapid off-diagonal decay uniformly in $(\cc,\epsilon)\in[0,1]\times[0,\epsilon_0]$. 

\begin{appendices}

\section{Brief reminder of the Bloch-Floquet Theory}\label{A-BF-Theory}

On $L^2(\X)$ we consider the canonical unitary representation of $\Rd\cong\X$ with the convention (used for the Weyl calculus) $\big(U(z)f\big)(x):=f(x-z)$. 

We are interested in analyzing the structure of self-adjoint, lower-bounded operators of the form $T:\mathcal{D}(T)\rightarrow{L}^2(\X)$ (with $\mathcal{D}(T)$ a dense linear subspace of $L^2(\X)$) that have the property:  
\begin{hypothesis}\label{H-GcomOp}
For any $\gamma\in\Gamma\cong\Zd$ the unitary $U(\gamma)$ leaves $\mathcal{D}(T)$ invariant and commutes with $T$.
	\end{hypothesis} 
More precisely, we shall be interested in operators satisfying Hypothesis \ref{H-GcomOp} and being of the form $T=\overline{\Op^{A^\circ}(F)}$ as in Definition \ref{D-Hcirc}. In \cite{CHP-5} we have discussed the extension of the Bloch-Floquet theory from usual periodic differential operators, to this class of magnetic pseudo-differential operators and thus we shall very briefly recall the main results insisting a little bit more on the vector-bundle structure associated with  the Bloch-Floquet decomposition.

Let us recall the Bloch-Floquet transformation:
\beq\label{DF-UBF}
\forall(\theta,x)\in\T_*\times\X,\quad\big(\U_{BF}\varphi\big)(\theta,x):=\underset{\gamma\in\Gamma}{\sum}\,e^{-i<\s_*(\theta),\gamma>}\,\varphi(x+\gamma),\quad\forall\varphi\in\mathscr{S}(\X),
\eeq
where the series is absolutely convergent due to the rapid decay of $\varphi\in\mathscr{S}(X)$. Noticing that for any $\varphi\in\mathscr{S}(X)$:
\beq
\big(\U_{BF}\varphi\big)(\theta,x+\alpha)=e^{i<\s_*(\theta),\alpha>}\big(\U_{BF}\varphi\big)(\theta,x),\quad\forall(\theta,x)\in\T_*\times\X,\ \forall\alpha\in\Gamma\,,
\eeq
we define the linear space:
\beq
\forall\theta\in\T_*:\qquad\mathscr{F}_\theta:=\big\{F\in{L}^2_{\text{\tt loc}}(\X),\ \tau_{\alpha}F=\chi_\theta(-\alpha)F,\ \forall\alpha\in\Gamma\}\,,
\eeq
where each element is completely determined by its restriction to the unit cell $\E\subset\X$ and endow it with the scalar product:
\beq\label{DF-tsp}
\big(F,G\big)_{\theta}:=\big(F\big|_{\E},G\big|_{\E}\big)_{L^2(\E)}.
\eeq
\begin{remark}
If we consider the family of smooth functions $\mathfrak{f}_{\xi}(x):=e^{i<\xi,x>}$ indexed by $\xi\in\X^*$ we notice that multiplication with $\mathfrak{f}_{\s_*(\theta)}$ defines a bijective linear map $\mathscr{F}_\theta\overset{\sim}{\longrightarrow}\mathscr{F}_{1}$ where $\mathscr{F}_1$ is clearly the space of square integrable $\Gamma$-periodic functions on $\X$.
\end{remark}
\begin{remark}
Given any function $f\in\mathscr{F}_1$ we can associate with  it the function $\mathring{f}:=f\circ\s:\T\rightarrow\Co$ and notice that it defines a unitary map $\mathscr{F}_1\overset{\sim}{\longrightarrow}L^2(\T)$ and also a bijection between smooth functions. Its inverse is clearly $\mathring{f}\mapsto\mathring{f}\circ\p$.
\end{remark}
It is rather evident that each $\mathscr{F}_{\theta}$ with \eqref{DF-tsp} is a Hilbert space and the two remarks above together with the constructions and arguments in \cite{Dix} allow us to define the measurable vector fields $\big\{\mathfrak{f}_{\s_*(\theta)}^{-1}\mathring{f},\ \mathring{f}\in{L}^2(\T)\big\}_{\theta\in\T_*}$ and the associated direct integral Hilbert space:
	\beq\label{DF-dir-prod}
	\mathscr{F}\,:=\,\int^{\oplus}_{\T_*}d\theta\,\mathscr{F}_\theta\,.
	\eeq
One can easily prove that the map $\U_{BF}$ in \eqref{DF-UBF} defines a unitary map $L^2(\X)\overset{\sim}{\longrightarrow}\mathscr{F}$. An important result of the Bloch-Floquet theory is that, for any operator $T$ satisfying Hypothesis \ref{H-GcomOp},  its Bloch-Floquet (BF)  representation $\U_{BF}T\U_{BF}^{-1}$ is a decomposable operator on $\mathscr{F}$ with respect to the direct integral structure \eqref{DF-dir-prod}, i.e.:
\beq
\U_{BF}T\U_{BF}^{-1}\,=\,\int^{\oplus}_{\T_*}d\theta\,\widehat{T}(\theta).
\eeq

Let us also consider the Bloch-Floquet-Zak (BFZ) transformation defined as \linebreak ${[\U_{BFZ}\varphi](\xi,x):=\mathfrak{f}_{-\xi}(x)[\U_{BF}\varphi](\p_*(\xi),x)}$, or explicitely:
\beq\label{DF-UBFZ}
\forall(\xi,x)\in\X^*\times\X,\quad\big(\U_{BFZ}\varphi\big)(\xi,x):=\underset{\gamma\in\Gamma}{\sum}\,e^{-i<\xi,x+\gamma>}\,\varphi(x+\gamma),\quad\forall\varphi\in\mathscr{S}(\X).
\eeq
One easily notice that $[\U_{BFZ}\varphi](\xi,x)=[\U_{BFZ}\varphi]\big(\xi,(\s\circ\p)x\big)$ and we can define:
\beq\label{DF-cUBFZ}
\forall(\xi,\omega)\in\X^*\times\T,\quad\big(\mathring{\U}_{BFZ}\varphi\big)(\xi,\omega):=[\U_{BFZ}\varphi]\big(\xi,\s(\omega)\big),\varphi(x+\gamma),\quad\forall\varphi\in\mathscr{S}(\X).
\eeq
\beq\label{DF-G}
\mathscr{G}\,:=\,\big\{F\in{L}^2_{\text{\tt loc}}\big(\X^*;L^2(\T)\big),\ \tau_{\gamma^*}F=\mathfrak{f}_{-\gamma^*}F\big\}
\eeq
and endow $\mathscr{G}$ with the scalar product:
\beq
\big(F,G\big)_{\mathscr{G}}\,:=\,\int_{{\cal{B}}}d\hat{\xi}\,\big(F(\hat{\xi}),G(\hat{\xi})\big)_{L^2(\T)}.
\eeq

\begin{definition}\label{D-reg-sp}
	\begin{itemize}
		\item For any $p\in\R$ we denote by:
		\begin{itemize}
			\item $\mathscr{F}^p_{\bz^*}:=\mathscr{F}_{\bz^*}\cap\mathscr{H}^p_{\text{\tt loc}}(\X),\quad\mathscr{F}^p:=\int^{\oplus}_{\T_*}d\theta\,\mathscr{F}^p_{\theta}$;
			\item $\mathscr{G}^p:=\big\{F\in{L}^2_{\text{\tt loc}}\big(\X^*;\mathscr{H}^p(\T)\big),\ \tau_{\gamma^*}F=\mathfrak{f}_{-\gamma^*}F\big\}$;
			\item $\mathscr{G}^p_\infty:=\big\{F\in{C}^\infty\big(\X^*;\mathscr{H}^p(\T)\big),\ \tau_{\gamma^*}F=\mathfrak{f}_{-\gamma^*}F\big\}$;
		\end{itemize}
		\item $\mathscr{F}^\infty_{\bz^*}:=\mathscr{F}_{\bz^*}\cap{C}^\infty(\X)$, $\mathscr{F}^\infty:=\int^{\oplus}_{\T_*}d\theta\,\mathscr{F}^\infty_{\theta}$;
		\item $\mathscr{F}^\infty_\infty:=\big\{F\in{C}^\infty\big(\T_*;C^\infty(\X)\big),\ F(\theta,x+\gamma)=\chi_{\theta}(-\gamma)F(\theta,x),\ \forall\gamma\in\Gamma\big\}$;
		\item $\mathscr{G}^\infty:=\big\{F\in{L}^2_{\text{\tt loc}}\big(\X^*;C^\infty(\T)\big),\ \tau_{\gamma^*}F=\mathfrak{f}_{-\gamma^*}F\big\}$;
		\item  $\mathscr{G}^\infty_\infty:=\big\{F\in{C}^\infty\big(\X^*;C^\infty(\T)\big),\ \tau_{\gamma^*}F=\mathfrak{f}_{-\gamma^*}F\big\}$.
	\end{itemize}
\end{definition}

If we denote by $\mathfrak{f}\in{C}^\infty(\X^*\times\X)$ the function $\mathfrak{f}(\xi,x):=e^{i<\xi,x>}$ we notice the evident equalities:
\beq\label{F-GF}
\mathscr{G}^p=\overline{\mathfrak{f}}\mathscr{F}^p,\quad\mathscr{G}^\infty=\overline{\mathfrak{f}}\mathscr{F}^\infty,\quad\mathscr{G}^p_\infty=\overline{\mathfrak{f}}\mathscr{F}^p_\infty,\quad\quad\mathscr{G}^\infty_\infty=\overline{\mathfrak{f}}\mathscr{F}^\infty_\infty.
\eeq

Using this framework, in Subsection 2.3 of \cite{CHP-5} we extend a number of results in \cite{Ku1, Ku} from periodic differential operators to our class of pseudo-differential operators and recall them here in the following theorem:
\begin{theorem}[Bloch-Floquet decomposition]\label{T-FBdec} If $H^{\circ}$ is the magnetic pseudo-differential operator defined in Definition \ref{D-Hcirc}, then it is a lower semi-bounded self-adjoint operator in $L^2(\X)$ with domain $\mathscr{H}^p(\X)$ ($p>0$), that commutes  with all the translations $\big\{U_\gamma,\,\gamma\in\Gamma\big\}$ and satisfies the following properties: 
	\begin{enumerate}
		\item  $H^\circ$  is unitarily equivalent with a direct integral (in the sense of \cite{Dix}). More precisely, in the BF representation \eqref{DF-UBF}, we have the decomposition:
		\beq\nonumber 
		\U_{BF}H^\circ \U_{BF}^{-1}=\int^\oplus_{\T_*}\hspace*{-5pt}d\theta\,\widehat{H^\circ}(\theta) 
		\eeq
		where each $\widehat{H^\circ}(\theta)$ is a lower semi-bounded self-adjoint operator having as domain $\mathcal{D}(\widehat{H^\circ})=\mathscr{F}^p_{\theta}$.
		\item For any $\zz\in\mathbb{C}\setminus\sigma(H^\circ)$, the map $$\X_*\ni\xi\mapsto\widetilde{R}^\circ_\zz(\xi):=\mathfrak{f}_{\xi}^{-1}\big(\widehat{H}^\circ\big(\p_*(\xi)\big)-\zz\bb1\big)^{-1}\mathfrak{f}_{\xi}\in\mathbb{B}\big(L^2(\T)\big)$$
		 is well-defined and smooth.
		\item For each $\theta\in\T_*$ and $\zz\in\mathbb{C}\setminus\sigma(H^\circ)$, the resolvent $\widehat{R}^\circ_\zz(\theta)$ is a compact operator. Thus for any $\theta\in\T_*$, the spectrum of $\widehat{H^\circ}(\theta)$ consists of discrete eigenvalues:  $$\sigma\big(\widehat{H^\circ}(\theta)\big)=\big\{\widetilde{\lambda}_k(\theta)\big\}_{k\in\mathbb{N}_\bullet}\,,$$ which satisfy  $-\infty<\widetilde{\lambda}_k(\theta)<\widetilde{\lambda}_{k+1}(\theta)$ for any $k\in\mathbb{N}_\bullet$. 
  
		\item If the geometric multiplicity of $\widetilde{\lambda}_k(\theta)$ is $N_k(\theta)$, then we can re-index the spectrum by using "simple" eigenvalues in the following way: 
  $$\widetilde{\lambda}_1(\theta)=\lambda_1(\theta)=\lambda_2(\theta)=\dots=\lambda_{N_1} (\theta)<\widetilde{\lambda}_2(\theta)=\lambda_{N_1+1}(\theta)=\dots=\lambda_{N_1+N_2}(\theta)<\dots $$
  In this way, for each $k\in\mathbb{N}_\bullet$, the function $\T_*\ni\theta\mapsto\lambda_k(\theta)\in\mathbb{R}$ is continuous. 
	\end{enumerate}
\end{theorem} 
\paragraph{The Bloch eigenprojections.}
In order to define the eigenprojections $\widehat{\pi}_k(\theta)$ associated with the Bloch eigenvalues $\big\{\lambda_k(\theta)\big\}_{k\in\mathbb{N}_\bullet}$, we apply the following procedure. To each $k\in \mathbb{N}_\bullet$ and $\theta\in \mathbb T^d_*$, we introduce the minimal labelling of $\lambda_k(\theta)$:  
$$\nu (k,\theta)=\inf\big\{j\in\mathbb{N}_\bullet\,,\,\lambda_j(\theta*) =\lambda_k(\theta)\big\}\,,$$
and  define the eigenprojections by the following formulas:
\beq\label{DF-pi-k}
\begin{split}
	\widehat{\pi}_k(\theta):=&\left\{
	\begin{array}{ll}
		=-\frac{1}{2\pi i}\oint_{\mathscr{C}_k(\theta)}d\z\,\big(\widehat{H^\circ}(\theta)-\z\bb1\big)^{-1}\in  {\mathbb{L}(\mathscr{F}_{
				\theta})} & \mbox{ if } \nu(k,\theta)=k\,,\\
		=0 & \mbox{ if }  \nu(k,\theta) < k\,,
	\end{array}
	\right.
\end{split}
\eeq 
where $\mathscr{C}_k(\theta)\subset\mathbb{C}$ is a circle surrounding $\lambda_k(\theta)$ in a anticlockwise direction and no other point from $\sigma\big(\widehat{H^\circ}(\theta)\big)$. They define measurable $\mathbb{B}(\mathscr{F}_{\theta})$-valued functions on $\mathbb{T}_*$. {Some standard arguments using the ellipticity of $h\in S^p_1(\X^*\times\X)$ imply the following statement.
\begin{theorem}\label{T-FBdec-proj}
Under the hypothesis of Theorem \ref{T-FBdec} the range of any finite dimensional orthogonal projection $\widehat{\pi}_k(\theta)$ is contained in the space of regular functions $\mathscr{F}_{\theta}^\infty$.
\end{theorem}

\section{Some properties of the magnetic pseudo-differential calculus}\label{A-m-PsiDO}

\subsection{Main definitions and properties}

In \cite{MP-1} the authors prove that starting from the basic principle that the influence of magnetic fields on Hamiltonian systems may be described by the so-called \textit{'minimal coupling procedure'} that consists in modifying the momentum variables by a translation with the magnetic vector potential, one may define a 'twisted' Weyl quantization providing the algebra of quantum observables of the given Hamiltonian system in magnetic fields. The procedure is developed in \cite{IMP-1, IMP-3, MPR1, AMP, CP-1, CP-2} and we have developed it in connection with periodic Hamiltonians in magnetic fields in (\cite{IP}, \cite{CIP}, \cite{CHP-1} - \cite{CHP-5}.)
\begin{definition}\label{D-OpA}
	Given a regular magnetic field $B\in\Fb^2(\X)$ with vector potential $A\in\Fp^1(\X)$ such that $B=dA$, we have the following explicit formula for the magnetic quantization  for any $\Phi\in\mathscr{S}(\Xi)$: 
	\beq\begin{split}\label{F-OpA}
		&\big(\Op^A(\Phi)\phi\big)(x)\,=\,\int_{\X}dy\int_{\X^*}d\eta\,e^{-i\int_{[x,y]}A}\,e^{i<\eta,x-y>}\,\Phi\big((x+y)/2,\eta\big)\,\phi(y),\\ 
		&\hspace*{10cm}\forall\phi\in\mathscr{S}(\X),\,\forall x\in\X.
	\end{split}\eeq
\end{definition}
\begin{remark}\label{R-KerOpA}
	The explicit formula of $\Op^A(\Phi)\in\mathcal{L}\big(\mathscr{S}(\X);\mathscr{S}(\X)\big)$ in Definition \ref{D-OpA} allows us to obtain a simple formula for the integral kernel $\mathfrak{K}^A[\Phi]\in\mathscr{S}(\X\times\X)$ of this operator:
\beq\label{F-KerOpA}
\mathfrak{K}^A[\Phi]\,=\,\Lambda^A\,\big[\big((\bb1\otimes\mathcal{F}_{\X^*})\Phi\big)\circ\Upsilon\big].
\eeq 
where $\Upsilon:\X\times\X\ni(x,y)\mapsto\big((x+y)/2,x-y\big)\in\X\times\X$ is a linear bijection with Jacobian 1.
	\begin{remark}\label{R-OpA}
		It is proven in \cite{MP-1} that for a vector potential $A$ with components of class $C^\infty_{\text{\tt pol}}(\Rd;\R)$ the map $\Op^A:\mathscr{S}(\Xi)\rightarrow\mathcal{L}\big(\mathscr{S}(\X);\mathscr{S}(\X)\big)$ defines a linear and topological isomorphism $\Op^A:\mathscr{S}(\Xi)\overset{\sim}{\rightarrow}\mathcal{L}\big(\mathscr{S}^\prime(\X);\mathscr{S}(\X)\big)$ that has a canonical extension to a linear and topological isomorphism $\Op^A:\mathscr{S}^\prime(\Xi)\overset{\sim}{\rightarrow}\mathcal{L}\big(\mathscr{S}(\X);\mathscr{S}^\prime(\X)\big)$.
	\end{remark}
\end{remark}
\begin{definition}\label{D-MoyalPr}
	Due to the Remark \ref{R-OpA}, for any magnetic field with components of class $BC^\infty(\Rd)$ we can define the following bilinear map $\sharp^B:\mathscr{S}(\Xi)\times\mathscr{S}(\Xi)\rightarrow\mathscr{S}(\Xi)$ given by the equality: $\Op^A(\Phi)\Op^A(\Psi):=\Op^A(\Phi\sharp^B\Psi)$ and called \textbf{the magnetic Moyal product}.
\end{definition}
\begin{remark}\label{R-MoyalPr}
	The results in \cite{MP-1} imply that the magnetic Moyal product is a bilinear jointly continuous map that has the following explicit formula:
	\beq\label{F-pr-Moyal}
	\big(\Phi\sharp^B\Psi\big)(X)\,=\,4^{-d}\int_{\Xi}dZ\int_{\Xi}dZ^\prime\,e^{-2i(<\xi-\zeta,x-z^\prime>-<\xi-\zeta^\prime,x-z>)}\,e^{-i\Theta^B(x;z,z^\prime)}\,\Phi(Z)\,\Psi(Z^\prime)
	\eeq
	where $\Theta^B(x;z,z^\prime)$ is the flux of the $2$-form $B$ through the following triangle in $\X$:
	\beq
	\mathcal{T}_{x;z,z^\prime}:=\big\{P_{x;z,z^\prime}(s,s^\prime):=z+z^\prime-x+2s(x-z)+2s^\prime(z-z^\prime)\,,\,s\in[0,1],\,s^\prime\in[0,s]\big\}.
	\eeq
\end{remark}

We shall constantly use the following properties of the magnetic pseudodifferential calculus, considering a magnetic field $B$ with components of class $BC^\infty(\Rd)$ and a choice of a vector potential $A$ with components of class $C^\infty_{\text{\tt pol}}(\Rd)$:

\begin{theorem}
	[Theorem 2.1 in \cite{IMP-3}.] The Moyal product extends to a bilinear, continuous map
	(see also Theorem 2.6 in \cite{IMP-1})
	$$
	S^{p_1}_\rho(\X^*,\X)\times S^{p_2}_\rho(\X^*,\X)\ni(F,G)\,\mapsto\,F\sharp^BG\in S^{p_1+p_2}_\rho(\X^*,\X).
	$$
\end{theorem}
\begin{proposition}
	[Proposition 4.23 in \cite{MP-1}.] The operator $\Op^A:\mathscr{S}^\prime(\Xi)\overset{\sim}{\rightarrow}\mathcal{L}\big(\mathscr{S}(\X);\mathscr{S}^\prime(\X)\big)$ restricts to a linear continuous map $\Op^A:S^p_\rho(\X^*,\X)\rightarrow\mathcal{L}\big(\mathscr{S}(\X);\mathscr{S}(\X)\big)\bigcap\mathcal{L}\big(\mathscr{S}^\prime(\X);\mathscr{S}^\prime(\X)\big)$ for any $(p,\rho)\in\R\times[0,1]$.
\end{proposition}
\begin{theorem} 
	[Theorem 2.6 in \cite{IMP-3} and Theorem 3.1 in \cite{IMP-1}.] If $f\in S^0_\rho(\X^*,\X)$, with any $\rho\in[0,1]$, then $\Op^A(f)\in\mathbb{B}\big(L^2(\X)\big)$ .
\end{theorem}

The next theorem is proved in \cite{IMP-3} (Theorem 2.6)  (see also Theorem 3.1 in \cite{IMP-1}):
\begin{theorem}
	Given $\rho \in (0,1)$, a magnetic field $B\in\Fb^2(\X)$ and an associated vector potential of class $\Fp^1(\X)$, if $f\in S^0_\rho(\Xi)$,  then $\Op^A(f)\in\mathbb{B}\big(L^2(\X)\big)$.
\end{theorem}
We also need the following results proven in \cite{IMP-3} (Theorem 2.7 and Proposition 2.4):
\begin{theorem}
	For $F\in S^p_1(\Xi)$ positive and elliptic, with $p>0$, the operator $\Op^A(F):\mathscr{S}(\X)\rightarrow L^2(\X)$ is essentially self-adjoint and its closure has the domain:
	\[
	\mathscr{H}^p(\X)\,:=\,\big\{f\in L^2(\X)\,,\,\Op(\mathfrak{m}_p)f\in L^2(\X)\big\}.
	\]
	Moreover, for any $\zz\in\mathbb{C}\setminus\sigma\big(\overline{\Op^A(F)}\big)$ there exists $\mathfrak{r}^B_\zz\in S^{-p}_1(\Xi)$ such that: 
	\beq\label{F-simb-rez}
	\big(\overline{\Op^A(F)}-\zz\bb1\big)^{-1}=\Op^A(\mathfrak{r}^B_\zz).
	\eeq
\end{theorem}

Using Propositions 1.3.3 and 1.3.6 in \cite{ABG} and Lemma A.4 in \cite{MPR1}, one obtains easily the following result:
\begin{proposition}\label{P-ker-OpA} Suppose we have a magnetic field $B\in\Fb^2(\X)$ with a vector potential $A\in\Fp^1(\X)$.
	\begin{enumerate}
		\item If $F\in S^p_1(\Xi)$ for any $(p,\rho)\in\R\times\{0,1\}$, then $\mathfrak{K}^A[F]\in\mathscr{S}^\prime(\X\times\X)$ as given in \eqref{F-KerOpA} is a smooth function on $\X\times\X\setminus\big\{(x,x)\in\X\times\X,\,x\in\X\big\}$ having rapid decay in the variable $x-y\in\X$.
		\item If $F\in S^p_1(\Xi)$ with $p<0$, then $\mathfrak{K}^A[F]\in L^1(\X\times\X)$.
		\item $F\in S^{-\infty}(\Xi)$ if and only if $\mathfrak{K}^A[F]\in C^\infty(\X\times\X)$ with rapid decay in the directions orthogonal to the diagonal of $\X\times\X$, together with all its derivatives.
	\end{enumerate}
\end{proposition}

\subsection{Some estimations in weak regular magnetic fields}

Let us formulate {a technical result which - in our current setting -} replaces Proposition~3.5 in \cite{CIP} and which can be proved with similar arguments.

\begin{proposition}\label{P-replP3_5}
	Given $(m,s)\in\R\times\R$, there exists $\epsilon_0>0$ such that for any $\epsilon\in[0,\epsilon_0]$ there exists a  continuous bilinear map  
	$$ S^m_1(\Xi)_\Gamma\times S^s_1(\Xi)_\Gamma \ni (F,G) \mapsto \mathcal{r}_{\epsilon}\big(F,G)\in S^{m+s-2}_1(\Xi)$$ 
	such that for $B$ as in \eqref{HF-Beps} and and $B^\circ$ as in Hypothesis \eqref{H-BGamma} and using notation \eqref{N-A-eps}:
	\beq\nonumber 
	F\sharp^{\epsilon}G\,=\,F\sharp^{\circ}G\,+\,\epsilon\,\mathcal{r}_{\epsilon}\big(F,G)
	\eeq
	{Moreover, the bilinear map $\mathcal{r}_{\epsilon}$ is uniformly bounded with respect to $\epsilon \in[0,\epsilon_0]$.}
\end{proposition}

\begin{proposition}\label{P-1_8-CIP} Suppose given $p\in\R$ and a real symbol $F\in{S}^p_1(\Xi)$ such that $\overline{\Op^\epsilon(F)}$ is a self-adjoint operator in $L^2(\X)$ for any $\epsilon\in[0,\epsilon_1]$ with some $\epsilon_1>0$.
	For $K\subset\big[\mathbb{C}\setminus\sigma\big(\overline{\Op^\circ(F)}\big)\big]$ with $B^\circ=dA^\circ$ as in Hypothesis \ref{H-BGamma}, there exists $\epsilon_0\in(0,\epsilon_1]$ such that:
	\begin{enumerate}
		\item  $K\subset\big[\mathbb{C}\setminus\sigma\big(\overline{\Op^\epsilon(F)}\big)\big]$ for any $\epsilon\in[0,\epsilon_0]$.
		\item We consider the magnetic symbols defined by $\Op^\epsilon\big(\mathfrak{r}^{\epsilon}_{F,\zz}\big)=\big(\overline{\Op^\circ(F)}-\zz\bb1\big)^{-1}$; then $\mathfrak{r}^{\epsilon}_{F,\zz}\in{S}^{-p}_1(\Xi)$ for $p\geq0$ and $\mathfrak{r}^{\epsilon}_{F,\zz}-1\in{S}^{p}_1(\Xi)$ for $p<0$ and the map $K\ni\zz\mapsto\mathfrak{r}^{\epsilon}_{F,\zz}\in S^{p}_1(\Xi)$, is continuous for the Fr\'{e}chet topology on $S^{p}_1(\Xi)$, uniformly in $(\epsilon,\zz)\in[0,\epsilon_0]\times{K}$.
		\item We have the equality:
		\beq \nonumber 
		\mathfrak{r}^{\epsilon}_{F,\zz}\,=\,(1+\epsilon\,\mathcal{r}_{\epsilon}(F,\mathfrak{r}^{\circ}_{F,\zz})\big)^-_{\epsilon}\,{\sharp^\epsilon}\,\mathfrak{r}^{\circ}_{F,\zz}.
		\eeq
	\end{enumerate}
\end{proposition}
\begin{proof}
	We shall  only present the small changes  to be done in the proof of Proposition 1.8 in \cite{CIP}. In fact, the main point is to replace Equation (3.22) in \cite{CIP} by the  equality:
	\beq\nonumber 
	1=\big(F-\zz\big)\sharp^{\circ}\mathfrak{r}^{\circ}_{F,\zz}=\big(F-\zz\big)\sharp^{\epsilon}\mathfrak{r}^{\circ}_{F,\zz}\,-\,\epsilon\,\mathcal{r}_{\epsilon}\big(F,\mathfrak{r}^{\circ}_{F,\zz}\big)
	\eeq 
	with $\mathcal{r}_{\epsilon}\big(F,\mathfrak{r}^{\circ}_{F,\zz}\big)$ in $ S^{-2}_1(\Xi)$ uniformly for $\epsilon\in[0,\epsilon_0]$. Once we made these replacements in \cite{CIP} the arguments in the proof of Proposition 1.8 of \cite{CIP} remain  true.
\end{proof}

\subsection{Some formulas concerning the magnetic 'phase functions'}

\subsubsection{Some estimations for $\tOmega^\epsilon$.}\label{SSS-PN-VI-2}

For $B:=\epsilon\,B^\epsilon$ as in \eqref{HF-Beps} we have to estimate several times scalar products of the following two forms:
\beq\left\{
\begin{array}{l}[\mathcal{I}^\epsilon_1(\phi,\psi)]_{\alpha,\beta}:=\int_{\X}dx\int_{\X}dy\ \tOmega^{\epsilon}(x,\alpha,y)\,\overline{\phi(x-\alpha)}\,\psi(y-\beta)\,,\\ 
{[\mathcal{I}^\epsilon_2(\phi,\psi)]_{\alpha,\beta}}:=\int_{\X}dx\int_{\X}dy\ \tOmega^{\epsilon}(\alpha,x,\beta)\,\overline{\phi(x-\alpha)}\,\psi(y-\beta)\,,
\end{array}\right.\ \forall(\phi,\psi)\in[\mathscr{S}(\X)]^2\,,
\eeq
taking into account that:
\beq
\tOmega^{\epsilon}(x,y,z):=\exp\Big(-i\epsilon\int_{<x,y,z>}B^\epsilon\Big)=1-i\epsilon\Big(\int_{<x,y,z>}B\Big)\int_0^1ds\exp\Big(-is\epsilon\int_{<x,y,z>}B\Big)\,,
\eeq
where $\int_{<x,y,z>}B$ denotes the integral of the 2-form $B$ on the triangle with vertices $\{x,y,z\}$.
We have to consider the following magnetic flux integrals:
\begin{align}
\int_{<\alpha,x,y>}\hspace*{-0.4cm}B^\epsilon\,&=\hspace*{-0.2cm}\underset{1\leq j<k\leq d}{\sum}\hspace*{-0.2cm}\Phi_{\alpha}[B^\epsilon]_{j,k}(x,y)\,(x-\alpha)_j(y-\alpha)_k \\ \nonumber
	&=\hspace*{-0.2cm}\underset{1\leq j<k\leq d}{\sum}\hspace*{-0.2cm}\Phi_{\alpha}[B^\epsilon]_{j,k}(x,y)\,(x-\alpha)_j(y-\beta)_k+\hspace*{-0.3cm}\underset{1\leq j<k\leq d}{\sum}\hspace*{-0.2cm}\Phi_{\alpha}[B^\epsilon]_{j,k}(x,y)\,(x-\alpha)_j(\beta-\alpha)_k\,,\\ \nonumber \int_{<\alpha,\beta,y>}\hspace*{-0.4cm}B^\epsilon\,&=\hspace*{-0.2cm}\underset{1\leq j<k\leq d}{\sum}\hspace*{-0.2cm}\Phi_{\alpha,\beta}[B^\epsilon]_{j,k}(x,y)\,(y-\alpha)_j(\beta-\alpha)_k=\hspace*{-0.3cm}\underset{1\leq j<k\leq d}{\sum}\hspace*{-0.2cm}\Phi_{\alpha,\beta}[B^\epsilon]_{j,k}(x,y)\,(y-\beta)_j(\beta-\alpha)_k\,,
\end{align}
for which we have used the notation:
\begin{align}
	&\Phi_{\alpha}[B^\epsilon]_{j,k}(x,y):=\int_0^1ds\int_0^1sdu\,B^\epsilon_{j,k}\big(\alpha+s(x-\alpha)+us(y-x)\big),\\
	&\Phi_{\alpha,\beta}[B^\epsilon]_{j,k}(y):=\int_0^1ds\int_0^1sdu\,B^\epsilon_{j,k}\big(\alpha+s(\beta-\alpha)+us(y-\beta)\big).
\end{align}
Finally, we obtain the following formulas:
\begin{align}
[\mathcal{I}^\epsilon_1(\phi,\psi)]_{\alpha,\beta}:&=\int_{\X}dx\int_{\X}dy\ \tOmega^{\epsilon}(x,\alpha,y)\,\overline{\phi(x-\alpha)}\,\psi(y-\beta)\\ \nonumber & =\big(\tau_\alpha\phi\,,\,\tau_\beta\psi\big)_{L^2(\X)}\\ \nonumber & \quad 
+i\epsilon\int_{\X}dx\int_{\X}dy\hspace*{-0.2cm}\underset{1\leq j<k\leq d}{\sum}\hspace*{-0.2cm}\Phi_{\alpha}[B^\epsilon]_{j,k}(x,y)\,\big((x-\alpha)_j\overline{\phi(x-\alpha)}\big)\,\big((y-\beta)_k\psi(y-\beta)\big)\\ \nonumber
&\quad+i\epsilon\int_{\X}dx\int_{\X}dy\hspace*{-0.2cm}\underset{1\leq j<k\leq d}{\sum}\hspace*{-0.2cm}\Phi_{\alpha}[B^\epsilon]_{j,k}(x,y)\,\big((x-\alpha)_j\overline{\phi(x-\alpha)}\big)\,\big((\alpha-\beta)_k\psi(y-\beta)\big)\\ \nonumber
&\quad+\mathscr{O}(\epsilon^2)\,,
	\end{align}
    \begin{align}
[\mathcal{I}^\epsilon_2(\phi,\psi)]_{\alpha,\beta}:&=\int_{\X}dx\int_{\X}dy\ \tOmega^{\epsilon}(\alpha,x,\beta)\,\overline{\phi(x-\alpha)}\,\psi(y-\beta)=\big(\tau_\alpha\phi\,,\,\tau_\beta\psi\big)_{L^2(\X)}+\\ \nonumber
&\quad-i\epsilon\int_{\X}dx\int_{\X}dy\hspace*{-0.2cm}\underset{1\leq j<k\leq d}{\sum}\hspace*{-0.2cm}\Phi_{\alpha}[B^\epsilon]_{j,k}(x,y)\,\big((\alpha-\beta)_j\overline{\phi(x-\alpha)}\big)\,\big((y-\beta)_k\psi(y-\beta)\big)\\ \nonumber
&\quad+\mathscr{O}(\epsilon^2)\,,
\end{align}
and estimations of the form (with $a=1,2$) for any $N\in\mathbb{N}$ {with a constant $C_N(\BB)>0$ (like in Notation \ref{N-B}) independent of $\epsilon\in[0,\epsilon_0]$}:
\beq\begin{split}
&\big|[\mathcal{I}^\epsilon_a(\phi,\psi)]_{\alpha,\beta}-\big(\tau_\alpha\phi\,,\,\tau_\beta\psi\big)_{L^2(\X)}\big|\,\leq\\
&\hspace*{3cm}\leq\,{C_N(\BB)} \,<\alpha-\beta>^{-N}\,\big[\epsilon\,\lnu_{N+d+2,0}(\phi)\,\lnu_{N+d+2,0}(\psi)\,+\,\mathscr{O}(\epsilon^2)\big]\,.
\end{split}\eeq

We shall use the notation $Q_j$ for the operator of multiplication with the variable $x_j$ in $\mathscr{S}(\X)$.
Using these computations in \eqref{F-PN-V-12} we obtain that:
\begin{align}
	&\big[\mathfrak{m}^{\epsilon}_{\B}[\mathfrak{K}[h^\circ]]_{\alpha,\beta}\big]_{p,q}=
	\end{align}
	\begin{align*}
	&\ =-i\,\Big\langle\,\mathfrak{K}[h^\circ]\,,\,(\tau_{-\alpha}\otimes\tau_{-\beta})\underset{1\leq j<k\leq d}{\sum}\,\Phi_{\alpha}[B^\epsilon]_{j,k}\big(\overline{Q_j\psi_p}\otimes{Q}_k\psi_q\big)\Big\rangle_{\mathscr{S}(\X\times\X)}\\
	&\quad -i\, \Big\langle\,\mathfrak{K}[h^\circ]\,,\,(\tau_{-\alpha}\otimes\tau_{-\beta})\underset{1\leq j<k\leq d}{\sum}\,\Phi_{\alpha}[B^\epsilon]_{j,k}(\beta-\alpha)_k\big(\overline{Q_j\psi_p}\otimes\psi_q\big)\Big\rangle_{\mathscr{S}(\X\times\X)}\\ \nonumber
	&\quad -i\, \Big\langle\,\mathfrak{K}[h^\circ]\,,\,(\tau_{-\alpha}\otimes\tau_{-\beta})\underset{1\leq j<k\leq d}{\sum}\,\Phi_{\alpha,\beta}[B^\epsilon]_{j,k}(\beta-\alpha)_k\big(\overline{\psi_p}\otimes{Q}_j\psi_q\big)\Big\rangle_{\mathscr{S}(\X\times\X)}
	\end{align*}
	\begin{align}\label{F-PN-VI-1}
	&\ =
	\Big\langle\,\,\mathfrak{K}[h^\circ]\,,\,(\tau_{-\alpha}\otimes\tau_{-\beta})\underset{1\leq j<k\leq d}{\sum}\,\Big[\Phi_{\alpha}[B^\epsilon]_{j,k}\Big(Q_j\otimes{Q}_k+(\beta-\alpha)_k(Q_j\otimes\bb1)\Big)\\ \nonumber
	&\qquad \qquad +\Phi_{\alpha,\beta}[B^\epsilon]_{j,k}(\beta-\alpha)_k(\bb1\otimes{Q}_j)\Big]\big(\overline{\psi_p}\otimes\psi_q\big)\Big\rangle_{\mathscr{S}(\X\times\X)}.
\end{align}

\subsubsection{Derivatives of $\Omega^B$}\label{SSS-PN-VI-1}

Let us consider the magnetic flux integral:
\[
\int_{<\alpha,x,y>}\hspace*{-0.4cm}B^{\epsilon}=\underset{1\leq j,k\leq d}{\sum}(x-\alpha)_j(y-x)_k\int_0^1ds\int_0^sdr\,[B{\epsilon}]_{j,k}\big(\alpha+s(x-\alpha)+r(y-x)\big)
\] 
and compute the derivative:
\begin{align*}
	\partial_{x_\ell}\int_{<\alpha,x,y>} & B^{\epsilon}=\underset{1\leq k\leq d}{\sum}(y-x)_k\int_0^1ds\int_0^sdr\,B^{\epsilon}_{\ell,k}\big(\alpha+s(x-\alpha)+r(y-x)\big)\\
	&+\underset{1\leq k\leq d}{\sum}(x-\alpha)_k\int_0^1ds\int_0^sdr\,B^{\epsilon}_{k,\ell}\big(\alpha+s(x-\alpha)+r(y-x)\big)\\
	&+\underset{1\leq j,k\leq d}{\sum}(x-\alpha)_j(y-x)_k\int_0^1ds\int_0^sdr\,(s-r)\big(\partial_{x_\ell}B^{\epsilon}_{j,k}\big)\big(\alpha+s(x-\alpha)+r(y-x)\big).
\end{align*}

\subsubsection{Estimations for \eqref{F-D-Omega}.}
Let us consider the derivatives in \eqref{F-D-Omega}
\begin{align}\label{F-D-Omega-1}
	\big(\partial^c_v\partial_z^{a}\tOmega^\epsilon\big)(z+v/2,\gamma,z-v/2)=\partial^c_v\partial_z^{a}\Big[\exp\Big(-i\epsilon\int_{<z+v/2,\gamma,z-v/2>}B^\epsilon\Big)\Big].
\end{align}

First we notice that:
\begin{align}
&\partial_t^p\exp\big(-i\epsilon{F}(t)\big)=\\ \nonumber
&\quad=\underset{1\leq N\leq p}{\sum}\,(-i\epsilon)^N \underset{p_1+\cdots+p_N=p}{\sum}\,C_p(p_1,\ldots,p_N)\,\underset{1\leq \ell \leq N}{\prod}\big[\big(\partial_t^{p_\ell}F\big)(t)\big]\,\Big[\exp\big(-i\epsilon{F}(t)\big)\Big]\,,
\end{align}
and thus:
\begin{align}
&\partial^a_v\partial_z^{b}\exp\big(-i\epsilon{F}(z,v)\big)=\\ \nonumber
&\quad=\underset{1\leq N\leq |a|+|b|}{\sum}\,(-i\epsilon)^N \underset{|a^{(1)}|+\cdots+|a^{(N)}|=|a|}{\sum}\,C_a(a^{(1)},\ldots,a^{(N)})\,\underset{|b^{(1)}|+\cdots+|b^{(N)}|=|b|}{\sum}\,C_b(b^{(1)},\ldots,b^{(N)})\,\times\\ \nonumber
&\qquad\times\,\underset{1\leq \ell \leq N}{\prod}\big[\big(\partial^{a^{(\ell)}}_v\partial_z^{b^{(\ell)}}F\big)(z,v)\big]\,\Big[\exp\big(-i\epsilon{F}(z,v)\big)\Big]\,.
\end{align}
Finally let us compute $\int_{<z+v/2,\gamma,z-v/2>}B^\epsilon$ as the flux of the 2-form $B^\epsilon\in\Fb^2(\X)$ through the triangle $<z+v/2,\gamma,z-v/2>=<z-v/2,z+v/2,\gamma>\subset\X$ that we parametrize starting with the vertex $z-v/2$ and defined by the oriented pair of egdes: $z+v/2-(z-v/2)=v$ and $\gamma-(z+v/2)$, i.e.:
\begin{align}\nonumber
\int_{<z+v/2,\gamma,z-v/2>}\hspace*{-1.5cm}B^\epsilon\hspace*{0.8cm}&=\underset{1\leq j<k\leq d}{\sum}v_j(\gamma-z-v/2)_k\int_0^1\hspace*{-0.3cm}ds\int_0^1\hspace*{-0.3cm}s\,d\spr\,B_{j,k}\big(z-v/2+sv+\spr(\gamma-z-v/2)\big)\\
&=\underset{1\leq j<k\leq d}{\sum}v_j(\gamma-z)_k\int_0^1\hspace*{-0.3cm}ds\int_0^1\hspace*{-0.3cm}s\,d\spr\,B_{j,k}\big(z-v/2+sv+\spr(\gamma-z-v/2)\big)
\end{align}
and its derivatives $\partial^a_v\partial_z^{b}\int_{<z+v/2,\gamma,z-v/2>}B^\epsilon$ will be finite sums (with at most $|a|+|b|$ terms) of terms of the form:
\[\begin{split}
&v_j(\gamma-z)_k\int_0^1\hspace*{-0.3cm}ds\int_0^1\hspace*{-0.3cm}s\,d\spr\,\partial^a_v\partial_z^{b}B_{j,k}\big(z-v/2+sv+\spr(\gamma-z-v/2)\big),\\
&v_j\int_0^1\hspace*{-0.3cm}ds\int_0^1\hspace*{-0.3cm}s\,d\spr\,\partial^a_v\partial_z^{b}B_{j,k}\big(z-v/2+sv+\spr(\gamma-z-v/2)\big),\\
&(\gamma-z)_k\int_0^1\hspace*{-0.3cm}ds\int_0^1\hspace*{-0.3cm}s\,d\spr\,\partial^a_v\partial_z^{b}B_{j,k}\big(z-v/2+sv+\spr(\gamma-z-v/2)\big),\\
&\int_0^1\hspace*{-0.3cm}ds\int_0^1\hspace*{-0.3cm}s\,d\spr\,\partial^a_v\partial_z^{b}B_{j,k}\big(z-v/2+sv+\spr(\gamma-z-v/2)\big).
\end{split}\]

The previous arguments imply that some $C(\BB)>0$ (like in Notation \ref{N-B}) exists, independent of $\epsilon\in[0,\epsilon_0]$ such that the following estimation holds:
\begin{equation}\label{C-PN-IV-1}
\big|\big(\partial^c_v\partial_z^{a}\tOmega^\epsilon\big)(z+v/2,\gamma,z-v/2)\big|\,\leq\,C\big(\epsilon\,|v|\,|z-\gamma|\,+\,\epsilon^2\big[|v|\,|z-\gamma|\big]^{|a|+|b|},\quad\forall\epsilon\in[0,\epsilon_0].
\end{equation}

\section{Some basic properties of frames on  Hilbert spaces}\label{A-frames}
In this appendix,  $\H$ denotes  a complex Hilbert space.

\begin{definition}\label{D-frame}~
	\begin{enumerate}
		\item We call  frame of $\H$, an at most countable family  of vectors $\big\{\psi_p\,,\,p\in\mathbb{N}_\bullet\big\}\subset\H$ such that there exists two positive constants $0<A\leq B$ satisfying:
		\beq\label{DF-fr}
		A\big\|f\|_{\H}^2\,\leq\,\underset{M\nearrow\infty}{\lim}\underset{1\leq p\leq M}{\sum}\big|\big(\psi_p\,,\,f\big)_{\H}\big|^2\,\leq\,B\big\|f\|_{\H}^2,\quad\forall f\in\H.
		\eeq
		\item A  frame  of $\H$  is called a Parseval frame, when the following equality is verified for any $f\in\H$:
		\beq\label{DF-Pfr}
		\big\|f\big\|^2_{\H}=\underset{p\in\mathbb{N}_\bullet}{\sum}\,\big|\big(\psi_p\,,\,f\big)_{\H}\big|^2\,,
		\eeq
		with the series converging in $\ell^2(\mathbb{N}_\bullet)$.
		\item Given a frame $\vec{\psi}\big\{\psi_p\,,\,p\in\mathbb{N}_\bullet\big\}\subset\H$,  we introduce its  coordinate map $\mathfrak{C}$ as: 
		\beq\label{DF-fr-1}
		\H\ni f\mapsto\big (\mathfrak{C}_{\vec{\psi}} f ) :=\{\big(\psi_p\,,\,f\big)_{\H}\big\}_{p\in \Nb} \in\ell^2(\Nb).
		\eeq
	\end{enumerate}
\end{definition}
\begin{lemma}\label{L-Parseval}
	Given a Parseval frame in $\H$, the following identities are valid, for $(f,g)\in\H\times\H$:
	\beq\label{F-P-1}
	\big(f\,,\,g\big)_{\H}=\underset{p\in\mathbb{N}_\bullet}{\sum}\,\big(f,\psi_p\big)_{\H}\big(\psi_p,g\big)_{\H},\qquad
	f=\underset{p\in\mathbb{N}_\bullet}{\sum}\,\big(\psi_p\,,\,f\big)_{\H}\,\psi_p\,,
	\eeq
	with the series converging in $\ell^2(\mathbb{N}_\bullet)$ and resp. in the strong topology of $\H$ and the coordinate map $\mathfrak{C}_{\vec{\psi}}:\H\rightarrow\ell^2(\Nb)$ is an isometry that may not be surjective.
\end{lemma}
\begin{proof}
	Let us write the polarization identity in $\H$:
	\[\nonumber 
	\big(f\,,\,g\big)_{\H}=(1/4)\underset{0\leq k\leq 3}{\sum}i^k\big\|f+(-i)^kg\big\|^2_{\H}\,,
	\]
	and use \eqref{DF-Pfr} to obtain:
	\begin{align*}
		\big(f\,,\,g\big)_{\H}&=(1/4)\underset{0\leq k\leq 3}{\sum}i^k\big\|f+(-i)^kg\big\|^2_{\H}=(1/4)\underset{0\leq k\leq 3}{\sum}i^k\underset{p\in\mathbb{N}_\bullet}{\sum}\,\big|\big(\psi_p\,,\,f+(-i)^kg\big)_{\H}\big|^2\\
		&=\underset{p\in\mathbb{N}_\bullet}{\sum}\,(2/4)\Big(2\Re\big(\psi_p\,,\,f\big)_{\H}\overline{\big(\psi_p\,,\,g\big)}_{\H}-2\Im\big(\psi_p\,,\,f\big)_{\H}\overline{\big(\psi_p\,,\,g\big)}_{\H}\Big)\\
		&=\underset{p\in\mathbb{N}_\bullet}{\sum}\big(\psi_p\,,\,g\big)_{\H}\overline{\big(\psi_p\,,\,f\big)}_{\H}.
	\end{align*}
	For the second identity one proves using the first identity that $\sum\limits_{p=1}^M\,\big(\psi_p\,,\,f\big)_{\H}\psi_p$ is a sequence converging weakly to $f$ and being Cauchy for the strong topology, due to the Parseval condition.
\end{proof}

\begin{proposition}\label{C-wlim-id}
If given a pair of vectors $(v,w)\in\H\times\H$ we denote by $|v><w|$ the linear rank 1 operator $|v> <w|u:=(w,u)_{\H}v$, the following limit exists in the weak operator topology:
$$
\underset{p\nearrow\infty}{\lim}|\psi_p><\psi_p|\,=\,\bb1.
$$
\end{proposition}

We also recall the following theorem in \cite{Chris} (Theorem 5.5.1.)
\begin{theorem} A sequence $\vec{f}\equiv\{f_k\}_{k=1}^\infty$ in $\H$ is a frame for $\H$ if and only if the application:
	\[ \nonumber 
	\vec{c}\equiv\{c_k\}_{k=1}^\infty \mapsto\mathfrak{S}_{\vec{f}}\big(\vec{c}\big)\,:=\,\underset{k\in\mathbb{N}_\bullet}{\sum}c_\ell \,f_\ell 
	\]
	is a well defined surjective map from $\ell^2(\mathbb{N}_\bullet)$ onto $\H$.
\end{theorem}
\begin{remark}\label{R-A-01}
	If we denote by $\{\mathcal{e}_\gamma\}_{\gamma\in\Gamma}$ the canonical orthonormal basis of $\ell^2(\Nb)$ and use the notation introduced in the beginning of Subsection \ref{ss1.3} we have the identities:
	\[\nonumber 
	\mathfrak{C}_{\vec{\psi}}=\underset{p\in\Nb}{\sum}\mathcal{e}_\p\bowtie\psi_p,\qquad\Id_{\H}=\mathfrak{C}_{\vec{\psi}}^*\mathfrak{C}_{\vec{\psi}},\qquad\mathfrak{C}_{\vec{\psi}}\mathfrak{C}_{\vec{\psi}}^*=:\mathfrak{P}_{\vec{\psi}}
	\]
	with $\mathfrak{P}_{\vec{\psi}}\in\mathbb{B}\big(\ell^2(\Nb)\big)$ the orthogonal projection on the range of $\mathfrak{C}_{\vec{\psi}}$.
\end{remark}

Let us consider the map:
\beq\label{DF-coord-op-hom}
\mathbb{B}\big(\H\big)\ni\,T\,\mapsto\,\mW_{\vec{\psi}}[T]:=\mathfrak{C}_{\vec{\psi}}\,T\,\mathfrak{C}_{\vec{\psi}}^*\in\mathbb{B}\big(\ell^2(\Nb)\big).
\eeq
We can associate with  $\mathfrak{W}_{\vec{\psi}}[T]$ the infinite matrix of the operator $\mathfrak{C}_{\vec{\psi}}\,T\,\mathfrak{C}_{\vec{\psi}}^*\in\mathbb{B}\big(\ell^2(\Nb)\big)$ with respect to the canonical orthonormal basis $\{\mathcal{e}_p\}_{p\in\Nb}$.

\begin{proposition}
The map $\mW_{\vec{\psi}}:\mathbb{B}\big(\H\big)\,\rightarrow\,\mathbb{B}\big(\ell^2(\Nb)\big)$ in \eqref{DF-coord-op-hom} is an injective $C^*$-algebras morphism.
\end{proposition}}
\end{appendices}
\vspace{0.5cm}

\noindent{\bf Acknowledgements.} HC acknowledges support from Grant DFF–5281-00046B of the Independent Research Fund Denmark $|$ Natural Sciences. RP acknowledges support from a grant of the Romanian Ministry of Research, Innovation and Digitization, CNCS -UEFISCDI, project number PN-IV-P1-PCE-2023-0264, within PNCDI IV. BH and RP acknowledge support from the CNRS International Research Network ECO-Math. We also thank our home institutions for hosting our reciprocal visits.

\end{document}